%% file: main_ST.tex
\documentclass[11pt]{article}

\usepackage{tikz}
\usetikzlibrary{backgrounds,positioning,fit,patterns,shadows,calc}

\usepackage{epsfig}
\usepackage{amsfonts}
\usepackage{amssymb}
\usepackage{amstext}
\usepackage{amsmath}

\usepackage{calligra}
\usepackage[T1]{fontenc}

\usepackage{amsthm, thmtools}

\usepackage{paralist}

\usepackage{nameref}
\definecolor{ForestGreen}{rgb}{0.1333,0.5451,0.1333}
\definecolor{DarkRed}{rgb}{0.8,0,0}
\definecolor{Red}{rgb}{1,0,0}
\usepackage[linktocpage=true,
pagebackref=true,colorlinks,
linkcolor=DarkRed,citecolor=ForestGreen,
bookmarks,bookmarksopen,bookmarksnumbered]
{hyperref}
\usepackage{cleveref}

\usepackage{thm-restate} 

\usepackage{xspace}
\usepackage{color}
\usepackage{enumitem}
\usepackage{comment}
\usepackage{caption}
\usepackage{subcaption}

\usepackage[noend]{algorithmic}
\usepackage[section,boxed]{algorithm}

\usepackage{array}
\usepackage{multirow}
\usepackage{tablefootnote}

\usepackage{wrapfig}

\makeatletter
\def\thmt@refnamewithcomma #1#2#3,#4,#5\@nil{%
	\@xa\def\csname\thmt@envname #1utorefname\endcsname{#3}%
	\ifcsname #2refname\endcsname
	\csname #2refname\expandafter\endcsname\expandafter{\thmt@envname}{#3}{#4}%
	\fi
}
\makeatother

\declaretheorem[numberwithin=section,refname={Theorem,Theorems},Refname={Theorem,Theorems},name={Theorem}]{thm}
\declaretheorem[numberlike=thm,refname={Lemma,Lemmas},Refname={Lemma,Lemmas},name={Lemma}]{lem}
\declaretheorem[numberlike=thm,refname={Corollary,Corollaries},Refname={Corollary,Corollaries},name={Corollary}]{cor}
\declaretheorem[numberlike=thm,refname={Fact,Facts},Refname={Fact,Facts},name={Fact}]{fact}

\declaretheorem[numberlike=thm,refname={Proposition,Propositions},Refname={Proposition,Propositions},name={Proposition}]{prop}

\declaretheorem[numberlike=thm,refname={Example,Examples},Refname={Example,Examples}]{example}
\declaretheorem[numberlike=thm,refname={Definition,Definitions},Refname={Definition,Definitions},name={Definition}]{defn}
\declaretheorem[style=remark,numberwithin=section,refname={Remark,Remarks},Refname={Remark,Remarks},name={Remark}]{rem}
\declaretheorem[style=remark,numberlike=thm,refname={Claim,Claims},Refname={Claim,Claims}]{claim}

\newcommand{\squishlist}{
	\begin{list}{$\bullet$}
		{ \setlength{\itemsep}{0pt}
			\setlength{\parsep}{2pt}
			\setlength{\topsep}{2pt}
			\setlength{\partopsep}{0pt}
			\setlength{\leftmargin}{1.5em}
			\setlength{\labelwidth}{1em}
			\setlength{\labelsep}{0.5em} } }
	\newcommand{\squishend}{
\end{list}  }

\def\*#1*\ {}

\ifdefined\ShowComment

\def\danupon#1{\marginpar{$\leftarrow$\fbox{D}}\footnote{$\Rightarrow$~{\sf #1 --Danupon}}}
\def\monika#1{\marginpar{$\leftarrow$\fbox{M}}\footnote{$\Rightarrow$~{\sf #1 --Monika}}}
\def\sebastian#1{\marginpar{$\leftarrow$\fbox{S}}\footnote{$\Rightarrow$~{\sf #1 --Sebastian}}}
\def\thatchaphol#1{\marginpar{$\leftarrow$\fbox{T}}\footnote{$\Rightarrow$~{\sf #1 --Thatchaphol}}}

\else

\def\danupon#1{}
\def\monika#1{}
\def\sebastian#1{}
\def\thatchaphol#1{}

\fi

\newboolean{short}
\setboolean{short}{true} 

\newcommand{\shortOnly}[1]{\ifthenelse{\boolean{short}}{#1}{}}
\newcommand{\longOnly}[1]{\ifthenelse{\boolean{short}}{}{#1}}

\usepackage[left=1in,top=1in,right=1in,bottom=1in]{geometry} 
\usepackage{booktabs}
\usepackage{threeparttable}

\makeatletter
\renewcommand{\paragraph}{%
	\@startsection{paragraph}{4}%
	{\z@}{1ex \@plus 1ex \@minus .2ex}{-1em}%
	{\normalfont\normalsize\bfseries}%
}
\makeatother

\newcommand{\patrascu}{P{\v a}tra{\c s}cu\xspace}

\newcommand{\po}{\hat{o}} 

\AtBeginDocument{%
	\let\ref\Cref
}

\begin{document}
	\global\long\def\sst{\textsf{SST}}
	\newcommand{\cA}{{\cal A}}
	\newcommand{\cB}{{\cal B}}
	\newcommand{\cC}{{\cal C}}
	\newcommand{\cD}{{\cal D}}
	\global\long\def\cE{{\cal E}}
	\global\long\def\opt{\textsf{OPT}}
	\global\long\def\decomp{\textsf{decomp}}
	
	\global\long\def\cN{{\cal N}}
	\global\long\def\cO{{\cal O}}
	\global\long\def\cL{{\cal L}}
	\global\long\def\cP{{\cal P}}
	\global\long\def\cG{{\cal G}}
	\global\long\def\cA{{\cal A}}
	\global\long\def\cB{{\cal B}}
	\global\long\def\cC{{\cal C}}
	\global\long\def\cT{{\cal T}}
	\global\long\def\cS{{\cal {\cal S}}}
	\global\long\def\cR{{\cal R}}
	\global\long\def\cF{{\cal F}}
	\global\long\def\rtensor{\otimes_{R}}
	\global\long\def\one{\mathbf{1}}

	\newcommand{\mst}{{\sf MSF}\xspace} 
	\newcommand{\st}{{\sf SF}\xspace}

	\global\long\def\instree{\textsf{insert\_tree}}
	\global\long\def\insnon{\textsf{insert\_nontree}}
	\global\long\def\deltree{\textsf{delete\_tree}}
	\global\long\def\delnon{\textsf{delete\_nontree}}
	\global\long\def\cutset#1{#1\textsf{-cutset}}
	\global\long\def\cutsize#1{#1\textsf{-cutsize}}
	\global\long\def\polylog{\mbox{polylog}}
	\global\long\def\poly{\mbox{poly}}
	\global\long\def\net{\mbox{\textsf{NET}}}
	
	\newcommand*\samethanks[1][\value{footnote}]{\footnotemark[#1]}

	\title{Dynamic Spanning Forest with Worst-Case Update Time:\\ Adaptive, Las Vegas, and $O(n^{1/2-\epsilon})$-Time}
	\author{
		Danupon Nanongkai\\{KTH Royal Institute of Technology}
		\and
		Thatchaphol Saranurak\\ {KTH Royal Institute of Technology}
	}
	\date{}
	\pagenumbering{roman}
	\maketitle
	\input{abstract}
	\pagebreak{}
	
	\tableofcontents{}
	
	\pagebreak{}
	\pagenumbering{arabic}

\input{intro.tex}

	\input{overview}

\input{prelim.tex}

\input{2dET.tex}

\input{sketches_det.tex}

\input{reduction.tex}

\input{decomposition_new.tex}

\input{local_decomposition.tex}

\input{monte_carlo.tex}

\input{las_vegas.tex}

	\paragraph{Acknowledgement.} The authors would like to thank Parinya Chalermsook for explaining them the cut-matching game framework. 
	
	This project has received funding from the European Research Council (ERC) under the European Union's Horizon 2020 research and innovation programme under grant agreement No 715672. The authors were also partially supported by the Swedish Research Council (Reg. No. 2015-04659.)
	
	\bibliographystyle{plain}
	\bibliography{references}

	\pagebreak{}
	
	\appendix

\input{formalize-correctness-det_sketch.tex}

\input{omit.tex}

\input{warm-up.tex}

\end{document}

%% file: abstract.tex
\begin{abstract}
We present two algorithms for dynamically maintaining a spanning forest of a graph undergoing edge insertions and deletions. Our algorithms guarantee {\em worst-case update time} and work against an {\em adaptive adversary}, meaning that an edge update can depend on previous outputs of the algorithms. 
We provide the {\em first polynomial improvement} over the long-standing $O(\sqrt{n})$ bound of [Frederickson~STOC'83, Eppstein, Galil, Italiano and Nissenzweig~FOCS'92] for such type of algorithms. 
The previously best improvement was $O(\sqrt{n (\log\log n)^2/\log n})$ [Kejlberg-Rasmussen, Kopelowitz, Pettie and Thorup~ESA'16].
We note however that these bounds were obtained by deterministic algorithms while our algorithms are randomized.

Our first algorithm is Monte Carlo and guarantees an $O(n^{0.4+o(1)})$ worst-case update time, where the $o(1)$ term hides the $O(\sqrt{\log\log n/\log n})$  factor. Our second algorithm is Las Vegas and guarantees an $O(n^{0.49306})$ worst-case update time with high probability.
Algorithms with better update time either needed to assume that the adversary is oblivious (e.g. [Kapron, King and Mountjoy~SODA'13]) or can only guarantee an amortized update time.
Our second result answers an open problem by Kapron~et~al.
To the best of our knowledge, our algorithms are among a few non-trivial randomized dynamic algorithms that work against adaptive adversaries.

The key to our results is a decomposition of graphs into subgraphs that either have high expansion or sparse. This decomposition serves as an interface between recent developments on (static) flow computation and many old ideas in dynamic graph algorithms: On the one hand, we can combine previous dynamic graph techniques to get faster dynamic spanning forest algorithms if such decomposition is given. On the other hand, we can adapt flow-related techniques (e.g. those from [Khandekar, Rao and Vazirani~STOC'06], [Peng~SODA'16], and [Orecchia and Zhu~SODA'14]) to maintain such decomposition. To the best of our knowledge, this is the first time these flow techniques are used in fully dynamic graph algorithms.

\end{abstract}

%% file: intro.tex
\section{Introduction}\label{sec:intro}

In the {\em dynamic spanning forest (\st)} problem, we want to maintain a spanning forest $F$ of an undirected unweighted graph $G$ undergoing edge insertions and deletions. In particular, we want to construct a data structure that supports the following operations. 
\begin{itemize}
	\item {\sc Preprocess($G$)}: Initialize the data structure with an input graph $G$. After this operation, the data structure outputs a spanning forest $F$ of $G$.
	\item {\sc Insert($u, v$)}: Insert edge $(u, v)$ to $G$. After this operation, the data structure outputs changes to $F$,  if any.
	\item {\sc Delete($u, v$)}: Delete edge $(u, v)$ from $G$. After this operation, the data structure outputs changes to $F$,  if any.
\end{itemize}
The goal is to minimize the {\em update time}, i.e. the time needed to execute the insert and delete operations. The time bound is in terms of $n$ and $m$, denoting respectively the number of nodes and edges in $G$ at the time the update happens. We use $\tilde{O}$ to hide $\poly\log(n)$ terms.

The update time is usually categorized into two types: For any $t$, an algorithm is said to have an {\em amortized} update time of $t$ if, for any $k$, the total time it spends to process the first $k$ updates (edge insertions/deletions) is at most $kt$. It is said to have a {\em worst-case} update time of $t$ if it spends at most $t$ time for {\em each} update operation. (In case of randomized algorithms, the guarantee may hold with high probability or in expectation.) Thus, roughly speaking an algorithm with a small amortized update time is fast ``on average'' but may take a long time to respond to a single update. In contrast, the worst-case update time is more preferable since it guarantees to hold for every operation. (We also note that the amortized update time is less desired because most algorithms that guarantee it have to assume that the input graph has no edge initially.)

The dynamic \st problem along with its closely related variations -- {\em dynamic connectivity} and {\em dynamic minimum spanning forest} (\mst) -- played a central role in the study of dynamic graph algorithms.
The first result for these problems dates back to Frederickson's deterministic algorithm from 1985~\cite{Frederickson85}, which provides a worst-case $O(\sqrt{m})$ update time. This bound, when combined with the general sparsification technique of Eppstein~et~al.~\cite{EppsteinGIN97} from 1992, implies a worst-case $O(\sqrt{n})$ update time.  
By allowing the update time to be amortized, this bound was significantly improved by Henzinger and King \cite{HenzingerK99} in 1995, who presented a Las Vegas randomized algorithm with polylogarithmic expected amortized update time. In the following decade, this result was refined in many ways, including algorithms with smaller update time which almost matches existing lower bounds (the gap is currently $O((\log\log n)^2)$; see, e.g., \cite{HuangHKP-SODA17,Thorup00,HenzingerT97,PatrascuD06}), deterministic algorithms (e.g. \cite{HolmLT01,Wulff-Nilsen13a}), and algorithms that works for harder problems such as two-edge connectivity and \mst (e.g. \cite{HolmLT01,HenzingerK01}).

Given that the problem is fairly well-understood from the perspective of amortized update time, many researchers have turned their attention back to the worst-case update time in the last decade (one sign of this trend is the 2007 work of \patrascu and Thorup \cite{PatrascuT07}). The question on worst-case update time is not only interesting in the context of the dynamic \st problem.
This is because for many problems (e.g. \mst, two-edge connectivity, shortest paths~\cite{Thorup05} and matching \cite{BhattacharyaHN17soda}) we have many techniques to argue about the amortized update time but do not understand the case of worst-case update time much.  
In the context of dynamic \st, the $O(\sqrt{n})$ worst-case update time of \cite{Frederickson85,EppsteinGIN97} has remained the best for decades until Kapron, King and Mountjoy~\cite{KapronKM13} showed a breakthrough polylogarithmic bound in 2013 (the bound was originally $O(\log^5 n)$ in \cite{KapronKM13} and was later improved to $O(\log^4 n)$ \cite{GibbKKT15}). 

We however still cannot claim a victory over this problem, because the algorithms of \cite{KapronKM13,GibbKKT15} still need to be refined in many ways. Besides the possibility to improve the bound further, one big challenge arose from the fact that they
are randomized {\em Monte Carlo}, meaning that they may make mistakes (with a small probability). 
While the ultimate goal would be to get a deterministic algorithm with a similar worst-case update time that works for harder problems (such as two-edge connectivity and \mst), the next challenges are to obtain (i) a {\em Las Vegas} algorithm  (thus making no errors) and (ii) an algorithm that works against {\em adaptive adversaries} (any deterministic algorithm will guarantee both of these).
Getting {\em Las Vegas} algorithms is a well-known and interesting question in many areas of theoretical computer science. Since the issue about adaptive adversaries might be quite specific to the area of dynamic algorithms, we discuss this issue briefly here. 

An adaptive adversary is the one who can determine an update based on previous outputs of the algorithm (e.g. the maintained spanning forests). It is used to contrast the weaker notion of {\em oblivious adversaries} who must fixed edge updates in advance and thus updates are not influenced by algorithms' outputs; for example, dynamic \st algorithms that work under the oblivious adversary assumption cannot reveal the maintained forest to an adversary. (Note that both types of adversaries do not have an access to the randomness used by the algorithm to make random choices\footnote{In fact, one can categorize adversaries further based on how they can access the randomness. For example, we can define an {\em randomness-oblivious} adversary to be the one that never see the random bits and an  {\em randomness-adaptive} adversary to be the one that see the random bits {\em after} they are used by the algorithm. We are not aware of algorithms that need the distinction between the two cases, so we do not discuss them here. We note that our algorithms work against the stronger notion of randomness-adaptive adversaries. Also note that one can define an even stronger notion of adversary who can see the entire random string from the beginning. Similar to the case of online algorithms \cite{Ben-DavidBKTW90}, it can be argue that randomization does not help against this type of adversary; i.e. if there is a randomized algortihm that works against such adversary, then there is also a deterministic algorithm with the same performance.}.)
One drawback of the algorithms of \cite{KapronKM13,GibbKKT15} is that they work only under this assumption. This drawback in fact reflects our lack of understanding in exploiting randomness for a wide range of dynamic graph problems: without this assumption {\em very few} algorithms are able take advantage of randomness\footnote{To the best of our knowledge, the only randomized dynamic graph algorithm that works against an adaptive adversary is that by \cite{HenzingerK99} and the follow-up improvement by \cite{Thorup00} and \cite{Kejlberg-Rasmussen16} which are also against adaptive adversary for the same reason. This excludes obvious cases where problem outputs are unique and thus there is no difference between adaptive and oblivious adversary for the problem.}
It is a fundamental question whether the {\em true source of power} of randomized dynamic algorithms is the randomness itself or in fact the oblivious adversary assumption. 
We note that this question is also {\em very important in some applications}.
%
For example, removing the oblivious adversary assumption from the randomized algorithms of Henzinger~et~al.~\cite{HenzingerKNFOCS14} and Roditty-Zwick~\cite{RodittyZ04} for the decremental single-source and all-pairs shortest-paths problems will lead to new fast static algorithms for computing maximum $s$-$t$ and multicommodity flows (see, e.g., \cite{BernsteinC-STOC16,Madry10} for further discussions). 
This is the main drive behind the current efforts in {\em derandomizing} the algorithms of \cite{HenzingerKNFOCS14,RodittyZ04} (e.g., \cite{BernsteinC-STOC16,BernsteinC-SODA17,HenzingerKN13}). (Note that deterministic algorithms always work against adaptive adversaries.)
%
Another example is the reduction by Henzinger and King \cite{HenzingerK99} to obtain a dynamic {\em approximate} minimum spanning forest algorithm $\cA$ from a dynamic spanning forest algorithm $\cB$. This reduction requires that $\cB$ works against adaptive adversaries, even if we just want $\cA$ to only work against oblivious adversaries\footnote{In particular, the reduction of Henzinger and King \cite{HenzingerK99} involves maintaining a spanning forest $F_i$ on graph $G_i$ consisting of edges of weight roughly $(1+\epsilon)^i$ in the input graph and the maintained minimum spanning forest $F$, for every $i$. Consequently, $F_i$ may affect $F$, which in turns affect the change $G_i$. So, the algorithm that maintains $F_i$ must be able to handle an adaptive adversary since its own output ($F_i$) may influence its input ($G_i$).}.
Thus, we cannot directly apply this reduction to the algorithms of \cite{KapronKM13,GibbKKT15} to obtain a dynamic approximate minimum spanning forest algorithm with polylogarithmic worst-case update time\footnote{Note that Kapron~et~al.~\cite{KapronKM13} and Gibb~et~al.~\cite{GibbKKT15} claimed that the reduction of Henzinger and King \cite{HenzingerK99} together with their algorithms imply dynamic algorithms for the approximate minimum spanning forest problem with polylogarithmic worst-case update time. We believe that this claim is not correct. 
However, we also believe that one can still obtain the latter result both by modifying the reduction of \cite{HenzingerK99} and by modifying the algorithms of \cite{KapronKM13,GibbKKT15}.}.

Motivated by the discussions above, we focus on scenarios where algorithms can guarantee worst-case update time and additionally either (i) work against adaptive adversaries, or (ii) are Las Vegas, or both. To this end, we note again that deterministic algorithms can guarantee both (i) and (ii). Thus, the classic $O(\sqrt{n})$ bound of \cite{EppsteinGIN97,Frederickson85} applies to this scenario. This bound was slightly improved to $O(\sqrt{n (\log\log n)^2/\log n})$ recently by Kejlberg-Rasmussen~et~al.~\cite{Kejlberg-Rasmussen16} using word operations. 
It remains an important open problem whether we can {\em polynomially} improved the long-standing $O(\sqrt{n})$ bound as this will likely need new techniques beyond word tricks.

\input{contribution}

\paragraph{An independent work}
Wulff-Nilsen \cite{Wulff-Nilsen16a} independently presents an algorithm for solving a harder problem of maintaining a {\em minimum} spanning forest and not just some spanning forest as in our result. His algorithm is Las Vegas randomized and has $O(n^{0.5-\epsilon})$ worst-case update time, for some constant $\epsilon>0$, both in expectation and with high probability.

%% file: contribution.tex
\paragraph{Our Contributions.} We show two algorithms that achieve the above.
Our first algorithm is a randomized {\em Monte Carlo} algorithm that works against adaptive adversaries and has $O(n^{0.4+o(1)})$ worst-case update time, where the $o(1)$ term hides the $O(\sqrt{\log\log n/\log n})$ factor. At any point in time the forest maintained by it is a spanning forest with high probability.
\begin{thm}
	There is a randomized dynamic $\st$ algorithm for any graphs with
	$n$ nodes and $m$ initial edges that works correctly against adaptive
	adversaries with high probability and has preprocessing time
	$O(m^{1+o(1)})$ and worst-case update time $O(n^{0.4+o(1)})$.
\end{thm}
The term ``works correctly against adaptive adversaries with high probability'' is a technical term which is formalized in \Cref{sec:prelim} and \Cref{sec:formalization}. Roughly it means that if we pick an arbitrary time $t$, and consider the solution maintained by the algorithm at that time, such solution is a spanning forest with high probability (when considered all possible random choices). Note that we are talking about an infinite number of updates here, so the algorithm might make several mistakes in the past (which the adversary can see). Regardless, the guarantee holds with high probability.\footnote{An intuition is that mistakes from the far past should not affect an algorithm since it has enough time to run a static algorithm to detect and fix such mistakes. Thus only the algorithm' outputs from a polynomial number of previous steps can affect the algorithm, but then we can use the union bound to argue that the algorithms did not make mistakes in any of these steps.}

Our second algorithm is a randomized {\em Las Vegas} algorithm that has $O(n^{0.49305+o(1)})$ worst-case update time with high probability (this also implies the same expected time). This algorithm also works against adaptive adversaries. 
\begin{thm}
	There is a randomized dynamic $\st$ algorithm for any graphs with
	$n$ nodes and $m$ initial edges that works correctly against adaptive
	adversaries with certainty and has preprocessing time $O(m^{1+o(1)})$
	and worst-case update time $O(n^{0.49305+o(1)})$ with high probability.\label{thm:intro:dynamic lv final}
\end{thm}

Note that both algorithms are among a few randomized algorithms that work against adaptive adversaries. Moreover, the second result answer the open problem raised by Kapron~et~al.~\cite{KapronKM13}. 

The key to our results is the notion of {\em expansion decomposition} which is a decomposition of a graph into a sparse graph and connected components with high expansion. This decomposition serves as an interface between recent flow-related techniques and known dynamic graph algorithmic techniques: On the one hand, we show how to efficiently construct and maintain the decomposition using  flow-related techniques such as fast (static) max-flow approximation \cite{Peng14,Sherman13,KelnerLOS14},  cut-matching games  \cite{KhandekarRV09}, and local cut improvement \cite{OrecchiaZ14}.
On the other hand, the decomposition allows us to focus on solving the dynamic \st problem only on a very sparse graph 
and a graph with high expansion. In these cases, we can combine known techniques in dynamic graph algorithms such as sparse recovery\footnote{Techniques for the $1$-sparse recovery problem was used in \cite{GibbKKT15,KapronKM13}. In this paper, we need more general techniques for the $s$-sparse recovery problem for large $s$.} \cite{GibbKKT15,KapronKM13,BerindeGIKS08}, sparsification \cite{EppsteinGIN97}, ET tree~\cite{HenzingerK99}, a modification of the reduction from $k$-weight \mst to \st~\cite{HenzingerK99}, 
and a variant of the 2-dimensional topology tree~\cite{Frederickson85,Thorup07mincut}.
We refer to \Cref{sec:overview} for a more comprehensive overview.

To the best of our knowledge, our results are the first applications of flow-related techniques in fully-dynamic graph algorithms. (Previously \patrascu and Thorup \cite{PatrascuT07} used some relevant techniques in a special setting where updates can happen (in bulk) only once.) These results suggest the possibility to obtain stronger results (e.g. lower update time, deterministic algorithms and algorithms for harder problems) by further understanding these techniques in the context of dynamic graphs. 

%% file: overview.tex
\section{Overview}\label{sec:overview}

To simplify our discussion, in this section we use the following notations. For any functions $f(n)$ and $g(n)$ of $n$, we say that  $f(n)=\po(g(n))$ if there exists some constant $\epsilon>0$ such that $f(n)=O(g(n)^{1-\epsilon})$. 
In this section we will focus on getting an $\po(n^{1/2})$ worst case update time; in other words, we focus on getting an $O(n^{1/2-\epsilon})$ worst-case update time for some very small constant $\epsilon>0$ without worrying about the specific value of $\epsilon$. 
First, it can be shown, using standard techniques, that we only have to focus on a special case as in the following lemma.

\begin{lem}[Details in \Cref{sec:reductions}]\label{lem:intro:reductions}
It is sufficient to construct a dynamic spanning forest algorithm that (i) has $n^{1+o(1)}$ proprocessing time, (ii) can handle a dynamic graph $G$ that has maximum degree $\Delta\leq 3$ at all time, and (iii) can handle $\tau=n^{1/2+\epsilon}$ edge updates for some small constant $\epsilon>0$ (as oppose to handling infinite number of updates).
\end{lem}
\begin{proof}
	[Proof Idea]First, we can assume that the graph is sparse by the sparsification
	technique of Eppstein~et~al.~\cite{EppsteinGIN97}\footnote{One thing we have to be careful of is that the sparsification technique of \cite{EppsteinGIN97} only works if we can maintain a {\em stable} spanning forest. The spanning forests maintained by our algorithms are {\em not} stable as exactly defined in \cite{EppsteinGIN97}. This is mainly because our algorithms are randomized. However, this can be easily dealt with by slightly modifying the original notion of stability, as done in \Cref{sec:reductions}.}. 
	Next, we reduce to the case where
	the graph has maximum degree 3, which is also done implicitly in \cite{Frederickson85,Kejlberg-Rasmussen16}.
	Given a graph $G$, we maintain a graph $G'$ with maximum degree $3$
	by ``splitting'' each node in $G$ into a path in $G'$ of length
	equals to its degree. As $G$ is sparse, there are $O(n)$ nodes in
	$G'$. We assign weight 0 to edges in those paths, and 1 to original
	edges. Observe that tree-edges with weight 1 in a 2-weight minimum
	spanning forest $F'$ of $G'$ form a spanning forest $F$ of $G$.
	So we are done by maintaining $F'$. (There is a reduction from dynamic
	$k$-weight $\mst$ to dynamic $\st$ by Henzinger and King~\cite{HenzingerK99}.)
	
	Next, we reduce to the case when there is only $\tau=n^{1/2+\epsilon}$
	updates. We divide a long sequence of updates into phases of
	length $\tau$. We concurrently maintain two instances $\cA_{1}$
	and $\cA_{2}$ of a dynamic $\st$ algorithm. In odd phases, $\cA_{1}$
	maintains a spanning forest with update time $O(n^{1/2-\epsilon})$,
	and in the same time we evenly distribute the work for preprocessing
	$\cA_{2}$ into each update of this phase, which takes $O(n^{1+o(1)}/\tau)$
	time per update. So, in the next even phase, $\cA_2$ is ready to maintain a spanning forest. 
	Then we do everything symmetrically in even phases. Finally, we can
	``combine'' the two dynamic spanning forests maintained by $\cA_{1}$
	and $\cA_{2}$ into one dynamic spanning forest using additional $O(\log^{2}n)$
	update time via a standard trick (which is implicit in a reduction
	from dynamic $k$-weight $\mst$ to dynamic $\st$ by \cite{HenzingerK99}).
	In total, the update time is $O(n^{1/2-\epsilon}+n^{1+o(1)}/\tau+\log^{2}n)=O(n^{1/2-\epsilon+o(1)})$\footnote{
	Using this technique for randomized algorithms, we can only reduce from the case of the polynomial-length update sequences. 
	This is because the trick for ``combining'' two dynamic spanning forests will accumulate the failure probability over time. 
	To handle infinite-length sequences, we need a different reduction. See \ref{lem:reduc infty length}.}.
\end{proof}

Our technical ideas evolve around the notion of {\em graph expansion}. Given a graph $G=(V, E)$ and a set of nodes $S\subset V$, let $\partial_G(S)$ denote the set of edges across the cut $S$ in $G$.   
The expansion of $S$ is $\phi_G(S)=|\partial_G(S)|/\min\{|S|,|V\setminus S|\}$, and the expansion of a graph $G$ is $\phi(G)=\min_{\emptyset\neq S\subset V}\phi_G(S)$.\footnote{Readers who are familiar with the notions of expansion and conductance may observe that the two notions can be used almost interchangeably in our case since we can assume that our input graph has maximum degree at most three.}
Each of our algorithms consists of two main components. The first component is algorithms to {\em decompose} a graph into subgraphs with high expansion together with some other subgraphs that are easy to handle. This allows us to focus only on high-expansion subgraphs. Maintaining a spanning forest in a high-expansion is our second component. 
To be concrete, let us start with our first decomposition algorithm which is a building block for both our results. This algorithm takes near-linear time to decompose an input graph into a sparse subgraph and many connected components with high expansion, where the sparsity and expansion are controlled by a given parameter $\alpha$ as in the following theorem.

\begin{thm}
	\label{thm:intro:high-exp-decomp}There is a randomized
	algorithm $\cA$ that takes as inputs an undirected graph $G=(V,E)$ with
	$n\ge2$ vertices and $m$ edges and a parameter $\alpha>0$ ($\alpha$ might depend on $n$). Then, in $O(m^{1+o(1)})$ time, $\cA$
	outputs two graphs $G^{s}=(V,E^{s})$ and $G^{d}=(V,E^{d})$ with the following properties.
	\begin{itemize}
		\item $\{E^{s}, E^{d}\}$ is a partition of $E$,
		\item $G^{s}$ is sparse, i.e. $|E^{s}|\le\alpha n^{1+o(1)}$, and
		\item with high probability, each connected component $C$ of $G^{d}$ either is a singleton or
		has high expansion, i.e. $\phi(C)\ge\alpha$.
	\end{itemize}
\end{thm}

\begin{proof}[Proof Idea (details in \Cref{sec:decomp alg})]
There are two main steps. The first step is to devise a near-linear time approximation algorithm for a problem
called \emph{most balanced sparse cut}. This problem
is to find a cut $S$ with largest number of nodes such that $|S|\leq n/2$ and
$\phi(S)\geq \alpha$.
This problem is closely related to sparsest cut and balanced cut problems
(cf. \cite{LeightonR99}). 
One way to approximately solve both problems is by using the {\em cut-matching game} framework of Khandekar,
Rao and Vazirani \cite{KhandekarRV09} together with known exact algorithms for the maximum flow problem. We modify this framework (in a rather straightforward way) so that (i) it gives a solution to the most balanced sparse cut problem and (ii) we can use {\em approximate} maximum flow algorithms instead of the exact ones. By plugging in near-linear time max flow algorithms \cite{Peng14,Sherman13,KelnerLOS14}, our algorithm runs in near-linear time.

The second step is to use the most balanced sparse cut algorithm to construct the decomposition. 
We note that a similar decomposition was constructed by Andoni~et~al. \cite{AndoniCKQWZ16} in the context of graph sketch. 
This algorithm repeatedly finds and removes a cut $S$ such that $\phi(S)<\alpha$ from the input graph. It is too slow for our purpose since it
may involve many cuts $S$ such that $|S|=1$, causing as many as $\Omega(n)$ repetitions and as much as $\Omega(mn)$ running time. 
Instead, we use the most balanced cut $S$. The idea is that if $S$ is large, the cut divides the graph into two subgraphs with similar number of nodes, and this cannot happen more than $O(\log n)$ time. It is still possible that we find a cut $S$ that $|S|$ is small, but we can argue that this does not happen often. To this end, we note that the exact algorithm and analysis is quite involved because we only have approximate guarantees about the cut. This incurs a factor of $n^{O(\sqrt{\log\log n/\log n}})=n^{o(1)}$ in the running time. 
\end{proof}

We call the above algorithm the {\em global expansion decomposition algorithm} to contrast it with another algorithm that is {\em local} in the sense that it does not read the whole graph (this algorithm will be discussed soon). 

\paragraph{The Monte Carlo Algorithm (Details in \Cref{sec:monte_carlo}).}
With the above decomposition, we already have one main subroutine for our Monte Carlo algorithm against adaptive adversaries. This subroutine is run at the preprocess with parameter $\alpha=1/n^{\epsilon}$, for some constant $\epsilon>0$, to decompose the input graph into $G^d$ and $G^s$. Now we argue that the only other main subroutine that we will need is maintaining a spanning forest on $G^d$; in other words, we do not have to worry so much about the graph $G^s$. The intuition is that $G^s$ has $O(\alpha n^{1+o(1)})=\po(n)$ edges, and thus we can maintain a spanning tree on $G^s$ in $\po(n^{1/2})$ worst-case update time using Frederickson's algorithm. (The real situation is slightly more complicated since we have to maintain a spanning tree in a graph $G'$ consisting of edges in $G^s$ {\em and} a spanning forest of $G^d$. This is because we need to get a spanning forest of $G=(V, E^d\cup E^s)$ in the end. Fortunately, although $G'$ may not have $\po(n)$ edges, we only need to slightly modify Frederickson's data structure to obtain the $\po(n^{1/2})$ update time as desired.) 
Similarly, since we can assume that there are at most $\tau=\po(n)$ edges inserted due to \Cref{lem:intro:reductions}, it will not be hard to handle edge insertions.

So, we are left with maintaining a spanning forest on each connected component $C$ of $G^d$ when it undergoes (at most $\tau$) edge deletions. First we need another tool that is a simple extension of the tool used by Kapron~et~al.~\cite{KapronKM13,GibbKKT15}. In \cite{KapronKM13,GibbKKT15}, Kapron~et~al. showed a deterministic data structure that can maintain a forest $F$  (not necessarily spanning) in a dynamic graph $G$ and answer the following query: Given a pointer to one tree $T$ in $F$, if $\partial_G(V(T))$ contains {\em exactly} one edge, then output that edge (otherwise, the algorithm can output anything). The update and query time of this data structure is polylogarithmic. 
The main tool behind it is a {\em $1$-sparse recovery} algorithm from data streams (e.g. \cite{BerindeGIKS08}).
By using an $s$-sparse recovery algorithm instead for a parameter $s$, this data structure can be easily extended to answer the following query: Given a pointer to one tree $T$ in $F$, if $\partial_G(V(T))$ contains {\em at most $s$ edges}, then output all those edges. The update and query time of this data structure is $\tilde O(s)$. We call this data structure a {\em cut recovery tree}. (See \Cref{sec:linear sketches} for details.) Now the algorithm:

\medskip\noindent{\bf Algorithm $\cB$.}  Recall that initially $\phi(C)\geq \alpha$, and we want to handle at most $\tau$ edge deletions. Consider a spanning forest $F$ of $C$ maintained at any point in time. Consider when an edge $e$ is deleted. If $e$ is not in $F$, we do nothing. If it is, let $T$ be the tree that contains $e$. The deletion of $e$ divides $T$ into two trees, denoted by $T_1$ and $T_2$. 
Our job is to find one edge in $\partial_C(V(T_1))$, if exists, to reconnect $T_1$ and $T_2$ in $F$. Assume wlog that $|V(T_1)|\leq |V(T_2)|$. First we sample $(1/\alpha+\sqrt{\tau/\alpha})\polylog(n)$ edges among edges in $C$ incident to nodes in $T_1$. This can be done in $\tilde O(1/\alpha+\sqrt{\tau/\alpha})$ worst-case time by ET tree just as Henzinger and King did in \cite{HenzingerK99}. If one of the sampled edges connects between $T_1$ and $T_2$, then we are done. If not, we use the cut recovery tree with parameter $s=\sqrt{\tau/\alpha}$ to list (at most $s$) edges in $\partial_C(V(T_1))$. If some edge in $\partial_C(V(T_1))$ is listed, then we use such edge to reconnect $T_1$ and $T_2$; otherwise, we leave $T_1$ and $T_2$ as two separated trees in $F$. 

The above algorithm takes worst-case update time $\tilde O(1/\alpha+\sqrt{\tau/\alpha})= \tilde O(n^{1/4+O(\epsilon)})$, which is $\po(n^{1/2})$ for small enough $\epsilon$. For the correctness analysis, we argue that with high probability the algorithm does not make any mistake at all over the period of $\tau$ updates. First note that we initially get the decomposition as promised by \Cref{thm:intro:high-exp-decomp} with high probability. So we will assume that this is the case, and in particular $\phi(C)\geq \alpha$ initially. Now consider three cases. 
\begin{itemize}
	\item Case 1: $|V(T_1)|\geq 2\tau/\alpha$. In this case, at the time we delete $e$ we know that 
	\begin{align}
	|\partial(V(T_1))| &\geq \alpha |V(T_1)| - \tau &&\mbox{(since we have deleted at most $\tau$ edges in total)} \nonumber\\
	& \geq \alpha |V(T_1)|/2 &&\mbox{(since $|V(T_1)|\geq 2\tau/\alpha$).} \label{eq:intro:large cut}
	\end{align}
Since there are at most $3|V(T_1)|$ edges incident to nodes in $T_1$, when we sample $(1/\alpha+\sqrt{\tau/\alpha})\polylog(n)$ edges in the first step, we will get one of the edges in $\partial(V(T_1))$ with high probability.

	\item Case 2: $|V(T_1)|< 2\tau/\alpha$ and $|\partial(V(T_1))|\geq \sqrt{\tau/\alpha}$ at the time we delete $e$.  Similarly to the previous case, since there are at most $3|V(T_1)|<6 \tau/\alpha$ edges incident to nodes in $T_1$, when we sample $(1/\alpha+\sqrt{\tau/\alpha})\polylog(n)$ edges in the first step, we will get one of the edges in $\partial(V(T_1))$ with high probability.
	
	\item  Case 3: $|V(T_1)|< 2\tau/\alpha$ and $|\partial(V(T_1))|< \sqrt{\tau/\alpha}$ at the time we delete $e$. Since $|\partial(V(T_1))|< \sqrt{\tau/\alpha}$ the cut recovery tree data structure with parameter $s= \sqrt{\tau/\alpha}$ will list all edges in $\partial(V(T_1))$ correctly in the second step of the algorithm. Thus we will always find an edge to reconnect $T_1$ and $T_2$, if there is one.
\end{itemize}

To conclude, in all cases the algorithm will find an edge to reconnect $T_1$ and $T_2$, if there is one, with high probability. Observe that the analysis exploits the fact that $C$ has high expansion (compared to $\tau$) initially by arguing (as in Case~1) that if $V(T_1)$ is large, then there will be plenty of edges in $\partial(V(T_1))$ that are not deleted. 

Finally, we provide some intuition why the analysis above holds even when we allow an adversary to see the maintained spanning forest (the adaptive adversary case). In fact, we argue that this is the case even when the algorithm {\em reveals all random choices} it has made so far to the adversary. In particular, we allow the adversary to see (i) all $(1/\alpha+\sqrt{\tau/\alpha})\polylog(n)$ edges sampled in Step 1 and (ii) the initial decomposition (from the algorithm in \Cref{thm:intro:high-exp-decomp}). We can reveal (i) because every random bit is used {\em once} to sample edges, and will not be used again in the future. Thus, knowing these random bits is useless for the adversary to predict the algorithm's behavior in the future. We can reveal (ii) because our three-case analysis only needs to assume that the decomposition works correctly initially. The only thing the adversary can exploit is when this is not the case, but this happens with a small probability. We refer to \Cref{sec:monte_carlo} for formal arguments and further details.

{\em Remark:} We note a lesson that might be useful in designing randomized dynamic algorithms against adaptive adversaries in the future. The reason that most randomized algorithms fail against adaptive adversaries is that their future behavior heavily depends on random bits they generated in the past. In contrast, the random bits our algorithm used for edge sampling are not reused, thus do not affect the future. 
The dynamic \st algorithm by Henzinger and King \cite{HenzingerK99} works against adaptive adversaries by exactly doing this. 
In addition to this, in our case, the random bits our decomposition algorithm used may affect the future since the decomposition is used throughout. However, what matters in the future is only whether the decomposition algorithm gives a correct output or not. How the output looks like does not really matter as long as it is correct, which is the case with high probability.

\paragraph{The Las Vegas Algorithm (Details in \Cref{sec:las_vegas}).} The goal now is to construct an algorithm that can {\em detect when it makes mistakes} (so that it can, e.g., restart the process). To motivate our new algorithm, let us re-examine the previous Monte Carlo algorithm when it makes mistakes. First, the initial decomposition might not be as guaranteed in \Cref{thm:intro:high-exp-decomp}, and it is not clear how we can check this. We will leave this issue aside for the moment (it will be easy to handle once we take care of other issues). Now assuming that the decomposition is correct, another issue is that Algorithm~$\cB$ may also make errors as it can happen that none of the $(1/\alpha+\sqrt{\tau/\alpha})\polylog(n)$ sampled edges are in $\partial(V(T_1))$, but there is actually an edge in $\partial(V(T_1))$ to reconnect $T_1$ and $T_2$. 
However, Case~1 of the analysis, which is the crucial part that exploits the fact that $C$ has high expansion (compared to $\tau$) initially, is still useful: in this case, we know that there {\em is} an edge to reconnect $T_1$ and $T_2$ since $|\partial(V(T_1))| \geq \alpha |V(T_1)|/2 $ (\Cref{eq:intro:large cut}). Thus in this case, the algorithm {\em knows} that it makes a mistakes if it does not find an edge to reconnect among the sampled edges. When $|V(T_1)|< 2\tau/\alpha$ however, we do not know how to distinguish between Cases 2 and 3, and thus will naively consider all edges incident to nodes in $C$. 
Using this idea with parameters slightly adjusted, we have the following algorithm.

\medskip\noindent{\bf Algorithm $\cC$.} {\em Case 1: $|V(T_1)|\geq 2\tau/\alpha$}. Sample $\polylog(n)/\alpha$ edges from edges incident to nodes in $T_1$. If one of these edges are in $\partial(V(T_1))$, then we can reconnect $T_1$ and $T_2$; otherwise, the algorithm realizes that it fails and outputs ``fail''. {\em Case 2:} Consider all edges incident to nodes in $T_1$. (Note that there are at most $3|V(T_1)|=O(\tau/\alpha)$ such edges.) If one of these edges are in $\partial(V(T_1))$, then we can reconnect $T_1$ and $T_2$; otherwise, $T_1$ and $T_2$ become separated trees in the maintained forest.

The above algorithm is clearly Las Vegas, as it realizes when it makes mistakes. This is also the case even when we take into account the fact that the initial decomposition may fail: when Algorithm $\cC$ outputs ``fail'' it means that either (i) $\cC$ does not get an edge in $\partial(V(T_1))$, which is guaranteed to exist if the decomposition is correct, as a sample, or (ii) the decomposition itself is incorrect that such edge does not exists. In other words, the output ``fail'' of $\cC$ capture the failures of its own and of the decomposition algorithm. 

Now, observe that the update time of $\cC$ is dominated by the second step, which is $\tilde O(\tau/\alpha)$. This is {\em not} the $\po(n^{1/2})$ as we desire (recall that $\tau=n^{1/2+\epsilon}$ and $\alpha=1/n^{\epsilon}$). Unfortunately, we do not know how to improve $\cC$'s update time. So, instead we have to be more clever in using $\cC$. The plan is the following. We will divide the sequence of at most $\tau$ updates on each high-expansion component $C$ into {\em phases}, where each phase consists of $\tau'=\po(n^{1/2})$ updates. 
In the beginning of each phase, we decompose $C$ further into components whose expansion is $\alpha'\leq \alpha$. (Ideally, we want $\alpha'$ to be as high as $\alpha$, but there are some limits to this, as we will see shortly.) 
In each component, say $C'$, we run algorithm $\cC$ as before, but with parameters $\alpha'$ and $\tau'$. The correctness of the algorithm is guaranteed in the same way as before because $C'$ has expansion at least $\alpha'$ initially. The worst-case update time is then $\tilde O(\tau'/\alpha')$. 
Now, observe that if the decomposition algorithm run in the beginning of each phase take $\lambda$ time, then this cost can be charged to $\tau'$ updates in the phase, causing the cost of $\lambda/\tau'$ on average. By a standard technique (similar to the proof sketch of \Cref{lem:intro:reductions}), we can turn this averaged time into a worst-case update time. Thus, the total worst-case update time of the algorithm will be 
\begin{align}
  \tilde O(\tau'/\alpha' + \lambda/\tau'). \label{eq:intro:las vegas update time}
\end{align}
Our goal is to make the above update time be $\po(n^{1/2})$. To do this, there is one obstruction though: If we use the decomposition algorithm from \Cref{thm:intro:local decomp} in the beginning of every phase, then $\lambda$ can be as large as $O(n^{1+o(1)})$ (when the connected component is big). It will then be impossible to make the above update time be $\po(n^{1/2})$ (note that $\tau' = \po(n^{1/2})$\thatchaphol{I changed $\alpha'<1$ to $\tau' = \po(n^{1/2})$}). So, we need a different algorithm for the initial decomposition of each phase, and in particular it should not read the whole connected component. The last piece of our algorithm is such a decomposition algorithm, which we call a {\em local algorithm} as it does not need to read the whole connected component we are trying to decompose.

\paragraph{Local Expansion Decomposition (\Cref{sec:local decomp}).} This algorithm, denoted by $\cA'$, operates in the {\em local} setting where there is a graph $G_b=(V, E_b)$ represented by an adjacency list stored in a memory, which was not read by the algorithm. (In our case, $G_b$ will be the connected component $C$.) Algorithm $\cA'$ then takes parameters $\alpha_b$ and $\epsilon'$, and a set of edge deletions $D\subseteq E_b$ as inputs. It then gives a decomposition of $G=(V, E_b-D)$ that is similar to that in \Cref{thm:intro:local decomp} at a high level; some major differences are: 
\begin{itemize}
	\item It takes $O(\frac{|D|^{1.5+\epsilon'}}{\alpha_{b}^{3+\epsilon'}}n^{o(1)})$ time, and in particular does not need to read the whole graph $G_b$. 
	\item It guarantees to outputs some ``desired'' decomposition only when $\phi(G_b)\geq \alpha_b$. 
	\item In the ``desired'' decomposition, the expansion of each high-expansion component $C$ is $\phi(C) = \Omega(\alpha_{b}^{1/\epsilon'})$. 
\end{itemize}

We note the trade-off parameter $\epsilon'$ that appears above: when $\epsilon'$ is small, the algorithm is fast but has a bad expansion guarantee in the output, and when $\epsilon'$ is big, the algorithm is slow but has a good expansion guarantee. We will have to choose this parameter carefully in the end. We state the result about the local decomposition algorithm in more details. (For full details, see \Cref{sec:local decomp}.)

\begin{thm}
	\label{thm:intro:local decomp}For
	any constant $\epsilon'\in(0,1)$, there is an algorithm $\cA'$ that
	can do the following:
	\begin{itemize}
		\item $\cA'$ is given pointers to $G_{b}$,$D$ and $\alpha_{b}$ stored in a memory:  
		$G_{b}=(V,E_{b})$ is a $3$-bounded degree graph with $n$ nodes represented by an adjacency list. 
		 $D\subset E_{b}$ is a set of edges in $G_{b}$. $\alpha_{b}$ is an expansion parameters.
		  Let $G=(V,E)=(V,E_{b}-D)$ be the graph that $\cA'$
		will compute the decomposition on. 	
		\item Then, in time $O(\frac{|D|^{1.5+\epsilon'}}{\alpha_{b}^{3+\epsilon'}}n^{o(1)})$,
		$\cA'$ either 
		\begin{itemize}
			\item reports failure meaning that $\phi(G_{b}) < \alpha_b$, 
			\item outputs two graphs $G^{s}=(V,E^{s})$ and $G^{d}=(V,E^{d})$ (except the largest connected component of $G^d$) where $\{E^s,E^d\}$ is a partition of $E$.
		\end{itemize}
    	(Note that in the latter case it is still possible that $\phi(G_{b}) < \alpha_b$, but $\cA'$ cannot detect it.)
		\item Moreover, if $\phi(G_{b}) \ge \alpha_b$, then with
		high probability we have
		
		\begin{itemize}
			\item $|E^{s}| = O(|D|/\alpha_{b})$, and
			\item each connected component $C$ of $G^{d}$ either is a singleton or
			has high expansion: $\phi(C) = \Omega(\alpha_{b}^{1/\epsilon'})$.
		\end{itemize}
	\end{itemize}
\end{thm}

\begin{proof}
	[Proof Idea (details in \ref{sec:local decomp})] Similar to the proof of \ref{thm:intro:high-exp-decomp}, there are two main
	steps. 
	The first step is to devise a \emph{local}
	approximation algorithm for a problem called \emph{locally balanced sparse
		cut} \emph{(LBS cut)}. In this
	problem, we are given a graph $G=(V,E)$, a target set $A\subset V$
	and a parameter $\alpha$. Then we need to find a $\alpha$-sparse
	cut $S$ (i.e. $\phi_{G}(S)<\alpha$) where $|S|\le|V-S|$ such that
	$|S|$ is larger than all $\alpha$-sparse cuts which are ``near'' the target set
	$A$. (``Nearness'' is defined precisely in \ref{def:overlapping}).
	To compare, in the most balanced sparse cut problem, $|S|$ needs to be
	larger than \emph{all} $\alpha$-sparse cuts. By slightly modifying and analyzing the
	algorithm by Orecchia and Zhu \cite{OrecchiaZ14} for a related problem
	called the {\em local cut improvement problem,} we obtain an approximation algorithm
	which is local (i.e. its running time depends essentially only on
	$|A|$, and not $|V|$).
	
	The second step to obtain the decomposition is to find an approximate LBS cut where the target set $A$ is the endpoints
	of $D$ (together with some additional nodes) and recurse on both sides.
	To bound the running time, we maintain the same kind of invariant
	as in the proof of \ref{thm:intro:high-exp-decomp} throughout the recursion. But, in order to argue
	that the invariant holds, the analysis is more involved. 
	The main reason is because we compute LBS cuts which
	have a weaker guarantee, instead of computing most balanced sparse
	cuts as in the global expansion decomposition algorithm. 
	However, an important observation is that,
	when $\phi(G_{b})\ge\alpha_{b}$, any $(\frac{\alpha_{b}}{2})$-sparse
	cut $S$ in $G=G_{b}-D$ must be ``near'' to $D$. Intuitively,
	this is because the expansion of $S$ get halved after deleting edges in
	$D$. This justifies why it is enough to find an approximate $(\frac{\alpha_{b}}{2})$-sparse
	LBS cut instead of finding an approximate most balanced sparse cut. As our approximate LBS cut
	algorithm is local, we can output the decomposition in time essentially
	independent from the size of $G$.
\end{proof}

\thatchaphol{TO WRITE: The input $D$ will be of size $\Theta(n^{1/2+\epsilon})$. So it is important that the dependency on $|D|$ is subquadratic, which is the case here.}

We now finish off our Las Vegas algorithm using the above local decomposition algorithm $\cA'$. Let $\epsilon$ be a very small constant, and we let $\tau=n^{1/2+\epsilon}$ and $\alpha=n^{\epsilon}$. Recall that we want to maintain a spanning forest of a connected component $C$ undergoing at most $\tau$ edge deletions that has expansion $\alpha$ initially, by dividing into phases of $\tau'$ updates. 
In the beginning of each phase, we invoke $\cA'$ with $G_b=C$, $\alpha_b=\alpha$, 
$\epsilon' = \sqrt{\epsilon}$, and $D$ being the set of edges in $C$ deleted so far. 
Since we start with a component $C$ with expansion at least $\alpha$, each high-expansion component $C'$ in the resulting decomposition has expansion at least $\alpha'=\Omega(\alpha^{1/\epsilon'})$. Since $|D|\leq \tau$, this algorithm takes time $\lambda=O(\frac{|\tau|^{1.5+\epsilon'}}{\alpha^{3+\epsilon'}}n^{o(1)})$. By plugging these values in \Cref{eq:intro:las vegas update time}, we have that the update time of our Las Vegas algorithm is 
\begin{align*}
 \tilde O(\frac{\tau'}{\alpha'} + \frac{\lambda}{\tau'}) & =\tilde O\left(\frac{\tau'}{\alpha'} + \frac{|\tau|^{1.5+\epsilon'}}{\alpha^{3+\epsilon'}\tau'}n^{o(1)}\right)
\end{align*}
It is left to pick right parameters to show that the above is $\po(n^{1/2})$. To this end, we choose $\tau' = \alpha'\times n^{1/2-\epsilon}$, so that the first term in the update time is $n^{1/2-\epsilon}$. The second term then becomes $n^{1/4+O(\epsilon+\epsilon\epsilon'+\epsilon/\epsilon'+\epsilon')}$, which is $\po(n^{1/2})$ when $\epsilon$ is small enough.

Finally, we note that the above algorithm is Las Vegas even though the local algorithm $\cA'$ is Monte Carlo for the same reason as we argued before for the case of global algorithm.

\subsection*{Organization}

We first introduce relevant notations and definitions
in \ref{sec:prelim}. By modifying and combining the known techniques,
we present basic tools and give a reduction that proves \ref{lem:intro:reductions}
in \ref{sec:basic}. In \ref{sec:decomp alg}, we show the global expansion
decomposition algorithm which is the backbone of all our dynamic $\st$
algorithms. Then, we extend the idea to show the novel local expansion
decomposition algorithm which is a crucial tool for our Las Vegas
algorithm, in \ref{sec:local decomp}. Finally, we show how all ideas
fit together and obtain the Monte Carlo algorithm in \ref{sec:monte_carlo}
and the Las Vegas algorithm in \ref{sec:las_vegas}. 

\ref{sec:formalization}
contains formal definitions of oblivious and adaptive adversaries. \ref{sec:Omitted-proof} contains omitted proofs.

%% file: prelim.tex
\section{Preliminaries\label{sec:prelim}}

In the following, we say that an event $e$ occurs \emph{with high
probability} if $\Pr[e]\ge1-1/n^{c}$ where $n$ is the size of the
problem instance depending on the context (in this paper, $n$ is
usually a number of nodes in a graph), and $c$ is a fixed constant
which can be made arbitrarily large. 

Given a graph $G=(V,E)$ and any cut $S\subset V$, we have the following
definitions
\begin{itemize}
\item the set of cut edges is $\partial_{G}(S)=\{(u,v)\in E\mid u\in S$
and $v\notin S\}$,
\item the set of edges induced by the cut $S$ is $E_{G}(S)=\{(u,v)\in E\mid u,v\in S\}$,
\item the cut size of $S$ is $\delta_{G}(S)=|\partial_{G}(S)|$, 
\item the expansion of $S$ is $\phi(S)=\delta(S)/\min\{|S|,|V\setminus S|\}$, 
\item the volume of $S$ is $vol_{G}(S)=\sum_{u\in S}\deg(u)=2|E_{G}(S)|+\delta_{G}(S)$,
i.e. the number of endpoints of edges in $E$ incident to $S$.
\end{itemize}
$V(G)=V$ is the set of nodes of $G$ and $E(G)=E$ is the set of
edges of $G$. The expansion of the graph $G$ is $\phi(G)=\min_{\emptyset\neq S\subset V}\phi(S)$.

For any graph $G=(V,E)$, the \emph{incidence matrix} $B\in\{-1,0,1\}^{\binom{|V|}{2}\times|V|}$
of $G$ is defined as following: 
\[
B_{e,u}=\begin{cases}
1 & \mbox{if }e=(u,v)\in E\\
-1 & \mbox{if }e=(v,u)\in E\\
0 & \mbox{otherwise.}
\end{cases}
\]
Let $B_{u}$ be the $u$-th column vector of $B$. We will use the
following simple fact.
\begin{prop}
\label{thm:incidence matrix}For any cut $S\subset V$ and $(i,j)\in\binom{|V|}{2}$,
the $(i,j)$-th entry of the vector $\sum_{u\in S}B_{u}$ is non-zero
iff $(i,j)\in\partial_{G}(S)$.
\end{prop}
For any graph $G=(V,E)$, let $F\subseteq E$ be a forest in $G$.
We say that an edge $e\in E$ is a \emph{tree edge }if $e\in F$,
otherwise $e$ is a \emph{non-tree edge}. For any cut $S\subset V$,
let $nt_{G,F}(S)=(E_{G}(S)\cup\partial_{G}(S))\setminus F$ be a set
of non-tree edges in $G$ incident to $S$ and let $nt\_vol_{G,F}(S)=2|E_{G}(S)\setminus F|+|\partial_{G}(S)\setminus F|$
be the number of endpoints of edges in $E\setminus F$ incident to
$S$.

We say that $T$ is a tree in forest $F$, or just $T\in F$, if $T$
is a connected component in $F$. We usually assume that there is
a pointer to each tree $T\in F$. For readability, we usually write
$T$ instead of a pointer to $T$ in functions of data structures,
such as $\textsf{tree\_size}(T)$ in \ref{def:aug ET tree} and $\cutset k(T)$
in \ref{def:cut-recovery-tree}.

\subsection{Augmented ET Tree\label{sec:ET tree}}

We define the following data structures for convenience. Just by combining
the functionality of both ET tree by Henzinger and King \cite{HenzingerK99}
with link-cut tree by Sleator and Tarjan \cite{SleatorT83}, that
can handle $\textsf{path}$ operation, we have the following:
\begin{defn}
[Augmented ET Tree]\label{def:aug ET tree}Augmented Euler Tour (ET)
tree is a data structure that preprocesses a graph $G=(V,E)$ and
a forest $F\subseteq E$ as inputs and handles the following operations:
\begin{itemize}
\item $\instree(e)$: add edge $e$ into $F$, given that $e\in E\setminus F$
is a non-tree edge and adding does not cause a cycle in $F$.
\item $\deltree(e)$: remove edge $e$ from $F$, given that $e\in F$. 
\item $\insnon(e)$: insert edge $e$ into $E$, given that $e\notin E$.
\item $\delnon(e)$: delete edge $e$ from $E$, given that $e\in E\setminus F$. 
\item $\textsf{tree\_size}(T)$: return $|V(T)|$, given a pointer to\emph{
}$T$ which is a tree in $F$.
\item $\textsf{count\_tree}()$: return a number of connected components
in $F$.
\item $\textsf{find\_tree}(v)$: return a pointer to tree $T_{v}$ where
$v\in T_{v}$ and $T_{v}$ is a tree in $F$.
\item $\textsf{sample}(T)$: sample (using fresh random coins) an edge from
$nt_{G,F}(S)$ where $S=V(T)$ with probability proportional to its
contribution to $nt\_vol_{G,F}(S)$, given a pointer to\emph{ }$T$
which is a tree in $F$. More precisely, an edge from $nt_{G,F}(S)\cap\partial_{G}(S)$
and $nt_{G,F}(S)\cap E_{G}(S)$ is sampled with probability $\frac{1}{nt\_vol_{G,F}(S)}$
and $\frac{2}{nt\_vol_{G,F}(S)}$ respectively.
\item $\textsf{list}(T)$: list all edges incident to $T$, given a pointer
to $T$ which is a tree in $F$.
\item $\textsf{path}(u,v,i)$: return the vertex $w$ where the distance
from $u$ to $w$ in $F$ is $i$ and $w$ is in a unique path from
$u$ to $v$, given that $u$ and $v$ are in the same tree in $F$.
\end{itemize}
\end{defn}
\begin{thm}
\label{thm:aug ET tree}There is an algorithm for augmented ET tree
$\cE$ for a graph $G$ with $n$ nodes and $m$ initial edges and
a forest $F$ with preprocessing time $O(m)$. $\cE$ can handle all
operations in $O(\log n)$ time except the following: $\textsf{tree\_size}$
in time $O(1)$, $\textsf{count\_tree}$ in time $O(1)$, and $\textsf{list}(T)$
in time $O(nt\_vol_{G,F}(T))$.
\end{thm}
For completeness, we include the proof in \ref{sec:proof aug ET tree}.

\subsection{Definition of Dynamic $\protect\st$}
\begin{defn}
[Dynamic $\st$]A dynamic $\st$ algorithm $\cA$ is given an initial
graph to be preprocessed, and then $\cA$ must return an initial spanning
forest. Then there is an online sequence of edge updates, both insertions
and deletions. After each update, $\cA$ must returns the list of
edges to be added or removed from the previous spanning tree to obtain
the new one. We say $\cA$ is an incremental/decremental $\st$ algorithm
if the updates only contain insertions/deletions respectively.
\end{defn}
The time an algorithm uses for preprocessing the initial graph and
for updating a new spanning forest is called \emph{preprocessing time
}and\emph{ update time} respectively.

In this paper, we consider the problem where the update sequence is
generated by an adversary.\footnote{This is in contrast to the case where the update sequence is generated
by a random process e.g. in \cite{NikoletseasRSY95}.} There are two kinds of adversaries: oblivious and adaptive. An \emph{oblivious
adversary} initially fixes the whole update sequence, and then shows
each update to the algorithm one by one. On the other hand, an \emph{adaptive
adversary} generates each update after he sees the previous answer
from the algorithm. So each update can depend on all previous answers
from the algorithm and all previous updates from the adversaries itself.
Adaptive adversaries are obviously as strong as oblivious ones. In
\ref{sec:formalization}, we formalize these definitions precisely.
By \ref{thm:det adv}, we can assume that all adversaries are deterministic. 

Throughout the paper, we use the following notation. Given a dynamic
$\st$ algorithm $\cA$ and an adversary $f$, let $G_{0}$ be the
initial graph given by $f$. Let $F_{0}$ be the the initial forest
returned from $\cA$. For $i\ge1$, let $G_{i}$ and $F_{i}$ be the
graph and the forest after the $i$-th update. More precisely, $G_{i}$
is obtained by updating $G_{0}$ with all updates from $f$ from time
$1$ to $i$, and $F_{i}$ is obtained by updating $F_{0}$ with all
answers from $\cA$ (which describe which tree-edges to be added and
removed) from time $1$ to $i$. Since $f$ can be assumed to be deterministic,
we have that if $\cA$ is deterministic, then $G_{i}$ and $F_{i}$
are determined for all $i\ge0$. On the other hand, if $\cA$ is randomized
and use a string $R$ as random choices, then $G_{i}$ and $F_{i}$
are random variables depending on $R$ for all $i\ge0$. 
\begin{defn}
[Correctness]For any $p\in[0,1]$ and $T\in\mathbb{N}\cup\{\infty\}$,
a dynamic $\st$ algorithm \emph{$\cA$ }works correctly against adaptive
adversaries with probability $p$ for the first $T$ updates iff,
for any fixed adaptive adversary, we have $\Pr_{R}[F_{i}$ is a spanning
forest of $G_{i}]\ge p$ for \emph{each} $0\le i\le T$ where $R$
is a random strings used by $\cA$ as its random choices. If $T=\infty$,
then we omit the phrase ``for the first $T$ updates''. If $p=1$,
then we say \emph{$\cA$ }works correctly against adaptive adversaries
with certainty.\label{def:adversary}
\end{defn}

\begin{defn}
[Running Time]For any $p\in[0,1]$ and $T\in\mathbb{N}\cup\{\infty\}$,
a dynamic $\st$ algorithm \emph{$\cA$, }that works against adaptive
adversaries for the first $T$ updates, has preprocessing time $t_{p}$
and worst-case update time $t_{u}$ with probability $p$ iff, for
any fixed adaptive adversary, we have $\Pr_{R}[$$\cA$ preprocesses
$G_{0}$ and returns $F_{0}$ in time at most $t_{p}$$]\ge p$ and
$\Pr_{R}[$$\cA$ updates $F_{i-1}$ to be $F_{i}$ in time at most
$t_{u}$$]\ge p$ for \emph{each} $1\le i\le T$ where $R$ is a random
strings used by $\cA$ as its random choices. \label{def:running time}
\end{defn}
Suppose that $F_{i}$ is a spanning forest of $G_{i}$. If at time
$i+1$, a tree-edge $e\in F_{i}$ is deleted and $F_{i}\setminus e$
no longer spans $G_{i+1}=G_{i}\setminus e$. If $e'$ is such that
$F_{i}\cup e'\setminus e$ spans $G_{i+1}$, then we call that $e'$
is a \emph{replacement edge}. Sometimes we refer to the sequences
$G_{0},G_{1},\dots$ and $F_{0},F_{1},\dots$ as $G$ and $F$, respectively.
We call $G$ and $F$\emph{ }a \emph{dynamic graph }and a \emph{dynamic
spanning forest} respectively. We say that $\cA$ \emph{maintains
$F$ in $G$}. $G$ has $m$ \emph{initial edges} if $E(G_{0})\le m$,
and $G$ has at most $m$ edges if $E(G_{i})\le m$ for all $i$.

%% file: 2dET.tex
\section{Basic Tools\label{sec:basic}}

In this section, we present some useful tools for our dynamic $\st$
algorithms. All the tools presented in this section are obtained by
combining or modifying the existing techniques. This include 2-dimensional
ET Tree in \ref{sec:2d ET tree}, cut recovery tree in \ref{sec:linear sketches},
and a reduction which greatly simplifies our goal in \ref{sec:reductions}.

\subsection{2-dimensional ET Tree\label{sec:2d ET tree}}

The following data structure is useful for both of our Monte Carlo
and Las Vegas dynamic $\st$ algorithms. It will be used to help maintaining
a spanning forest in a graph with few non-tree edges.
\begin{defn}
[2-dim ET Tree]\label{def:2d ET tree}2-dimensional (2-dim) ET tree
is a data structure that preprocesses a graph $G=(V,E)$ and a forest
$F\subseteq E$ as inputs and handles the following operations:
\begin{itemize}
\item $\instree(e)$, $\deltree(e)$, $\insnon(e)$, $\delnon(e)$ and $\textsf{find\_tree}(u)$:
same definitions as in \ref{def:aug ET tree}.
\item $\textsf{find\_edge}(T)$: return an edge $(u,v)\in\partial_{G}(V(T))$
if exists, given a pointer to $T$ which is a tree in $F$.
\item $\textsf{find\_edge}(T_{1},T_{2})$\footnote{In the subsequent sections, the operation $\textsf{find\_edge}(T_{1},T_{2})$
will not be used. But we state it because it can be useful for other
applications.}: return an edge $(u,v)$ where $u\in V(T_{1})$ and $v\in V(T_{2})$
if exists, given pointers to\emph{ }$T_{1}$ and $T_{2}$ which are
both trees in $F$.
\end{itemize}
\end{defn}
We say that the number of non-tree edges of the underlying graph $G=(V,E)$
is bounded by $k$, if throughout all the update operations, we have
$|E-F|\le k$.
\begin{thm}
\label{thm:2dim ET tree}There is a deterministic data structure for
2-dim ET tree denoted by $\cE^{2d}$ that preprocesses $G,F$ and
a number $k$ as inputs where $G$ is a graph with $n$ nodes, $F$
is a forest in $G$, and the number of non-tree edges of $G$ is always
bounded by $k$. $\cE^{2d}$ preprocesses the input in time $O(n+k)$,
and can handle all operations in $O(\sqrt{k}\log k)$ time.
\end{thm}
There are several data structures based on topological tree \cite{Frederickson85}
that can be used to prove a slightly weaker version of \ref{thm:2dim ET tree}.
These data structures include 2-dimensional topological tree by Frederickson
\cite{Frederickson85}, and a generalization of top tree by Alstrup
et al. \cite[Section 5]{AlstrupHLT05} or by Thorup \cite[Section 5.3]{Thorup07mincut}.
They handle all operations of 2-dim ET tree in time $\tilde{O}(\sqrt{m})$
where $m$ is the number of \emph{all edges}. However, in our application,
we need to handle the operations in time $\tilde{O}(\sqrt{k})$ where
$k$ is the number of\emph{ non-tree edges} but $k$ might be much
less than $m$. 

We believe that, with some careful modification, these topological-tree-based
data structures can be used to prove \ref{thm:2dim ET tree}. Hence,
this theorem may be considered as a folklore. But, for completeness,
we instead show an implementation that is based on ET tree instead
which is simpler than topological tree and then prove \ref{thm:2dim ET tree}.
Again, even this ET-tree-based implementation might be considered
as a folklore.

For the rest of this section, we will show how to construct 2-dim
ET tree data structure that cannot handle $\textsf{find\_tree}$ operation.
But it is easy how to augment it so that it can handle $\textsf{find\_tree}$,
for example by combining with augmented ET tree.

\subsubsection{Reduction to a Constant Number of Connected Components}

We first show the reduction that allows us to assume that there are
only $O(1)$ number of connected components in $F$, i.e. there are
only $O(1)$ trees in $F$. Recall that we call each connected component
$T$ in $F$ as a tree $T\in F$.
\begin{lem}
Suppose there is a 2-dim ET tree data structure $\cD$ that preprocesses
any $(G,F)$ in time $t_{p}(n,k)$ where $G$ is a graph with $n$
nodes and $F$ is a forest in $G$, and the number of non-tree edges
of $G$ is bounded by $k$. Moreover, all operations of $\cD$ can
done in $t_{u}(n,k)$ time as long as the number of connected component
in $F$ is bounded by a constant. Then there is a 2-dim ET tree data
structure $\cE$ that preprocesses any $(G,F)$ in time $O(n+k+t_{p}(n,k))$
where $G$ is a graph with $n$ nodes and $F$ is a forest in $G$,
and the number of non-tree edges of $G$ is bounded by $k$. Moreover,
all operations of $\cE$ can done in $O(t_{u}(n,k))$ time.\label{thm:2-dim constant tree}\end{lem}
\begin{proof}
We will show how get obtain $\cE$ given that we have $\cD$. Given
$G=(V,E)$ and $F\subset E$, $\cE$ preprocesses $(G,F)$ as follows.
First, we construct $G'=(V',E')$ and $F'$. Let $r$ be an additional
node and let $V'=V\cup\{r\}$. The invariant of $F'$ is such that
for each connected component $T$ in $F$, there is a tree edge $(r,u_{T})\in F'$
where $u_{T}\in V(T)$ is some node in $T$. Let $E'=E\cup F'$. Then,
we preprocess $(G',F')$ using $\cD$. The total time is $O(n+k)+t_{p}(n,k)$.
Note that $F'$ has one connected component.

Given the operation $\instree(e)$, $\deltree(e)$, $\insnon(e)$,
$\delnon(e)$ for updating $G$ and $F$, it is clear how to update
$G'$ and $F'$ such that the invariant is maintained: for each $T\in F$,
there is an edge $(r,u_{T})\in F'$ where $u_{T}\in V(T)$ using a
constant number of update operations to $G'$ and $F'$ while making
sure that $F'$ has at most $O(1)$ many connected components.

Given a query $\textsf{find\_edge}(T)$ where $T\in F$, we give the
following update to $G'$ and $F'$: $\deltree(r,u_{T})$, $\textsf{find\_edge}(T)$,
$\instree(r,u_{T})$. Similarly, given a query $\textsf{find\_edge}(T_{1},T_{2})$
where $T_{1},T_{2}\in F$, we update $G'$ and $F'$ as follows: $\deltree(r,u_{T_{1}})$,
$\deltree(r,u_{T_{2}})$, $\textsf{find\_edge}(T_{1},T_{2})$, $\instree(r,u_{T_{1}})$,
$\instree(r,u_{T_{2}})$. It is clear also answers from $\cD$ maintain
$(G',F')$ are correct and takes at most $O(1)\times t_{u}(n,k)=O(t_{u}(n,k))$
time. This concludes the reduction.
\end{proof}

\subsubsection{List Intersection Oracle}

In this section, we show the main data structure which is used inside
2-dim ET tree. For any list of numbers $L=(e_{1},\dots,e_{s})$, we
denote $|L|=s$ the length of the list. We say that a number $e$
\emph{occurs} in $L$ for $c$ times if $|\{j\mid e=e_{j}\in L\}|=c$.
Let $\cL=\{L_{i}\}_{i}$ be a collection of lists. A number $e$ occurs
in $\cL$ for $c$ times if $\sum_{L\in\cL}|\{j\mid e=e_{j}\in L\}|=c$. 
\begin{defn}
[List intersection oracle]\emph{List intersection oracle }is a data
structure that preprocesses a collection of lists of numbers $\cL=\{L_{i}\}_{i}$
as inputs and handles the following operations:
\begin{itemize}
\item $\textsf{create}(e):$ return a pointer to a new list $L=\{e\}$,
given that $e$ is a number.
\item $\textsf{merge}(L_{1},L_{2})$: concatenate two lists into a new list
$L=L_{1}L_{2}$, given pointers to $L_{1},L_{2}\in\cL$.
\item $\textsf{split}(L,i)$: split $L=(e_{1},\dots,e_{s})$ into two lists
$L_{1}$ and $L_{2}$ where $L_{1}=(e_{1},\dots,e_{i})$ and $L_{2}=(e_{i+1},\dots,e_{s})$,
given a pointer to $L\in\cL$ and $1\le i\le s$.
\item $\textsf{intersection}(L_{1},L_{2})$: return an element $e$ where
$e\in L_{1}\cap L_{2}$ if $v$ exists (otherwise return $\bot$),
given two pointers two $L_{1},L_{2}\in\cL$.
\end{itemize}
\end{defn}
\begin{thm}
For any parameter $k$, there is a list intersection oracle  $\cD$
that preprocess a collection of lists $\cL=\{L_{1}\}_{i}$ as inputs
in time $O(k)$. Suppose that 1) the total length of lists in $\cL$
is bounded by $\sum_{i}|L_{i}|\le k$, 2) the number of lists in $\cL$
is bounded by some constant, and 3) each number occurs in $\cL$ at
most a constant number of times. Then $\cD$ can handle $\textsf{merge}$
and $\textsf{split}$ in $O(\sqrt{k}\log k)$ time, and can handle
$\textsf{create}$ and $\textsf{intersection}$ in $O(1)$ time.\label{thm:list intersection}
\end{thm}
We note again that all the ideas of this algorithm have been presented
already by Alstrup et al. \cite[Section 5]{AlstrupHLT05} and by Thorup
\cite[Section 5.3]{Thorup07mincut}. Since they did not explicitly
define a list intersection oracle as a problem, we include it for
completeness in \ref{sec:proof list intersection}.

\subsubsection{Non-tree Euler Tour}

The goal of this section is the following:
\begin{thm}
There is a 2-dim ET tree data structure $\cD$ that preprocesses any
$(G,F)$ in time $O(n+k)$ where $G$ is a graph with $n$ nodes and
$F$ is a forest in $G$, and the number of non-tree edges of $G$
is bounded by $k$. Moreover, all operations of $\cD$ can done in
$O(\sqrt{k}\log k)$ time as long as the number of connected component
in $F$ is bounded by a constant. \label{thm:2dim ET tree constant tree}
\end{thm}
By \ref{thm:2-dim constant tree}, this immediately implies \ref{thm:2dim ET tree}. 

The key to this result is the definition of \emph{non-tree Euler Tour
}of a tree, which is a slight modification of Euler Tour defined in
\cite{HenzingerK99}. While Euler tour is an order list of nodes in
a tree, non-tree Euler Tour is a list of non-tree edges incident to
nodes in a tree. The construction is almost the same. 

Throughout this section, there are only $O(1)$ number of trees $T\in F$.
For each tree $T\in F$, we root $T$ at arbitrary node $r_{T}$.
For each node $u\in T$, let $p_{T}(u)$ be the parent of $u$ in
$T$, and we fix an arbitrary ordering to edges incident to $u$.
\begin{defn}
[Non-tree Euler Tour]For any graph $G=(V,E)$, forest $F\subset E$,
and tree $T\in F$, a \emph{non-tree Euler tour of $T$} in $G$,
denoted by $\net_{G}(T)$, is an ordered list of non-tree edges in
$G$ that is incident to vertices in $T$. The precise ordering in
$\net_{G}(T)$ is defined by the output of $\textsf{ListNET}(G,F,r_{T})$
where $r_{T}$ is the root of $T$. See \ref{alg:net}. \emph{Non-tree
Euler tour of $F$ }in $G$, denoted by $\net_{G}(F)=\{\net_{G}(T)\}_{T\in F}$,
is a collection of the tours of trees $T\in F$. 
\end{defn}
\begin{algorithm}[h]
\begin{enumerate}
\item For each edge $(u,v)$ incident to $u$:

\begin{enumerate}
\item If $v=p_{T}(u)$, CONTINUE.
\item If $(u,v)\notin F$, then OUTPUT $e=\{u,v\}$.
\item If $(u,v)\in F,$ then $\textsf{ListNET}(G,F,v)$.
\end{enumerate}
\end{enumerate}
\caption{$\textsf{ListNET}(G,F,v)$ generates a list of non-tree edges of $G$
incident to vertices in a subtree in $F$ rooted at $v$. \label{alg:net}}
\end{algorithm}

Next, we observe some easy properties of non-tree Euler tour which
are similarly shown for Euler tour. The number of occurrence of an
edge $e$ in $\net_{G}(F)$ is $o(e)=\sum_{T\in F}$$|\{i\mid e=e_{i}\in\net_{G}(T)\}|$. 
\begin{prop}
Each non-tree edge of $G$ occurs in $\net_{G}(F)$ exactly two times.\label{fact:each occur twice}
\end{prop}

\begin{prop}
For any $T_{1},T_{2}\in F$, there is a non-tree edge $(u,v)$ where
$u\in V(T_{1})$ and $v\in V(T_{2})$ iff $\net_{G}(T_{1})\cap\net_{G}(T_{2})\neq\emptyset$.
\end{prop}

\begin{prop}
Let $G=(V,E)$ be a graph and $F\subset E$ be a forest in $G$, and
$|E-F|=$$k$. Then $\sum_{T\in F}|\net(T)|=2k$.\label{fact:tour short}
\end{prop}
Now, we describe how to use\emph{ $\net_{G}(F)$} for 2-dimenstion
ET tree. Let $\cO$ be list intersection oracle from \ref{thm:list intersection}
with parameter $2k$. By \ref{fact:tour short}, \ref{fact:each occur twice},
and the fact that $O(1)$ number of trees $T\in F$, the conditions
in \ref{thm:list intersection} are satisfied. And hence $\cO$ can
handle $\textsf{merge}$ and $\textsf{split}$ in $O(\sqrt{k}\log k)$
time, and can handle $\textsf{create}$ and $\textsf{intersection}$
in $O(1)$ time. We have the following:
\begin{prop}
Let $G=(V,E)$ be a graph and $F\subset E$ be a forest in $G$. Suppose
that there is a list intersection oracle $\cO$ that has preprocessed
$\net_{G}(F)=\{\net_{G}(T)\}_{T\in F}$. 
\begin{itemize}
\item For any $T_{1},T_{2}\in F$, the answer of $\textsf{find\_edge}(T_{1},T_{2})$
and be obtained from $\textsf{intersection}(\net_{G}(T_{1}),\net_{G}(T_{2}))$.
\item Let $G'$ and $F'$ be obtained from $G$ and $F$ by either inserting
or deleting one tree edge or non-tree edge. Then $\net_{G'}(F')$
can be obtained from by a constant number of $\textsf{create}$, $\textsf{merge}$
and $\textsf{split}$ on $\cD$ .\label{thm:constant splitting}
\end{itemize}
\end{prop}
Now, we can prove \ref{thm:2dim ET tree constant tree}:
\begin{proof}
[Proof of \ref{thm:2dim ET tree constant tree}]Given $G$ and $F$,
we can constructing the non-tree Euler tour $\net_{G}(F)$ in time
$O(n+k)$. Then we give $\net_{G}(F)$ to the list intersection oracle
$\cO$ as inputs, and $\cO$ takes $O(k)$ times to preprocess. The
worst-case time of any operation handled by $\cO$ is $O(\sqrt{k}\log k)$
by \ref{thm:list intersection}. By \ref{thm:constant splitting},
we can do $\instree(e)$, $\deltree(e)$, $\insnon(e)$, $\delnon(e)$,
and $\textsf{find\_edge}(T_{1},T_{2})$ in time $O(\sqrt{k}\log k)$.
To answer $\textsf{find\_edge}(T)$, we can just query $\textsf{find\_edge}(T,T')$
for all $T'\neq T$.\end{proof}

%% file: sketches_det.tex
\subsection{Cut Recovery Tree\label{sec:linear sketches}}

The following data structure is useful for our Monte Carlo dynamic
$\st$ algorithm.
\begin{defn}
[Cut Recovery Tree]\label{def:cut-recovery-tree}Cut recovery tree
is a data structure that preprocesses a graph $G=(V,E)$, a forest
$F\subseteq E$ and a parameter $k$ as inputs and handles the following
operations:
\begin{itemize}
\item $\instree(e)$, $\textsf{\ensuremath{\deltree}}(e)$, $\insnon(e)$
and $\delnon(e)$: same definitions as in \ref{def:aug ET tree}.
\item $\cutset k(T)$: return all edges in $\partial_{G}(V(T))$, given
a pointer to\emph{ }$T$ which is a tree in $F$ and given that $\delta_{G}(V(T))\le k$.
\end{itemize}
\end{defn}
\begin{thm}
\label{thm:cut recovery tree}There is a deterministic data structure
for cut recovery tree denoted by $\cD$ that are given $G,F$ and
a parameter $k$ as inputs where $G=(V,E)$ is a graph with $n$ nodes
and $m$ initial edges, and $F\subseteq E$ is a forest in $G$. $\cD$
can preprocess the inputs in time $O(m\beta\log n)$ time where $\beta=2^{O(\log\log n)^{3}}=n^{o(1)}$.
And then $\cD$ takes 
\begin{itemize}
\item $O(k\beta\log^{2}n)$ time to handle $\instree$ and $\deltree$,
\item $O(\beta\log^{2}n)$ time to handle $\insnon$ and $\textsf{\ensuremath{\delnon}}$,
and 
\item $O(k\beta\log^{2}n)$ time to answer $\cutset k$.
\end{itemize}
\end{thm}
This theorem is proved by maintaining a version of ET tree where each
node of the ET tree contains a $k$-sparse recovery sketch. This approach
is already used by Kapron, King and Mountjoy \cite{KapronKM13} to
obtain their data structure called \emph{cut set data structure}.
But in \cite{KapronKM13}, each node of the ET tree contains instead
contains a $\ell_{0}$-sampling sketch (which is essentially a $1$-sparse
recovery sketch combining with sampling). Here, we generalize the
same idea to $k$-sparse recovery sketch. 

A nice aspect about both cut set data structure and cut recovery tree
is that they are obtained by combining ideas from two different communities:
1) linear sketches from streaming literature, see \ref{sec:sparse recovery},
and 2) a version of ET tree which is a classic dynamic data structure,
see \ref{sec:d word ET tree}. We show how to combine the two and
prove \ref{thm:cut recovery tree} in \ref{sec:combining cut recovery tree}.

\subsubsection{Sparse Recovery\label{sec:sparse recovery}}

In this section, we show the following about sparse recovery. We say
that a vector $v$ is \emph{$k$-sparse} if there are at most $k$
non-zero entries in $v$. 
\begin{thm}
\label{thm:sparse recovery}There are algorithm $\cA_{recover}$,
and algorithm $\cA_{construct}$ such that, given any number $n$
and $k$, they can do the following. Let $\beta_{1}=2^{O(\log\log n)^{3}}=n^{o(1)}$,
$\beta_{2}=O(\polylog n)$, $d=O(k\beta_{1}\log n)$ and $s=O(\beta_{1}\log n)$.
\begin{itemize}
\item There is a matrix $\Phi\in\{0,1\}^{d\times n}$ where each column
of $\Phi$ is $s$-sparse that can be computed by $\cA_{construct}$
in the following sense: given an index $i\in[n]$, $\cA_{construct}$
can return a list of non-zero entries in the $i$-th of column of
$\Phi$, i.e. $(j,\Phi_{j,i})$ where $\Phi_{j,i}\neq0$ for each
$j\in[d]$, in time $s\beta_{2}$ time.
\item For any $k$-sparse vectors $v\in[-M,M]^{n}$ where $M=\poly(n)$,
given $\Phi v$, $\cA_{recover}$ can return a list of non-zero entries
in $v$, i.e. the list of $(i,v_{i})$ where $v_{i}\neq0$ for each
$i\in[n]$, in time $O(d\log n)$.
\end{itemize}
\end{thm}
This theorem is essentially proved by Berinde et al. \cite[Section 5]{BerindeGIKS08}.
However, in \cite{BerindeGIKS08}, they did not explicitly state how
the matrix $\Phi$ can be constructed and the running time of the
construction. Here, we make it explicit by defining the algorithm
$\cA_{construct}$. To do this, we observe that the construction of
\cite{BerindeGIKS08} uses an explicit \emph{unbalanced expander}
from Capalbo et al. \cite{CapalboRVW02}. The explicitness of their
expanders directly allows us to obtain $\cA_{construct}$. We also
give the proof of this theorem for completeness in \ref{sec:proof sparse recovery}.

\subsubsection{$d$-word ET Tree\label{sec:d word ET tree}}

Next, we need the following data structure which is a variant of ET
tree from Henzinger and King's paper \cite{HenzingerK99} so that
it can maintain a sum of a large vector associated with each node.
This data structure is also used by Kapron King and Mountjoy \cite{KapronKM13}
to obtain a data structure called \emph{cut set data structure}. Since
the definition is implicit there, we give a precise definition here:
\begin{defn}
[$d$-word ET tree]For any parameter $d$, a $d$-word ET tree is
a data structure that preprocesses a forest $F$ with a set $V$ of
nodes and a vector $x_{u}\in[-M,M]^{d}$ where $M=\poly(|V|)$ for
each $u\in V$ as inputs and handles the following operations:
\begin{itemize}
\item $\instree(e)$: add edge $e$ into $F$, given that $e\notin F$ and
adding does not cause a cycle in $F$.
\item $\deltree(e)$: remove edge $e$ from $F$, given that $e\in F$.
\item $\textsf{sum\_vec}(T)$: return $\sum_{u\in T}x_{u}$, given that
$T$ is a tree in $F$.
\item $\textsf{update\_vec}(u,i,\alpha)$: set $x_{u}[i]=\alpha$, given
that $u\in V$, $i\in[d]$ and $\alpha\in[-M,M]$.
\end{itemize}
\end{defn}
\begin{thm}
\label{thm:d-word ET tree}For any parameter $d$, there is an algorithm
for $d$-word ET tree that preprocesses a forest $F$ with a set $V$
of $n$ nodes and a set of vectors $x=\{x_{u}\}_{u\in V}$ which takes 
\begin{itemize}
\item $O(nnz(x)\log n)$ preprocessing time where $nnz(x)=\sum_{u\in V}|x_{u}|_{0}$
is the total number of non-zero entries of each vector in $\{x_{u}\}_{u\in V}$,
\item $O(d\log n)$ time to handle $\instree$ and $\deltree$, 
\item $O(d)$ time to return $\textsf{sum\_vec}$, and 
\item $O(\log n)$ time to handle $\textsf{update\_vec}(u,i,\alpha)$.
\end{itemize}
\end{thm}
It is important that preprocessing time is $\tilde{O}(nnz(x))$ to
obtain a cut recovery tree with fast preprocessing time. Note that
the total number of entries of $x_{u}$ over all $u\in V$ is $n\times d$.
But $nnz(x)$ can be much smaller. For completeness, we give a proof
in \ref{sec:proof d word ET}.

\subsubsection{Put Them Together\label{sec:combining cut recovery tree}}

\ref{thm:cut recovery tree} is obtained by combining the above two
results:
\begin{proof}
[Proof of \ref{thm:cut recovery tree}]Let $\Phi^{*}$ be the matrix
of size $d\times\binom{n}{2}$ from \ref{thm:sparse recovery} with
parameter $\binom{n}{2}$ and $k$ such that each column of $\Phi^{*}$
is $s$-sparse, where $d=O(k\beta_{1}\log n)$, $s=O(\beta_{1}\log n)$
and $\beta_{1}=2^{O(\log\log n)^{3}}$. We will not actually construct
the whole $\Phi^{*}$ at the beginning. Instead, let $\Phi$ be a
zero matrix of size $d\times\binom{n}{2}$. We will gradually compute
some entries of $\Phi^{*}$and copy it to the corresponding entry
in $\Phi$. For each $(u,v)\in\binom{n}{2}$, let $\Phi_{(u,v)},$$\Phi_{(u,v)}^{*}\in\{0,1\}^{d}$
denote the $(u,v)$-th column of $\Phi$ and $\Phi^{*}$ respectively.

Let $E_{0}$ be the set of initial edges of $G$. At any time, let
$E_{I}$ be the set of edges have been inserted into $G$ and let
$B\in\{0,1\}^{\binom{n}{2}\times n}$ be the incidence matrix of $G$.
We have the following invariants: 1) for each $(u,v)\in E_{0}\cup E_{I}$,
$\Phi{}_{(u,v)}=\Phi^{*}{}_{(u,v)}$, otherwise $\Phi_{(u,v)}=0$,
and 2) for each $u\in V$, we have a vector $x_{u}=\Phi B_{u}$ of
dimension $d$, where $B_{u}\in\{0,1\}^{\binom{n}{2}}$ be the $u$-th
column of $B$. It follows from the invariants that, for any $S\subset V$,
we have $\Phi(\sum_{u\in S}B_{u})=\Phi^{*}(\sum_{u\in S}B_{u})$ because
if the $(u,v)$-th entry of $\sum_{u\in S}B_{u}$ is non-zero, then
$(u,v)\in E_{0}\cup E_{I}$.

To preprocess, we do the following. We assume that it takes constant
time to allocate space for a zero matrix $\Phi$. Then, to validate
the invariants, we first set $\Phi{}_{(u,v)}=\Phi^{*}{}_{(u,v)}$
for each initial edge $(u,v)\in E_{0}$. This takes in total $ms\beta_{2}$
time because all non-zero entries of each column of $\Phi^{*}$ can
be identified in $s\beta_{2}$ time by \ref{thm:sparse recovery}
where $\beta_{2}=O(\polylog n)$. Second, we compute $x_{u}=\Phi B_{u}$.
Since each column of $\Phi^{*}$ is $s$-sparse, so is $\Phi$. As,
$B_{u}$ is $\deg(u)$-sparse, this also takes in total $\sum_{u\in V}O(\deg(u)s)=O(ms)$
time. Finally, let $\cE$ be a $d$-word ET tree from \ref{thm:d-word ET tree}.
We use $\cE$ to preprocess $F$ and $x=\{x_{u}\}_{u\in V}$ taking
time $O(nnz(x)\log n)=O(\sum_{u\in V}\deg(u)\log n)=O(m\log n)$.
The total preprocessing time is $O(ms\beta_{2}+ms+m\log n)=O(m\beta_{1}\beta_{2}\log n)$.

To handle $\instree$ and $\deltree$, we just update $\cE$ accordingly.
So, by \ref{thm:d-word ET tree}, the tree-edge update time is $O(d\log n)=O(k\beta_{1}\log^{2}n)$.
To handle $\insnon$ and $\textsf{\ensuremath{\delnon}}$, if $e=(u,v)$
is inserted and $e\notin E_{0}\cup E_{I}$, then, to maintain the
first invariant, we set $\Phi{}_{(u,v)}=\Phi^{*}{}_{(u,v)}$ in time
$s\beta_{2}$. Next, to maintain the second invariant, we want to
maintain that $x_{u}=\Phi B_{u}$ for all $u\in V$. When we update
non-tree edge $e=(u,v)$, one entry of both $B_{u}$ and $B_{v}$
are updated. Since each column of $\Phi$ is s-sparse, $s$ entries
of $x_{u}$ and $x_{v}$ are updated. Each call for $\textsf{update\_vec}$
takes $O(\log n)$ time, so this takes $O(s\log n)$. In total, the
non-tree-edge update time is $O(s\beta_{2}+s\log n)=O(\beta_{1}\beta_{2}\log n)$.

To answer $\cutset k(T)$, we have that for any fixed $T$, we can
construct $\partial_{G}(V(T))$ from $\sum_{u\in T}x_{u}$ in time
$O(d\log n)$. Because 1) by calling $\textsf{sum\_vec}(T)$, we get
$\sum_{u\in T}x_{u}$ in time $O(d)$, 2) $\sum_{u\in T}x_{u}=\sum_{u\in T}\Phi B_{u}=\Phi(\sum_{u\in T}B_{u})=\Phi^{*}(\sum_{u\in S}B_{u})$,
3) $\partial_{G}(V(T))$ corresponds to the non-zero entries of $\sum_{u\in T}B_{u}$
by \ref{thm:incidence matrix}, and 4) by \ref{thm:sparse recovery},
for any $k$-sparse $v\in[-poly(n),poly(n)]^{\binom{n}{2}}$, we can
construct $v$ from from $\Phi^{*}v$ in time $O(d\log n)$. Since
$\delta_{G}(V(T))\le k$, we can construct $\partial_{G}(V(T))$ in
time $O(d\log n)=O(k\beta_{1}\log^{2}n)$. 

By setting the parameter $\beta=\beta_{1}\times\beta_{2}$, this completes
the proof.\end{proof}
\begin{rem}
We note that $\beta$ is a factor depending on product of the degree
and explicitness of the best known explicit construction of unbalanced
expanders. Therefore, a better construction of expanders would immediately
speed up our data structure.\end{rem}

%% file: reduction.tex
\subsection{Classic Reductions\label{sec:reductions}}

In this section, we show that, roughly, to get a dynamic $\st$ with
$O(n^{1/2-\epsilon+o(1)})$ worst-case update time for an infinite-length
update sequence, it is enough to show a dynamic $\st$ with the same
update time and $O(n^{1+o(1)})$ preprocessing time where the underlying
graph has maximum degree $3$ and the length of the update sequence
is just $n^{1/2+\epsilon}$. 

This reduction also helps the algorithms in \cite{KapronKM13,GibbKKT15}
that can handle only polynomial-length update sequences to be able
handle infinite-length sequences.

Recall \ref{def:adversary} and \ref{def:running time} for formal
definitions of ``working against adversaries''. We also need the
following definition:
\begin{defn}
A dynamic algorithm runs \emph{on $\Delta$-bounded degree graphs
}if throughout the whole sequence of updates, all nodes in the underlying
graph always has maximum degree at most $\Delta$.\end{defn}
\begin{thm}
For any $n$, $p_{c}$, $p_{t}$ and $T$, suppose there is a dynamic
$\st$ algorithm $\cA$ for $3$-bounded-degree graphs with $n$ nodes
and $\Theta(n)$ edges that, for the first $T$ updates, works correctly
against adaptive adversaries with probability $1-p_{c}$ and $\cA$
has preprocessing time $t_{p}(n,p_{c},p_{t},T)$ and worst-case update
time $t_{u}(n,p_{c},p_{t},T)$ with probability $1-p_{t}$. Suppose
that for any $n_{1},n_{2}$, $t_{p}(n_{1},p_{c},p_{t},T)+t_{p}(n_{2},p_{c},p_{t},T)\le t_{p}(n_{1}+n_{2},p_{c},p_{t},T)$,
then there is a dynamic $\st$ algorithm $\cB$ for any graph with
$n$ nodes that works correctly against adaptive adversaries with
probability $1-p_{c}\cdot\poly(n)$ and $\cB$ has preprocessing time
$O(t_{p}(O(m),p_{c},p_{t},T)\log n)$ where $m$ is the number of
initial edges and worst-case update time $O(t_{u}(O(n),p_{c},p_{t},T)\log n+t_{p}(O(n),p_{c},p_{t},T)/T+\log^{2}n)$
with probability $1-p_{t}\cdot\poly(n)$.\label{thm:classic reduc}
\end{thm}
Note that if $\cA$ works correctly with certainty ($p_{c}=0$) or
with high probability, then so does $\cB$. Similarly,if $\cA$ guarantees
the worst-case update time with certainty ($p_{t}=0$) or with high
probability, then so does $\cB$. 

For convenience, we state the reduction that follows immediately from
\ref{thm:classic reduc} and will be used in \ref{sec:monte_carlo}
to obtain our Monte Carlo algorithm.
\begin{cor}
For any constant $\epsilon$, suppose there is a dynamic $\st$ algorithm
$\cA$ for $3$-bounded-degree graphs with $n$ nodes and $\Theta(n)$
edges that, for the first $T$ updates, works correctly against adaptive
adversaries with probability $1-p_{c}$ and $\cA$ has preprocessing
time $t_{p}(n,p_{c})=O(n^{1+o(1)}\log\frac{1}{p_{c}})$ and worst-case
update time $t_{u}(n,p_{c},T)$ where $t_{u}(n,p_{c},n^{1/2+\epsilon})=O(n^{1/2-\epsilon+o(1)}\log\frac{1}{p_{c}})$.
Then there is a dynamic $\st$ algorithm $\cB$ for any graph with
$n$ nodes that works correctly against adaptive adversaries with
high probability and $\cB$ has preprocessing time $O(m^{1+o(1)})$
where $m$ is the number of initial edges and worst-case update time
$O(n^{1/2-\epsilon+o(1)})$.\label{thm:classic reduc mc}\end{cor}
\begin{proof}
We note that the function $t_{p}(n,p_{c})$ does not depend on $T$.
We set $T=n^{1/2+\epsilon}$. By \ref{thm:classic reduc}, there is
dynamic $\st$ algorithm $\cB$ for any graph with $n$ nodes that
works correctly against adaptive adversaries with probability $1-p_{c}\cdot\poly(n)$
and $\cB$ has preprocessing time $O(t_{p}(O(m),p_{c})\log n)=O(m^{1+o(1)}\log\frac{1}{p_{c}})$
where $m$ is the number of initial edges and worst-case update time
$O(t_{u}(O(n),p_{c},T)\log n+t_{p}(O(n),p_{c})/T+\log^{2}n)=O(n^{1/2-\epsilon+o(1)}\log\frac{1}{p_{c}}+\frac{n^{1+o(1)}\log\frac{1}{p_{c}}}{n^{1/2+\epsilon}}+\log^{2}n)=O(n^{1/2-\epsilon+o(1)}\log\frac{1}{p_{c}})$.
By choosing $p_{c}=1/n^{k}$ where $k$ is an arbitrarily large constant,
we are done. 
\end{proof}
Next, we state the similar reduction that will be used in \ref{sec:las_vegas}
to get our Las Vegas algorithm. 
\begin{cor}
For any $n$ and constant $\epsilon$, suppose there is a dynamic
$\st$ algorithm $\cA$ for $3$-bounded-degree graphs with $n$ nodes
and $\Theta(n)$ edges that, for the first $T$ updates, works correctly
against adaptive adversaries with certainty and $\cA$ has, with probability
$1-p_{t}$, preprocessing time $t_{p}(n,p_{t})=O(n^{1+o(1)}\log\frac{1}{p_{t}})$
and worst-case update time $t_{u}(n,p_{t},T)$ where $t_{u}(n,p_{t},n^{1/2+\epsilon})=O(n^{1/2-\epsilon+o(1)}\log\frac{1}{p_{t}})$.
Then there is a dynamic $\st$ algorithm $\cB$ for any graph with
$n$ nodes that works correctly against adaptive adversaries with
certainty and $\cB$ has preprocessing time $O(m^{1+o(1)})$ where
$m$ is the number of initial edges and worst-case update time $O(n^{1/2-\epsilon+o(1)})$
with high probability.\label{thm:classic reduc lv}
\end{cor}
\ref{thm:classic reduc lv} has almost the same proof as \ref{thm:classic reduc mc}
so we omit it. The rest of this section is devoted for proving \ref{thm:classic reduc}.

\subsubsection{Overview of the Sequence of Reductions}

In this section, we will state the sequence of reductions that we
will use for proving \ref{thm:classic reduc}. For each reduction,
we give a high level overview. Then we show how to combine them. For
convenience, we define dynamic $\st$ algorithms with parameters as
follows:
\begin{defn}
$(n,p_{c},p_{t},T)$-algorithms are dynamic $\st$ algorithms for
a graph with $n$ nodes that, for the first $T$ updates, works correctly
against adaptive adversaries with probability $1-p_{c}$ and the preprocessing
and update time is guaranteed with probability $1-p_{t}$. $(n,p_{c},p_{t},\infty)$-algorithms
give the same guarantee for correctness and running time for any length
of updates.
\end{defn}
The following property called ``stability'' of the algorithm is
useful for us. Note that this definition is not the same definition
as in \cite[Definition 3.3.1]{EppsteinGIN97} as we will discuss again
before the statement of \ref{lem:reduc sparse}.
\begin{defn}
[Stability]\label{def:stable}A dynamic $\st$ algorithm $\cA$ is
\emph{stable }if the spanning forest maintained by $\cA$ never changes
except when a tree-edge $e$ is deleted. If a tree-edge $e$ is deleted
from a spanning forest $F$, then $\cA$ either do not update $F$
or add a replacement edge $e'$ into $F$. That is, let $F'$ be a
new spanning forest updated by $\cA$. Then, either $F'=F\setminus\{e\}$
or $F'=F\cup\{e'\}\setminus\{e\}$.\footnote{All known dynamic spanning tree algorithms in the literature \cite{Frederickson85,KapronKM13,HolmLT01,HenzingerK99,Kejlberg-Rasmussen16}
are stable according to this definition. }
\end{defn}
An $(n,p_{c},p_{t},T)$-algorithm which is stable is called a stable
$(n,p_{c},p_{t},T)$-algorithm. The first reduction shows that we
can assume that the algorithm is stable. The technique for proving
this is implicit in \cite[Section 3.1]{HenzingerK99}.
\begin{lem}
Suppose there is an $(n,p_{c},p_{t},T)$-algorithm $\cA$ for any
graph with preprocessing time $t_{p}$ and update time $t_{u}$. Then
there is a stable $(n,O(Tp_{c}),p_{t},T)$-algorithm $\cB$ for the
same graph with preprocessing time $O(t_{p}+m_{0})$ where $m_{0}$
is the number of initial edges, and update time $O(t_{u}\log n+\log^{2}n)$.\label{lem:reduc stable}
\end{lem}
Next, the reduction below shows how to get an algorithm that can handle
long update sequences from an algorithm that handles only short sequences.
The technique is standard, i.e. dividing the long update sequence
into phases, and maintaining two concurrent data structures.
\begin{lem}
Suppose there is a stable $(n,p_{c},p_{t},T)$-algorithm $\cA$ for
a 3-bounded-degree graph where the number of edges is always between
$n/4$ and $2n$ with preprocessing time $t_{p}(n,p_{c},p_{t},T)$
and update time $t_{u}(n,p_{c},p_{t},T)$. Then, for any $U\ge T$,
there is a stable $(n,O(Up_{c}),O(p_{t}),U)$-algorithm $\cB$ for
a 3-bounded-degree graph where the number of edges is always between
$n/4$ and $2n$ with preprocessing time $O(t_{p}(n,p_{c},p_{t},T))$
and update time $O(t_{u}(n,p_{c},p_{t},T)+t_{p}(n,p_{c},p_{t},T)/T+\log^{2}n)$.\label{lem:reduc poly length}
\end{lem}
The next reduction shows how to assume that the graph is 3-bounded-degree.
The reduction is implicit in \cite{Frederickson85,Kejlberg-Rasmussen16},
so the technique is standard. Given a graph $G$, we maintain a graph
$G'$ with max degree $3$ by ``splitting'' each node in $G$ into
a path in $G'$ of length equals to its degree. We assign weight 0
to edges in those paths, and 1 to original edges. Then we maintain
a 2-weight minimum spanning forest in $G'$. To do this, we use the
reduction from $k$-weight minimum spanning forest to spanning forest
from \cite{HenzingerK99}.
\begin{lem}
Suppose there is a stable $(n,p_{c},p_{t},T)$-algorithm $\cA$ for
a 3-bounded-degree graph where the number of edges is always between
$n/4$ and $2n$, and $\cA$ has preprocessing time $t_{p}(n,p_{c},p_{t},T)$
and update time $t_{u}(n,p_{c},p_{t},T)$. Then there is a stable
$(n,O(Tp_{c}),O(p_{t}),T)$-algorithm $\cB$ for any graph with where
the number of edges is always between $n/4$ and $2n$ with preprocessing
time $O(t_{p}(O(n),p_{c},p_{t},T))$ and update time $O(t_{u}(O(n),p_{c},p_{t},T)+\log^{2}n)$.\label{lem:reduc bound degree}
\end{lem}
The following shows how to obtain an algorithm whose running time
depends only on the number of edges, and not depends on the number
$n$ of nodes. It is useful when the number of edges is \emph{much
less than} the number of nodes. The technique for this reduction is
also standard, i.e. dividing the long update sequence into phases,
and maintaining two concurrent data structures.
\begin{lem}
Suppose there is a stable $(n,p_{c},p_{t},T)$-algorithm $\cA$ for
any graph where the number of edges is always between $n/4$ and $2n$,
and $\cA$ has preprocessing time $t_{p}(n,p_{c},p_{t},T)$ and update
time $t_{u}(n,p_{c},p_{t},T)$. Then there is a stable $(n,O(Tp_{c}),O(p_{t}),T)$-algorithm
$\cB$ for any graph, initially represented as a list of its edges\footnote{Note that it is possible that the preprocessing time of $\cB$ is
$O(t_{p}(O(m_{0}),p_{c},p_{t}))=o(n)$, so the algorithm $\cB$ needs
that the initial graph \emph{$G$ is represented as a list of its
edges }i.e. $G$ is represented as a list $((u_{1},v_{1}),(u_{2},v_{2})\dots,(u_{m},v_{m}))$
so that all edges of $G$ can be listed in time $O(m)$ even when
$m=o(n)$.}, with $n$ nodes where the number of edges is always at most $2n$,
and $\cB$ has preprocessing time $O(t_{p}(O(m_{0}),p_{c},p_{t}))$
where $m_{0}$ is a number of initial edges, and update time $O(t_{u}(O(m),p_{c},p_{t})+\log^{2}n)$
where $m$ is a number of edges when update. \label{lem:reduc lower bound edge}
\end{lem}
The next reduction shows, given an algorithm with running time depending
only on the number of edges, how to get an algorithm with update time
depending only on the number of nodes in even dense graphs. This reduction
uses the general sparsification technique introduced in \cite{EppsteinGIN97}.
However, it is not immediate that we can apply their framework because
their notion of \emph{stability }is too strong for us. 

To be more precise, a\emph{ stable certificate} according to \cite[Definition 3.3.1]{EppsteinGIN97}
must be uniquely determined once the graph is fixed. However, this
is not the case for a spanning forest maintained by a randomized algorithm,
because its edges may depend on the random choices of the algorithm.
So we need to use the weaker notion of stability as defined in \ref{def:stable}
and prove that this framework still goes through. 
\begin{lem}
Suppose there is a stable $(n,p_{c},p_{t},T)$-algorithm $\cA$ for
a graph, initially represented as a list of its edges, with $n$ nodes
where the number of edges is always at most $2n$, and $\cA$ has
preprocessing time $t_{p}(m_{0},p_{c},p_{t},T)$ where $m_{0}$ is
a number of initial edges and update time $t_{u}(m,p_{c},p_{t},T)$
where $m$ is a number of edges when update. Suppose that for any
$n_{1},n_{2}$, $t_{p}(m_{1},p_{c},p_{t},T)+t_{p}(m_{2},p_{c},p_{t},T)\le t_{p}(m_{1}+m_{2},p_{c},p_{t},T)$,
then there is a stable $(n,O(np_{c}),O(np_{t}),T)$-algorithm $\cB$
for any graph with $n$ nodes with preprocessing time $O(t_{p}(m_{0},p_{c},p_{t},T)\log n)$
where $m_{0}$ is a number of initial edges and update time $O(t_{u}(O(n),p_{c},p_{t},T))$.\label{lem:reduc sparse}
\end{lem}
Observe that many reductions above, given an $(n,p_{c},p_{t},T)$-algorithm
$\cA$, we obtain the $(n,p'_{c},p'_{t},T')$-algorithm $\cB$ where
$p'_{c}=\Omega(T'p_{c})$. That is, the failure probability for correctness
is always increased by a factor of $T'$, which is the length of update
sequence that the resulting algorithm $\cB$ can handle. The following
reduction shows that, if $T\ge n^{2}$, how to obtain an algorithm
that can handle update sequence of \emph{any length}, while the failure
probability does not increase too much. This kind of reduction has
been shown for cut/spectral sparsifiers in \cite{AbrahamDKKP16}.
However, the reduction from \cite{AbrahamDKKP16} does not work because
the structure of spanning forests is more restricted than sparsifiers. 
\begin{lem}
Suppose there is a stable $(n,p_{c},p_{t},T)$-algorithm $\cA$ for
any graph with $n$ nodes with preprocessing time $t_{p}(m_{0},p,p_{t},T)$
where $m_{0}$ is a number of initial edges and update time $t_{u}(n,p_{c},p_{t},T)$.
If $T\ge n^{2}$, then there is an $(n,O(Tp_{c}),O(p_{t}),\infty)$-algorithm
$\cB$ for any graph with $n$ nodes with preprocessing time $O(t_{p}(m_{0},p_{c},p_{t},T))$
where $m_{0}$ is a number of initial edges and update time $O(t_{u}(n,p_{c},p_{t},T))$.\label{lem:reduc infty length}\end{lem}
\begin{rem}
The spanning forest maintained by the algorithm $\cB$ in \ref{lem:reduc infty length}
is not stable.
\end{rem}
With all reductions above, we conclude the following which is equivalent
to \ref{thm:classic reduc}.
\begin{thm}
[Restatement of \ref{thm:classic reduc}]Suppose there is an $(n,p_{c},p_{t},T)$-algorithm
$\cA$ for a 3-bounded-degree graph where the number of edges is always
between $n/2$ and $2n$, and $\cA$ has preprocessing time $t_{p}(n,p_{c},p_{t},T)$
and update time $t_{u}(n,p_{c},p_{t},T)$. Suppose that for any $n_{1},n_{2}$,
$t_{p}(n_{1},p_{c},p_{t},T)+t_{p}(n_{2},p_{c},p_{t},T)\le t_{p}(n_{1}+n_{2},p_{c},p_{t},T)$,
then there is an $(n,\poly(n)\cdot p_{c},\poly(n)\cdot p_{t},\infty)$-algorithm
$\cB$ for any graph with $n$ nodes with preprocessing time $O(t_{p}(O(m_{0}),p_{c},p_{t},T)\log n)$
where $m_{0}$ is a number of initial edges and update time $O(t_{u}(O(n),p_{c},p_{t},T)\log n+t_{p}(O(n),p_{c},p_{t},T)/T+\log^{2}n)$.
Moreover, if $\cA$ is stable, then $\cB$ has update time $O(t_{u}(O(n),p_{c},p_{t},T)+t_{p}(O(n),p_{c},p_{t},T)/T+\log^{2}n)$.\end{thm}
\begin{proof}
By \ref{lem:reduc stable}, we get a stable algorithm. That is, there
is a stable $(n,O(Tp_{c}),p_{t},T)$-algorithm $\cA_{0}$ for a 3-bounded-degree
graph where the number of edges is always between $n/2$ and $2n$
with preprocessing time $O(t_{p}(n,p_{c},p_{t},T))$ and update time
$O(t_{u}(O(n),p_{c},p_{t},T)\log n+\log^{2}n)$.

Let $U=n^{2}.$ If $T\le U$, then by \ref{lem:reduc poly length},
we get an algorithm that can handle long sequences of updates. That
is, there is a stable $(n,\poly(nU)\cdot p_{c},\poly(nU)\cdot p_{t},U)$-algorithm
$\cA_{1}$ for a 3-bounded-degree graph where the number of edges
is always between $n/2$ and $2n$ with preprocessing time $O(t_{p}(n,p_{c},p_{t},T))$
and update time $O(t_{u}(n,p_{c},p_{t},T)\log n+t_{p}(n,p_{c},p_{t},T)/T+\log^{2}n)$.
If $T>U$, then we can just treat $\cA$ as $\cA_{1}$. In either
cases, we obtain $\cA_{1}$.

By \ref{lem:reduc bound degree}, we remove the condition that the
graph has maximum degree 3. That is, there is a stable $(n,\poly(nU)\cdot p_{c},\poly(nU)\cdot p_{t},U)$-algorithm
$\cA_{2}$ for any graph with where the number of edges is always
between $n/2$ and $2n$ with preprocessing time $O(t_{p}(n,p_{c},p_{t},T)+n)=O(t_{p}(n,p_{c},p_{t},T))$
and update time $O(t_{u}(n,p_{c},p_{t},T)\log n+t_{p}(n,p_{c},p_{t},T)/T+\log^{2}n)$.

Next, by \ref{lem:reduc lower bound edge}, we remove the condition
that the number of edges is not too small. That is, there is a stable
$(n,\poly(nU)\cdot p_{c},\poly(nU)\cdot p_{t},U)$-algorithm $\cA_{3}$
for any graph with $n$ nodes where the number of edges is always
at most $2n$, and $\cA_{3}$ has preprocessing time $O(t_{p}(O(m_{0}),p_{c},p_{t},U))$
where $m_{0}$ is a number of initial edges, and update time $O(t_{u}(O(m),p_{c},p_{t},T)\log m+t_{p}(O(m),p_{c},p_{t},T)/T+\log^{2}m)$
where $m$ is a number of edges when update. 

Next, by \ref{lem:reduc sparse}, we remove the condition that a graph
is sparse. We obtain a stable $(n,\poly(nU)\cdot p_{c},\poly(nU)\cdot p_{t},U)$-algorithm
$\cA_{4}$ for any graph with $n$ nodes with preprocessing time $O(t_{p}(O(m_{0}),p_{c},p_{t},T)\log n)$
where $m_{0}$ is a number of initial edges and update time $O(t_{u}(O(n),p_{c},p_{t},T)\log n+t_{p}(O(n),p_{c},p_{t},T)/T+\log^{2}n)$.

Finally, by \ref{lem:reduc infty length}, we can handle update sequences
of any length. That is, as $U=n^{2}$, we obtain an $(n,\poly(nU)\cdot p_{c},\poly(nU)\cdot p_{t},\infty)$-algorithm
$\cA_{5}$ for any graph with $n$ nodes with preprocessing time $O(t_{p}(O(m_{0}),p_{c},p_{t},T)\log n)$
where $m_{0}$ is a number of initial edges and update time $O(t_{u}(O(n),p_{c},p_{t},T)\log n+t_{p}(O(n),p_{c},p_{t},T)/T+\log^{2}n)$.

Note that if $\cA$ is stable, then we do not need to use \ref{lem:reduc stable}
in the first step, and we have that $\cA_{5}$ has update time $O(t_{u}(O(n),p_{c},p_{t},T)+t_{p}(O(n),p_{c},p_{t},T)/T+\log^{2}n)$.\end{proof}

%% file: decomposition_new.tex
\section{Global Expansion Decomposition\label{sec:decomp alg}}

The backbone of both of our Monte Carlo and Las Vegas dynamic $\st$
algorithms is the algorithm for finding \emph{expansion decomposition}
of a graph, which decomposes a graph into components of expanders
and a sparse remaining part:
\begin{thm}
[Global Expansion Decomposition Algorithm]\label{thm:high-exp-decomp}There
is a randomized algorithm $\cA$ takes as inputs an undirected graph
$G=(V,E)$ with $n\ge2$ vertices and $m$ edges and an expansion
parameter $\alpha>0$, and a failure probability parameter $p$. Then,
in $O(m\gamma\log\frac{1}{p})$ time, $\cA$ outputs two graphs $G^{s}=(V,E^{s})$
and $G^{d}=(V,E^{d})$ with the following properties:
\begin{itemize}
\item $\{E^{s},E^{d}\}$ is a partition of $E$,
\item $G^{s}$ is sparse: $|E^{s}|\le\alpha n\gamma$, and
\item with probability $1-p$, each connected component $C$ of $G^{d}$
either is a singleton or has high expansion, i.e. $\phi(C)\ge\alpha$.
\end{itemize}
where the factor $\gamma=n^{O(\sqrt{\log\log n/\log n})}=n^{o(1)}$.
\end{thm}
Throughout this section, we call the output of the algorithm above
an \emph{expansion decomposition}. There are two main steps for proving
this theorem. First, we devise a near-linear time approximation algorithm
for a problem called \emph{most balanced sparse cut}. Second, we use
our most balanced sparse cut algorithm as a main procedure for constructing
the decomposition algorithm for \ref{thm:high-exp-decomp}.

Let us explain a high-level idea. Most balanced sparse cut problem
is to find a cut $S$ with largest number of nodes while $|S|\le|V-S|$
such that $\phi(S)<\alpha$ where $\alpha$ is an input. This problem
is closely related to sparsest cut and balanced cut problems (cf.
\cite{LeightonR99}). One way to approximately solve both problems
is by using the \emph{cut-matching game} framework of Khandekar, Rao
and Vazirani \cite{KhandekarRV09} together with known exact algorithms
for the maximum flow problem. We modify this framework (in a rather
straightforward way) so that (i) it gives a solution to the most balanced
sparse cut problem and (ii) we can use \emph{approximate} maximum
flow algorithms instead of the exact ones. By plugging in near-linear
time max flow algorithms \cite{Peng14,Sherman13,KelnerLOS14}, our
algorithm runs in near-linear time. See \ref{sec:sparsest cut} for
details.

Then we use an approximate most balanced sparse cut algorithm as a
main procedure for finding expansion decomposition. The main idea
is simply to find the cut and recurse on both sides. However, the
exact way for recursing is quite involved because we only have an
approximate guarantee about the cut. Because of this, there is a factor
of $\gamma=n^{O(\sqrt{\log\log n/\log n}})=n^{o(1)}$ in \ref{thm:high-exp-decomp}.
See \ref{sec:decomp recurse} for details.

\subsection{Most Balanced Sparse Cut\label{sec:sparsest cut}}

For any unweighted undirected graph $G=(V,E)$ with $n$ vertices
and $m$ edges, we say that a cut $S\subset V$ is \emph{$\alpha$-sparse}
if $\phi_{G}(S)<\alpha$. Let $\opt(G,\alpha)$ be the \emph{number
of vertices} of the largest $\alpha$-sparse cut $S^{*}$ where $|S^{*}|\le|V-S^{*}|$.
If $\phi(G)\ge\alpha$, then we define $\opt(G,\alpha)=0$. Note that,
if $\alpha\le\alpha'$, then $\opt(G,\alpha)\le\opt(G,\alpha')$.
Now, we define most balanced sparse cut formally:
\begin{defn}
[Most Balanced Sparse Cut]\label{def:most balanced sparse cut} For
any unweighted undirected graph $G=(V,E)$ and parameters $c_{size},c_{exp}$
and $\alpha$, a cut $S\subset V$ where $|S|\le n/2$ is a $(c_{size},c_{exp})$-approximate
most balanced $\alpha$-sparse cut if $\phi_{G}(S)<\alpha$ and $|S|\ge\opt(G,\alpha/c_{exp})/c_{size}$.
That is, the size of $S$ is as large as any $(\alpha/c_{exp})$-sparse
cut up to the $c_{size}$ factor.
\end{defn}

\begin{defn}
[Most Balanced Sparse Cut Algorithm]\label{def:most balanced sparse cut alg}A
$(c_{size},c_{exp})$-approximate algorithm $\cA$ for most balanced
sparse cut problem is given an unweighted undirected graph $G=(V,E)$
and a parameter $\alpha$, and then $\cA$ either 
\begin{itemize}
\item finds a $(c_{size},c_{exp})$-approximate most balanced $\alpha$-sparse
cut $S$, or
\item reports that $\phi(G)\ge\alpha/c_{exp}$. 
\end{itemize}
\end{defn}
The main theorem in this section is the following. 
\begin{thm}
\label{thm:most balanced sparse cut} There is a $(c_{size},c_{exp})$-approximate
most balanced sparse cut algorithm with running time $\tilde{O}(m)$
where $c_{size}=\Theta(\log^{2}n)$ and $c_{exp}=\Theta(\log^{3}n)$.
\end{thm}
As the proof is by a straightforward modification of the cut-matching
game framework by Khandekar, Rao and Vazirani \cite{KhandekarRV09}
for solving balanced cut problem, we defer the proof to \ref{sec:proof most balanced}.%

\subsection{The Global Decomposition Algorithm\label{sec:decomp recurse}}

In this section, we show how to construct the global expansion decomposition.
Our fast construction is motivated by two algorithms which are too
slow. However, they give main ideas of our algorithm so we discuss
them below.

\paragraph{Early Attempts}

First, the following is a straight-forward recursive algorithm, which
is too slow, for constructing the decomposition: use a sparsest cut
algorithm to find an $\alpha$-sparse cut $S$. If there is no $\alpha$-sparse
cut, then return the current graph as a component in $G^{d}$, otherwise
include the cut edges $\partial(S)$ into $G^{s}$. Then recurse on
$G[S]$ and $G[V\setminus S]$. In \cite{AndoniCKQWZ16}, they implicitly
find a global expansion decomposition with this approach, since they
focus on the space complexity and not running time. This algorithm
can take $\Omega(mn)$ time because the cut $S$ might be very small,
so the recursion tree can be very unbalanced.

Second, a next natural approach is to use a most balanced sparse cut
(recall \ref{def:most balanced sparse cut}) instead of a sparsest
cut. Suppose that there exists a $(2,1)$-approximate most balanced
sparse cut algorithm, this approach would have been straightforward: 
\begin{prop}
If there is a $(2,1)$-approximate most balanced sparse cut algorithm
$\cA$ with running time $T(m)$ on a graph with $n$ vertices and
$m$ edges, then there is an algorithm for computing expansion decomposition
for graph with $n$ vertices and $m$ edges in $O(T(m,n)\log n)$
time.\end{prop}
\begin{proof}
For any graph $H=(V_{H},E_{H})$, the recursive procedure $\decomp(H)$
does the following: if $\cA(H,\alpha)$ reports that $\phi(H)\ge\alpha$
then return $H$ as a connected components in $G^{d}$. Else $\cA(H,\alpha)$
finds a cut $S$ where $\opt(H,\alpha)/2\le|S|\le\opt(H,\alpha)$.
Add the cut edges $\partial_{H}(S)$ into $G^{s}$, then recurse on
$\decomp(H[S])$ and $\decomp(H[V_{H}\setminus S])$. Given an input
graph $G$, we claim that $\decomp(G)$ returns an expansion decomposition
of $G$.

To see the correctness, 1) each component $H$ in $G^{d}$ has high
expansion: $\phi(H)\ge\alpha$ by construction, and 2) the number
of edges in $G^{s}$ is at most $O(\alpha n\log n)$ because whenever
$\partial_{H}(S)$ is added into $G^{s}$, we charge these edges to
nodes in $S$. By averaging, each node is charged $\delta_{H}(S)/|S|=\phi_{H}(S)\le\alpha$
Since $|S|\le|V_{H}|/2$, each node can be charged $O(\log n)$ times.

To bound the running time, it is enough to show that the recursion
depth is $O(\log n)$ because the graphs in the two branches of the
recursion $H[S]$ and $H[V_{H}\setminus S]$ are disjoint. Let $S'$
be such that $|S'|=\opt(H[V_{H}\setminus S],\alpha)$. The case where
$|S|\ge|V_{H}|/4$ or $|S'|\ge|V_{H}|/4$ can happen at most $O(\log n)$
times throughout the recursion. Now, suppose that $|S|<|V_{H}|/4$
and $|S'|<|V_{H}|/4$. So $|S\cup S'|<|V_{H}|/2$. We claim that $|S'|\le\opt(H,\alpha)/2$
which can also happen at most $O(\log n)$ times. Suppose otherwise,
then $S\cup S'$ is an $\alpha$-sparse cut in $H$ where $|V_{H}|/2\ge|S\cup S'|>\opt(H,\alpha)/2+\opt(H,\alpha)/2=\opt(H,\alpha)$
which is a contradiction. So the recursion depth is $O(\log n)$.
\end{proof}
Unfortunately, we do not have an efficient $(2,1)$-approximate most
balanced sparse cut algorithm. By \ref{thm:most balanced sparse cut},
we only have $(O(\log^{2}n),O(\log^{3}n))$-approximate most balanced
sparse cut algorithm. Our main technical goal is, then, to use this
weaker guarantee while making this approach goes through.

\paragraph{Efficient Algorithm}

Let $\cA'$ be the algorithm from \ref{thm:most balanced sparse cut}.
Recall that, given any graph $H$ with $n'$ nodes and a parameter
$\alpha'$, $\cA'$ has approximation ratio $c_{size}(n')=\Theta(\log^{2}n')$
and $c_{exp}(n')=\Theta(\log^{3}n')$, i.e., with probability $1-\frac{1}{(n')^{k}}$
for a large constant $k$, $\cA'$ either outputs an $\alpha'$-sparse
cut $S$ where $\opt(H,\alpha'/c_{exp}(n'))/c_{size}(n')\le|S|\le n'/2$
or reports that $\phi(H)\ge\alpha'/c_{exp}(n')$. In the following,
sometimes the input graph $H$ is small and so the probability $1-\frac{1}{(n')^{k}}$
is high enough for our purpose. To boost this probability to $1-p'$
for arbitrarily small $p$, we define the algorithm $\cA$ to be the
algorithm that repeatedly runs $\cA'$ for $O(\log\frac{1}{p'})$
iterations. If there is some iteration where $\cA'$ returns an $\alpha'$-sparse
cut $S$, then $\cA$ just return $S$. Otherwise, there is no iteration
that $\cA'$ returns an $\alpha'$-sparse cut, then $\cA$ report
that $\phi(H)\ge\alpha'/c_{exp}(n')$ with probability $1-p'$. In
the following, we will use this modified algorithm $\cA$ instead
of the algorithm $\cA'$ from \ref{thm:most balanced sparse cut}. 
\begin{fact}
For any given $p$, the algorithm $\cA$ defined above works correctly
with probability $1-p$.\label{thm: correct report whp}
\end{fact}
Suppose that the input of our decomposition is a graph $G$ with $n$
vertices and $m$ edges, a parameter $\alpha$, and the target success
probability $1-p$. We now define some parameters. Let $\bar{c}_{size}=c_{size}(n)$,
$\bar{c}_{exp}=c_{exp}(n)$, and $\epsilon=\sqrt{\log\bar{c}_{exp}/\log n}$.
Let $\bar{s}_{1},\dots,\bar{s}_{L}$ be such that $\bar{s}_{1}=n/2+1$,
$\bar{s}_{L}\le1$, and $\bar{s}_{\ell}=\bar{s}_{\ell-1}/n^{\epsilon}$
for $1<\ell<L$. Hence, $L\le1/\epsilon$. Let $\alpha_{1},\dots,\alpha_{L}$
be such that $\alpha_{L}=\alpha\bar{c}_{exp}$ and $\alpha_{\ell}=\alpha_{\ell+1}\bar{c}_{exp}$
for $\ell<L$. Hence, $\alpha_{1}=\alpha\bar{c}_{exp}^{L}$. Let $p'=p/n^{2}$.
So $\cA$ repeatedly runs $\cA'$ for $O(\log\frac{1}{p'})=O(\log\frac{n}{p})$
iterations. 

For any graphs $H=(V_{H},E_{H})$, $I=(V_{I},E_{I})$ and a number
$\ell$, the main procedure $\decomp(H,I,\ell)$ is defined as in
\ref{alg:decomp}. 
\begin{algorithm}
\caption{\label{alg:decomp}$\protect\decomp(H,I,\ell)$ where $H=(V_{H},E_{H})$
and $I=(V_{I},E_{I})$.}

Given that, $\opt(I,\alpha_{\ell})<\bar{s}_{\ell}$ do the following: 
\begin{enumerate}
\item If $H$ is a singleton or $\cA(H,\alpha_{\ell})$ reports $\phi(H)\ge\alpha_{\ell}/c_{exp}(|V_{H}|)$,
then return $H$ as a connected components in $G^{d}$. 
\item Else, $\cA(H,\alpha_{\ell})$ returns an $\alpha_{\ell}$-sparse cut
$S$ in $H$ where $\opt(H,\alpha_{\ell}/c_{exp}(|V_{H}|))/c_{size}(|V_{H}|)\le|S|\le|V_{H}|/2$.

\begin{enumerate}
\item If $|S|\ge\bar{s}_{\ell+1}/\bar{c}_{size},$ then add the cut edges
$\partial_{H}(S)$ into $G^{s}$ and recurse on $\decomp(H[S],H[S],1)$
and $\decomp(H[V_{H}\setminus S],I,\ell)$. 
\item Else, recurse on $\decomp(H,H,\ell+1)$. \end{enumerate}
\end{enumerate}
\end{algorithm}

We claim that $\decomp(G,G,1)$ returns an expansion decomposition
for $G$ with parameter $\alpha$ in time $\tilde{O}(m^{1+o(1)})$.
The following lemmas show the correctness. For readability, we first
assume that $\cA$ works correctly with certainty, and then we remove
the assumption later.
\begin{prop}
\label{prop:invariant} For any graphs $H,I\subseteq G$ and $\ell$,
if $\decomp(H,I,\ell)$ is called, then the invariant $\opt(I,\alpha_{\ell})<\bar{s}_{\ell}$
is satisfied. \end{prop}
\begin{proof}
When $\ell=1$, $\opt(I,\alpha_{1})<\bar{s}_{1}=n/2+1$ for any $I$
in a trivial way. In particular, the invariant is satisfied when $\decomp(G,G,1)$
or $\decomp(H[S],H[S],1)$ is called. The invariant $\decomp(H[V_{H}\setminus S],I,\ell)$
is the same as the one for $\decomp(H,I,\ell)$, and hence is satisfied.
Finally, $\decomp(H,H,\ell+1)$ is called in step 2.b only when $|S|<\bar{s}_{\ell+1}/\bar{c}_{size}$,
so we have 
\begin{eqnarray*}
\opt(H,\alpha_{\ell+1}) & = & \opt(H,\alpha_{\ell}/\bar{c}_{exp})\le\opt(H,\alpha_{\ell}/c_{exp}(|V_{H}|))\\
 & \le & c_{size}(|V_{H}|)|S|<\frac{c_{size}(|V_{H}|)}{\bar{c}_{size}}\le\bar{s}_{\ell+1}.
\end{eqnarray*}
These inequalities hold because $c_{exp}(|V_{H}|)\le\bar{c}_{exp}$
and $c_{size}(|V_{H}|)\le\bar{c}_{size}$ as $|V_{H}|\le n$, and
we know that $\opt(H,\alpha)\le\opt(H,\alpha')$ for any $\alpha\le\alpha'$.\end{proof}
\begin{prop}
\label{lem:decomp terminate} At level $L$, $\decomp(H,H,L)$ always
return the graph $H$ (no further recursion). \end{prop}
\begin{proof}
The invariant $\opt(H,\alpha_{L})<\bar{s}_{L}\le1$ implies that $\phi(H)>\alpha_{L}$.
Hence, $\cA(H,\alpha_{L})$ can never find an $\alpha_{L}$-sparse
cut. So $\cA(H,\alpha_{L})$ must report $\phi(H)\ge\alpha_{L}/c_{exp}(|V_{H}|)$.
So $H$ is returned.\end{proof}
\begin{prop}
\label{lem:expansion} Each component $H$ in $G^{d}$ either is a
singleton or has expansion $\phi(H)\ge\alpha$.\end{prop}
\begin{proof}
$H$ is returned as a component in $G^{d}$ only when $H$ is a singleton
or $\cA$ reports that $\phi(H)\ge\alpha_{\ell}/c_{exp}(|V_{H}|)\ge\alpha_{L}/\bar{c}_{exp}=\alpha$.
Note that $\ell<L$ by \ref{lem:decomp terminate}.\end{proof}
\begin{prop}
\label{lem:crossing edges} The number of edges in $G^{s}$ is at
most $O(\alpha_{1}n\log n)$. \end{prop}
\begin{proof}
Suppose $\partial_{H}(S)$ is added into $G^{s}$ when $\decomp(H,I,\ell)$
for some graphs $H,I\subset G$, we charge these edges to nodes in
$S$. By averaging, each node is charged $\delta_{H}(S)/|S|=\phi_{H}(S)\le\alpha_{\ell}\le\alpha_{1}$.
Since $|S|\le|V_{H}|/2$, each node can be charged $O(\log n)$ times. 
\end{proof}
To analyze the running time, we need some more notation. We define
the recursion tree $\mathcal{T}$ of $\decomp(G,G,1)$ as follows.
Each node of $\mathcal{T}$ represents the parameters of the procedure
$\decomp(H,I,\ell)$. That is, $(G,G,1)$ is the root node. For each
$(H,I,\ell)$, if $\decomp(H,I,\ell)$ returns $H$, then $(H,I,\ell)$
is a leaf. If $\decomp(H,I,\ell)$ recurses on $\decomp(H[S],H[S],1)$
and $\decomp(H[V\setminus S],I,\ell)$, then $(H[S],H[S],1)$ and
$(H[V\setminus S],I,\ell)$ is a left and right children of $(H,I,\ell)$
respectively, and the edges $((H,I,\ell),(H[S],H[S],1))$ and $((H,I,\ell),(H[V\setminus S],I,\ell))$
are \emph{left edge} and \emph{right edge} respectively. If $\decomp(H,I,\ell)$
recurses on $\decomp(H,H,\ell+1)$, then $(H,H,\ell+1)$ is the only
child of $(H,I,\ell)$, and the edge $((H,I,\ell),(H,H,\ell+1))$
is a \emph{down edge}. 
\begin{lem}
\label{lem:depth implies time}If the depth of the recursion tree
$\mathcal{T}$ is $D$, then the total running time is $\tilde{O}(mD)$.\end{lem}
\begin{proof}
The running time on $\decomp(H,I,\ell)$, excluding the time spent
in the next recursion level, is at most $\tilde{O}(|E_{H}|\log\frac{1}{p})$
because the algorithm $\cA$ just repeatedly runs the algorithm $\cA'$
from \ref{thm:most balanced sparse cut} for $O(\log\frac{1}{p})$
times, and $\cA'$ runs in near-linear time. 

Let $\mathcal{T}_{d}$ be the set of all the nodes in $\mathcal{T}$
of depth \emph{exactly} $d$. We can see that, for any two nodes $(H,I,\ell),(H',I,\ell')\in\mathcal{T}_{d}$,
$H$ and $H'$ are disjoint subgraphs of $G$. So, the total running
time we spent on $\mathcal{T}_{d}$ is $\tilde{O}(m\log\frac{1}{p'})$.
Hence, the total running time of $\mathcal{T}$ is $\tilde{O}(mD\log\frac{1}{p'})$.\end{proof}
\begin{lem}
\label{lem:bound depth}Let $P$ be any path from leaf to root of
$\mathcal{T}$. $P$ contains at most $\log n$ left edges, $L\log n$
down edges, and $L\log n\times\gamma$ right edges where $\gamma=\bar{c}_{size}n^{\epsilon}$.
Therefore, $\cT$ has depth at most $L\log n\cdot\bar{c}_{size}n^{\epsilon}$.\end{lem}
\begin{proof}
Consider the left edge $((H,I,\ell),(H[S],H[S],1))$. Since $|S|\le|V_{H}|/2$,
$P$ contains $\log n$ left edges. Between any two left edges in
$P$ there are at most $L$ down edges be \ref{lem:decomp terminate}.
So $P$ contains at most $L\log n$ down edges. To prove that there
are at most $L\log n\times\gamma$ right edges in $P$, it suffices
to prove that there cannot be $\gamma$ right edges between any left
edges or down edges in $P$.

Suppose that $(H_{1},I,\ell),\dots,(H_{\gamma},I,\ell)$ are nodes
in $P$ where, for each $i$, $((H_{i},I,\ell),(H_{i+1},I,\ell))$
is a right edge, and $(H_{1},I,\ell)$ is a deeper endpoint of left
edges or a down edges (hence $I=H_{1}$).

For each $i$, let $S_{i}$ be the cut such that $H_{i+1}=H_{i}[V_{H_{i}}\setminus S_{i}]$
and $\phi_{H_{i}}(S_{i})<\alpha_{\ell}$. Since $\{S_{i}\}_{i}$ are
mutually disjoint and $\partial_{H_{1}}(\bigcup_{i=1}^{\gamma}S_{i})\subset\bigcup_{i=1}^{\gamma}\partial_{H_{i}}(S_{i})$,
we can conclude $\phi_{H_{1}}(\bigcup_{i=1}^{\gamma}S_{i})<\alpha_{\ell}$.
However, we also have that $|S_{i}|\ge\bar{s}_{\ell+1}/\bar{c}_{size}$,
for all $i$, and hence $|\bigcup_{i=1}^{\gamma}S_{i}|
\ge\gamma\bar{s}_{\ell+1}/\bar{c}_{size}\ge n^{\epsilon}\bar{s}_{\ell+1}=\bar{s}_{\ell}$. So $\bigcup_{i=1}^{\gamma}S_{i}$ contradicts the invariant for
$\decomp(H_{1},H_{1},\ell)$ (note that $I=H_{1})$ which says $\opt(H_{1},\alpha_{\ell})<\bar{s}_{\ell}$. 
\end{proof}
Now, we can conclude the theorem.
\begin{proof}
[Proof of \ref{thm:high-exp-decomp}] Recall that $\epsilon=\sqrt{\log\bar{c}_{exp}/\log n}=\Theta(\sqrt{\log\log n/\log n})$
and $L\le1/\epsilon$. Given $G$ and $\alpha$, by \ref{lem:depth implies time}
and \ref{lem:bound depth} the running time of $\decomp(G,G,1)$ is
$\tilde{O}(mL\bar{c}_{size}n^{\epsilon}\log\frac{1}{p'})=O(m^{1+O(\sqrt{\log\log n/\log n})}\log\frac{1}{p})$.
By \ref{lem:expansion} and \ref{lem:crossing edges}, for each connected
component $H$ of $G^{d}$, $H$ is a singleton or $\phi(H)\ge\alpha$
and the number of edges in $G^{s}$ is $O(\alpha_{1}n\log n)=O(\alpha c_{exp}^{L}n\log n)=\alpha n^{1+O(\sqrt{\log\log n/\log n})}$.
This is the output for the  expansion decomposition as claimed. 

Finally, we remove the assumption that $\cA$ is deterministic, and
show that with probability $1-p$, the outputted decomposition is
correct. Observe that, as the depth of the recursion tree $\mathcal{T}$
of $\decomp(G,G,1)$ is $D=L\log n\cdot\bar{c}_{size}n^{\epsilon}$,
then $\cA$ is called at most $nD$ times. So the probability that
$\cA$ always works correctly is at least $1-p'nD\ge1-p$.
\end{proof}
Finally, we add some easy observation that will be used in \ref{sec:local decomp}.
\begin{prop}
All the edges added into $E^{s}$ by \ref{alg:decomp} are the cut
edges $\partial_{H}(S)$ for some subgraph $H$ of $G$ and a cut
$S\subset V(H)$ where $\phi_{H}(S)\le\alpha_{1}=\alpha\bar{c}_{exp}^{L}=\alpha\gamma$
where $\gamma=n^{O(\sqrt{\log\log n/\log n})}$. \label{thm:add only sparse cut}\end{prop}

%% file: local_decomposition.tex
\section{Local Expansion Decomposition\label{sec:local decomp}}

In this section, we show a novel ``local'' version of the global
expansion decomposition algorithm from \ref{sec:decomp alg}. This
algorithm differs from the algorithm from \ref{sec:decomp alg} in
two aspects. First, it needs that the input graph $G$ is obtained
by deleting some edge set $D$ from another graph $G_{b}$, which
``we expect that'' $\phi(G_{b})\ge\alpha_{b}$. Second, the algorithm
is \emph{local, }in the sense that its running time essentially depends
only on the size of $|D|$ and $\alpha_{b}$, and not the size of
$G$. 

This algorithm is a crucial tool for obtaining our Las Vegas dynamic
$\st$ algorithm. It is important that the dependency of $|A|$ in
the running time is truly subquadratic (i.e. $|A|^{2-\epsilon}$ for
some $\epsilon$), which is the case here. For any $n$, recall $\gamma=n^{O(\sqrt{\log\log n/\log n}})=n^{o(1)}$
is the factor from \ref{thm:high-exp-decomp}.
\begin{thm}
[Local Expansion Decomposition Algorithm]\label{thm:local decomp}For
any constant $\epsilon\in(0,1)$, there is an algorithm $\cA$ that
can do the following:
\begin{itemize}
\item $\cA$ is given pointers to $n,p,\Delta,G_{b}$,$D,\alpha_{b}$ which
are stored in a memory: $p$ is a failure probability parameter. $G_{b}=(V,E_{b})$
is a $\Delta$-bounded degree graphs with $n$ nodes (represented
by an adjacency list). $\alpha_{b}$ is an expansion parameters where
$\alpha_{b}<\frac{1}{\gamma^{\omega(1)}}$ . $D\subset E_{b}$ is
a set of edges in $G_{b}$. Let $A$ denote a set of endpoints of
edges in $D$. Let $G=(V,E)=(V,E_{b}-D)$ be the graph that $\cA$
will compute the decomposition on.
\item Then (without reading the whole inputs), in time $\tilde{O}(\frac{\Delta^{5.5}|A|^{1.5+\epsilon}}{\alpha_{b}^{4+\epsilon}\epsilon}\gamma\log\frac{1}{p})$,
$\cA$ either 

\begin{itemize}
\item reports that $G_{b}$ has expansion less than $\alpha_{b}$, or 
\item outputs 1) a set of edges $E^{s}\subset E$, and 2) all components,
except the largest one, of $G^{d}$ where $G^{d}=(V,E^{d})=(V,E-E^{s})$. 
\end{itemize}
\item Moreover, if $G_{b}$ has expansion at least $\alpha_{b}$, then with
probability $1-O(p)$ we have

\begin{itemize}
\item $|E^{s}|\le4\Delta|A|/\alpha_{b}$, and
\item each connected component $C$ of $G^{d}$ either is a singleton or
has high expansion: $\phi(C)\ge\alpha$ where $\alpha=(\alpha_{b}/6\Delta)^{1/\epsilon}$.
\end{itemize}
\end{itemize}
\end{thm}
Observe that $\epsilon$ is a trade-off parameter such that, on one
hand when $\epsilon$ is small, the algorithm is fast but has a bad
expansion guarantee in the output, on the other hand when $\epsilon$
is big, the algorithm is slow but has a good expansion guarantee. 
\begin{rem}
\label{rem:know components}In \ref{thm:local decomp}, by \emph{outputting
a components }$C$ of $G^{d}$, it means for each nodes $u\in V(C)$,
$(u,C)$ is outputted indicating that $u\in V(C)$. Since the algorithm
$\cA$ in \ref{thm:local decomp} outputs all components, except one,
of $G^{d}$. We can infer the size of \emph{all }components in $G^{d}$,
and for any node $u\in V$, sa component $C$ of $G^{d}$ where $u\in V(C)$.
\end{rem}
Although the main idea for proving \ref{thm:local decomp} is similar
to \ref{thm:high-exp-decomp}, the analysis is more involved. We explain
the high level idea below. There are two main steps. First, in \ref{sec:local sparse cut}
we show an algorithm for a problem called \emph{locally balanced sparse
cut} \emph{(LBS cut)}, which is basically a ``local version'' of
most balanced sparse cut defined in \ref{sec:sparsest cut}. Second,
we use this algorithm as a main procedure for constructing the decomposition
in \ref{sec:local decomp alg}.

In LBS cut problem, we are given a graph $G=(V,E)$, a target set
$A\subset V$ and a parameter $\alpha$. Then we need to find a $\alpha$-sparse
cut $S$ (i.e. $\phi_{G}(S)<\alpha$) where $|S|\le|V-S|$ such that
$|S|$ is larger than all $\alpha$-sparse cuts which are ``near''
the target set $A$. (``Nearness'' is defined precisely in \ref{def:overlapping}).
To compare, in the most balanced sparse cut problem, $|S|$ needs
to be larger than \emph{all} $\alpha$-sparse cuts. The nice thing
about LBS cut problem is that there is an approximate algorithm for
this problem which is ``local'' (i.e. its running time depends essentially
only on $vol(A)$, and not $|V|$). This algorithm can be obtained
quite easily by slightly modifying and analyzing the algorithm by
Orecchia and Zhu \cite{OrecchiaZ14} for local cut improvement problem.
See \ref{sec:local sparse cut} for details.

Next, to obtain the decomposition, the main approach is to find an
approximate LBS cut where the target set $A$ is the endpoints of
$D$ (union with some additional nodes) and recurse on both sides.
However, the analysis is more involved than the one in \ref{thm:high-exp-decomp}.
One important observation is that, when $\phi(G_{b})\ge\alpha_{b}$,
any $(\frac{\alpha_{b}}{2})$-sparse cut $S$ in $G=G_{b}-D$ must
be ``near'' to $D$. Intuitively, this is because the expansion
of $S$ get halved after deleting edges in $D$. This justifies why
it is enough to find approximate $(\frac{\alpha_{b}}{2})$-sparse
LBS cuts instead of finding approximate most balanced sparse cuts
as in the algorithm for \ref{thm:high-exp-decomp}. As our approximate
LBS cut algorithm is local, we can output the decomposition in time
essentially independent from the size of $G$. See \ref{sec:local decomp alg}
for details.

\subsection{Locally Balanced Sparse Cut\label{sec:local sparse cut}}

In this section, we show the crucial tool for proving \ref{thm:local decomp}.
First, we need this definition:
\begin{defn}
[Overlapping]For any graph $G=(V,E)$ and set $A\subset V$, a cut
$S\subset V$ is $(A,\sigma)$-overlapping in $G$ if $|S\cap A|/|S|\ge\sigma$.\label{def:overlapping}
\end{defn}
Let $G=(V,E)$ be a graph. Recall that a cut $S$ is $\alpha$-sparse
if $\phi(S)=\frac{\delta(S)}{\min\{|S|,|V-S|\}}<\alpha$. For any
set $A\subset V$, an overlapping parameter $\sigma$ and an expansion
parameter $\alpha$, let $S^{*}$ be the largest $\alpha$-sparse
$(A,\sigma)$-overlapping cut where $|S^{*}|\le|V-S^{*}|$. We define
$\opt(G,\alpha,A,\sigma)=|S^{*}|$. If $S^{*}$ does not exist, then
$\opt(G,\alpha,A,\sigma)=0$. Now, we define LBS cut problem formally:
\begin{defn}
[Locally Balanced Sparse Cut]\label{def:most balanced sparse cut-1}
For any graph $G=(V,E)$, a set $A\subset V$, and parameters $c_{size},c_{exp},\sigma$
and $\alpha$, a cut $S$ where $|S|<|V-S|$ is a $(c_{size},c_{exp})$-approximate
locally balanced $\alpha$-sparse cut w.r.t. $(A,\sigma)$-overlapping
cuts if $\phi(S)<\alpha$ and $|S|\ge\opt(G,\alpha/c_{exp},A,\sigma)/c_{size}$.
We also write $S$ is a $(c_{size},c_{exp},\alpha,A,\sigma)$-LBS
cut.
\end{defn}
We note that $(c_{size},c_{exp},\alpha,A,\sigma)$-LBS cut $S$ may
not be $(A,\sigma)$-overlapping. The existence of $S$ just show
that for any $(A,\sigma)$-overlapping cut of size at least $c_{size}|S|$
must have expansion at least $\alpha/c_{exp}$. 
\begin{defn}
[Locally Balanced Sparse Cut Algorithm]\label{def:most balanced sparse cut alg-1}A
$(c_{size},c_{exp})$-approximate algorithm $\cA$ for locally balanced
sparse cut problem is given a graph $G=(V,E)$, a set $A\subset V$,
and an overlapping parameter $\sigma$ and an expansion parameter
$\alpha$ , and then $\cA$ either 
\begin{itemize}
\item finds a $(c_{size},c_{exp},\alpha,A,\sigma)$-LBS cut $S$, or
\item reports that there is no $(\alpha/c_{exp})$-sparse $(A,\sigma)$-overlapping
cut.
\end{itemize}

We also write that $\cA$ is a $(c_{size},c_{exp})$-approximate LBS
cut algorithm.

\end{defn}
The main result of this section is the following:
\begin{thm}
Given a graph $G=(V,E)$, a set $A\subset V$ where $|A|\le4|V-A|$,
an expansion parameter $\alpha$, and an overlapping parameter $\sigma\in[\frac{3|A|}{|V-A|},\frac{3}{4}]$,
there is a $(c_{size},c_{exp})$-approximate LBS cut algorithm with
running time $O((vol(A)/\sigma)^{1.5}\log^{2}(\frac{vol(A)}{\sigma}))$
where $c_{size}=3/\sigma$ and $c_{exp}=3/\sigma$.\label{thm:local MSO cut}
\end{thm}
This algorithm can be obtained quite easily by slightly modifying
and analyzing the algorithm by Orecchia and Zhu \cite{OrecchiaZ14}
for local cut improvement problem. Therefore, we defer the proof to
\ref{sec:proof local MOS cut}.

\subsection{The Local Decomposition Algorithm\label{sec:local decomp alg}}

Throughout this section, let $n$ and $p$ be given. Let $G_{b}=(V,E_{b})$
be a graph with $n$ nodes with maximum degree $\Delta$. Let $D\subset E$
be some set of edges. Let $G=(V,E)=(V,E_{b}-D)$ and $A$ be a set
of endpoints of edges in $D$. We call $G_{b}$ a \emph{before graph}.
Our algorithm is given $D$, a constant $\epsilon$, and an expansion
parameter $\alpha_{b}$ as inputs. We want the algorithm to output
a set of edges $E^{s}\subset E$ and all components, except the largest
one, of $G^{d}$ where $G^{d}=(V,E^{d})=(V,E-E^{s})$. Let $G^{s}=(V,E^{s})$.
Let $\bar{t}=\tilde{O}(\frac{\Delta^{5.5}|A|^{1.5+\epsilon}}{\alpha_{b}^{4+\epsilon}\epsilon}\gamma\log\frac{1}{p})$
be the parameter indicating the time limit.

We now define some more notations. Let $\cA_{cut}$ be the deterministic
algorithm for finding LBS cuts from \ref{thm:local MSO cut}. We set
the overlapping parameter $\sigma=\alpha_{b}/2\Delta$ for $\cA_{cut}$.
Hence the approximation ratios of $\cA_{cut}$ are $c_{size}=3/\sigma$
and $c_{exp}=3/\sigma$. Let $\cA_{decomp}$ be the randomized algorithm
for finding an expansion decomposition from \ref{thm:high-exp-decomp}.
We set $p'=p/n^{2}$ as the failure probability parameter for $\cA_{decomp}$.

The following notations was similarly defined as in \ref{sec:decomp recurse}.
Let $\bar{s}_{1},\dots,\bar{s}_{L}$ be such that $\bar{s}_{1}=4|A|/\alpha_{b}+1$,
$\bar{s}_{L}\le1$, and $\bar{s}_{\ell}=\bar{s}_{\ell-1}/(\bar{s}_{1})^{\epsilon}$
for $1<\ell<L$. Hence, $L\le1/\epsilon$. We denote $\alpha=(\alpha_{b}/6\Delta)^{1/\epsilon}$.
Let $\alpha_{1},\dots,\alpha_{L}$ be such that $\alpha_{L}=\alpha$
and $\alpha_{\ell}=\alpha_{\ell+1}c_{exp}$ for $\ell<L$. Hence,
$\alpha_{1}=\alpha c_{exp}^{L-1}$. 
\begin{fact}
$\alpha_{1}<\alpha_{b}/2$.\end{fact}
\begin{proof}
We have 
\[
\alpha_{1}=\alpha c_{exp}^{L-1}=(\frac{\alpha_{b}}{6\Delta})^{1/\epsilon}(\frac{3}{\sigma})^{1/\epsilon-1}\le(\frac{\alpha_{b}}{6\Delta})^{1/\epsilon}(\frac{6\Delta}{\alpha_{b}})^{1/\epsilon-1}=\alpha_{b}/6\Delta<\alpha_{b}/2.
\]

\end{proof}
For any graphs $H=(V_{H},E_{H})$, $I=(V_{I},E_{I})$, and a number
$\ell$, the main procedure $\decomp(H,I,\ell)$ is defined as in
\ref{alg:local decomp}. For any $\alpha'$ and $B\subset V_{H}$,
recall that $\opt(H,\alpha')$ is the size of the largest $\alpha'$-sparse
cut $S$ in $H$ where $|S|\le|V_{H}-S|$, and $\opt(H,\alpha',B,\sigma)$
is the size of the largest $\alpha'$-sparse $(B,\sigma)$-overlapping
cut $S$ in $H$ where $|S|\le|V_{H}-S|$. By definition, $\opt(H,\alpha')\ge\opt(H,\alpha',B,\sigma)$.

The algorithm is simply to run $\decomp(G,G,1)$ with time limit $\bar{t}$.
If $\decomp(G,G,1)$ takes time more than $\bar{t}$, then we reports
that $\phi(G_{b})<\alpha_{b}$.

\begin{algorithm}
\caption{\label{alg:local decomp}$\protect\decomp(H,I,\ell)$ where $H=(V_{H},E_{H})$
and $I=(V_{I},E_{I})$}

\begin{enumerate}
\item Set $B_{H}=(A\cup A_{H})\cap V_{H}$ where $A_{H}$ is the set of
endpoints of edges in $\partial_{G}(V_{H})$
\item If $|V_{H}-B_{H}|<\frac{3}{\sigma}|B_{H}|$, then run $\cA_{decomp}(H,\alpha,p')$
which partitions $E_{H}$ into $E_{H}^{s}$ and $E_{H}^{d}$. Add
$E_{H}^{s}$ into $G^{s}$, and return each connected component of
a graph induced by $E_{H}^{d}$ as a components in $G^{d}$. 
\item If $\ell=L$, then return $H$ as a connected components in $G^{d}$.
\item If $H$ is a singleton or $\cA_{cut}(H,\alpha_{\ell},B_{H},\sigma)$
reports there is no $(\alpha_{\ell}/c_{exp})$-sparse $(B_{H},\sigma)$-overlapping
cut, then return $H$ as a connected components in $G^{d}$. 
\item Else, $\cA_{cut}(H,\alpha_{\ell},B_{H},\sigma)$ returns an $\alpha_{\ell}$-sparse
cut $S$ in $H$ where $\opt(H,\alpha_{\ell}/c_{exp},B_{H},\sigma)/c_{size}\le|S|\le|V_{H}|/2$.

\begin{enumerate}
\item If $|S|\ge\bar{s}_{\ell+1}/c_{size},$ then add the cut edges $\partial_{H}(S)$
into $G^{s}$ and recurse on $\decomp(H[S],H[S],1)$ and $\decomp(H[V_{H}\setminus S],I,\ell)$. 
\item Else, recurse on $\decomp(H,H,\ell+1)$. \end{enumerate}
\end{enumerate}
\end{algorithm}

\subsubsection{Validity\label{sec:Validity}}

We first show that the parameters for $\cA_{cut}$ are valid when
it is called. 
\begin{lem}
Whenever $\cA_{cut}(H,\alpha_{\ell},B_{H},\sigma)$ is called, we
have that $\sigma\ge3\frac{|B_{H}|}{|V_{H}-B_{H}|}$ satisfying the
requirement for $\cA_{cut}$ as stated in \ref{thm:local MSO cut}.\end{lem}
\begin{proof}
Observe that $\cA_{cut}$ can be called only when the condition in
Step 2 of \ref{alg:local decomp} is false: $|V_{H}-B_{H}|\ge\frac{3}{\sigma}|B_{H}|$.
That is, $\sigma\ge3\frac{|B_{H}|}{|V_{H}-B_{H}|}$.
\end{proof}
Next, we show that if $\decomp(G,G,1)$ finishes with in time limit
$\bar{t}$, then it outputs 1) a set of edges $E^{s}\subset E$, 2)
all components, except the largest one, of $G^{d}$ where $G^{d}=(V,E^{d})=(V,E-E^{s})$.
Recall \ref{rem:know components}, by outputting a components $C$
of $G^{d}$, it means for each nodes $u\in V(C)$, $(u,C)$ is outputted
indicating that $u\in V(C)$. 

It is clear from \ref{alg:local decomp} that when $E^{s}$ is outputted,
we have $G^{d}=(V,E^{d})=(V,E-E^{s})$. It is also obvious how to
prevent \ref{alg:local decomp} from outputting the largest component
$C$ of $G^{d}$, which may take a lot of time. That is, whenever
\ref{alg:local decomp} return $H$ as a component of $G^{d}$ in
Step 2,3 or 4, we do nothing if $H$ is ``reached'' by recursing
from only the ``larger'' side of the cut. Otherwise, we list all
nodes $u\in V(H)$, and output $(u,H)$.

\subsubsection{Correctness}

In this section, we prove that if $\phi(G_{b})\ge\alpha_{b}$, then
the outputs of $\decomp(G,G,1)$ have the desired properties. For
readability, we first assume that $\cA_{decomp}$ works correctly
with certainty, and then we remove the assumption later. 

To bound $|E^{s}|$, we first need the following fact:
\begin{prop}
Suppose that $\phi(G_{b})\ge\alpha_{b}$. For any $\alpha'<\alpha_{b}/2$,
any $\alpha'$-sparse cut $S$ in $G$ must have size at most $|S|\le4|A|/\alpha_{b}$.\label{thm:sparse cut small}\end{prop}
\begin{proof}
Suppose otherwise that $|S|>4|A|/\alpha_{b}$. Then, as $G=G_{b}-D$,
we have 
\[
\delta_{G}(S)\ge\delta_{G_{b}}(S)-|D|\ge\alpha_{b}|S|-2|A|>\alpha_{b}|S|-\alpha_{b}|S|/2=\alpha_{b}|S|/2,
\]
which means, $\phi_{G}(S)>\alpha_{b}/2>\alpha'$, a contradiction.\end{proof}
\begin{lem}
If $\phi(G_{b})\ge\alpha_{b}$, then $|E^{s}|\le4\Delta|A|/\alpha_{b}$.\label{thm:local decomp sparse }\end{lem}
\begin{proof}
Observe all the edges added into $E^{s}$ by \ref{alg:local decomp}
are the cut edges $\partial_{H}(S)$ for some subgraph $H$ of $G$
and a cut $S\subset V(H)$. Let us list all sets of edges $\partial_{H_{1}}(S_{1}),\dots,\partial_{H_{t}}(S_{t})$
that constitute $E^{s}$. 
\begin{claim}
For all $i$, $\phi_{H_{i}}(S_{i})<\alpha_{b}/2$.\end{claim}
\begin{proof}
The cut edges $\partial_{H_{i}}(S_{i})$ are added into $E^{s}$ by
\ref{alg:local decomp} either in Step 2 or Step 5.a. In the first
case, we know from \ref{thm:add only sparse cut} that $\phi_{H_{i}}(S_{i})<\alpha\gamma.$
Now, as $\alpha=O(\alpha_{b}/\Delta)^{1/\epsilon}$ and $\epsilon$
is less than $1$ by some constant, we have $\alpha\le\alpha_{b}^{1-\delta}$
for some constant $\delta>0$. Since $\alpha_{b}\le\frac{1}{\gamma^{\omega(1)}}$,
we have $\gamma\le\frac{1}{\alpha_{b}^{o(1)}}$. Therefore, $\phi_{H_{i}}(S_{i})<\alpha\gamma\le\alpha_{b}^{1-\delta+o(1)}<\alpha_{b}/2$.
In the second case, we know from Step 5 that $\phi_{H_{i}}(S_{i})<\alpha_{1}<\alpha_{b}/2$.
\end{proof}
Observe that for any $S_{i}$ and $S_{j}$ where $i\neq j$, either
$S_{i}\subset S_{j}$, $S_{j}\subset S_{i}$ or $S_{i}\cap S_{j}=\emptyset$.
We say that $S_{i}$ is \emph{maximal }if there is no $j\neq i$ where
$S_{j}\supset S_{i}$. Let $S'_{1},\dots,S'_{t'}$ be all the maximal
sets. By the maximality, we have that $S'_{1},\dots,S'_{t'}$ are
mutually disjoint, and $|E^{s}|\le vol(\bigcup_{i=1}^{t'}S'_{i})$.
Let $H'_{1},\dots,H'_{t'}$ be the corresponding subgraphs of $G$
where $\partial_{H'_{i'}}(S'_{i'})\subset E^{s}$. By choosing the
appropriate ordering, we have write $H'_{1}=G$, $H'_{2}=G[V-S'_{1}],\dots,H'_{t'}=G[V-\bigcup_{i=1}^{t'-1}S_{i}']$.
\begin{claim}
$\delta_{G}(\bigcup_{i=1}^{t'}S'_{i})\le\sum_{i=1}^{t'}\delta_{H'_{i}}(S'_{i})$\end{claim}
\begin{proof}
We will prove that $\partial_{G}(\bigcup_{i=1}^{t'}S'_{i})\subseteq\bigcup_{i=1}^{t'}\partial_{H'_{i}}(S'_{i})$.
Let $(u,v)\in\partial_{G}(\bigcup_{i=1}^{t'}S'_{i})$. Suppose that
$u\in S'_{j}$ for some $j$. Then $v\in V-\bigcup_{i=1}^{t'}S'_{i}\subseteq V(H'_{j})-S'_{j}$
because $V(H'_{j})=V-\bigcup_{i=1}^{j-1}S'_{i}$. Therefore, $(u,v)\in\partial_{H'_{j}}(S'_{j})\subset\bigcup_{i=1}^{t'}\partial_{H'_{i}}(S'_{i})$.
\end{proof}
By the two claims above, we have that $\phi_{G}(\bigcup_{i=1}^{t'}S'_{i})<\alpha_{b}/2$.
This is because 
\begin{eqnarray*}
\phi_{G}(\bigcup_{i=1}^{t'}S'_{i}) & = & \frac{\delta_{G}(\bigcup_{i=1}^{t'}S'_{i})}{|\bigcup_{i=1}^{t'}S'_{i}|}=\frac{\delta_{G}(\bigcup_{i=1}^{t'}S'_{i})}{\sum_{i=1}^{t'}|S'_{i}|}\le\frac{\sum_{i=1}^{t'}\delta_{H'_{i}}(S'_{i})}{\sum_{i=1}^{t'}|S'_{i}|}\\
 & \le & \max_{i\le t'}\frac{\delta_{H'_{i}}(S'_{i})}{|S'_{i}|}=\max_{i\le t'}\phi_{H'_{i}}(S'_{i})<\alpha_{b}/2.
\end{eqnarray*}
Therefore, by \ref{thm:sparse cut small} we have:
\begin{claim}
$|\bigcup_{i=1}^{t'}S'_{i}|\le4|A|/\alpha_{b}$. \label{claim:maximal small}
\end{claim}
So $|E^{s}|\le vol(\bigcup_{i=1}^{t'}S'_{i})\le\Delta|\bigcup_{i=1}^{t'}S'_{i}|\le4\Delta|A|/\alpha_{b}$.
\end{proof}
Next, we would like to prove an important invariant given that $\phi(G_{b})\ge\alpha_{b}$:
if $\decomp(H,I,\ell)$ is called, then $\opt(I,\alpha_{\ell})<\bar{s}_{\ell}$.
In order to prove this, we need two lemmas.
\begin{lem}
Suppose that $\phi(G_{b})\ge\alpha_{b}$. If $\decomp(H,I,1)$ is
called, then $\opt(I,\alpha_{1})<4|A|/\alpha_{b}+1=\bar{s}_{1}$.\label{lem:invaraint 1}\end{lem}
\begin{proof}
Let $S'_{1},\dots,S'_{t'}$ be all the maximal sets as defined in
the proof of \ref{thm:local decomp sparse }. If $\decomp(H,I,1)$
is called, we know that either $V_{I}\subset\bigcup_{i=1}^{t'}S'_{i}$
or $V_{I}\cap\bigcup_{i=1}^{t'}S'_{i}=\emptyset$. 

First, if $V_{I}\subset\bigcup_{i=1}^{t'}S'_{i}$, then $|V_{I}|\le|\bigcup_{i=1}^{t'}S'_{i}|\le4|A|/\alpha_{b}$
as shown in \ref{claim:maximal small}. So $\opt(I,\alpha_{1})<4|A|/\alpha_{b}+1=\bar{s}_{1}$.
Second, if $V_{I}\cap\bigcup_{i=1}^{t'}S'_{i}=\emptyset$, let $S$
be any $\alpha_{1}$-sparse cut in $I$ where $|S|\le|V_{I}|/2$.
We can similarly show, as in \ref{thm:local decomp sparse }, that
$\phi_{G}(S\cup\bigcup_{i=1}^{t'}S'_{i})<\alpha_{b}/2$ and again
have that $|S\cup\bigcup_{i=1}^{t'}S'_{i}|\le4|A|/\alpha_{b}$. So
$\opt(I,\alpha_{1})<4|A|/\alpha_{b}+1=\bar{s}_{1}$ again.\end{proof}
\begin{lem}
For any subgraph $H=(V_{H},E_{H})$ of $G$ induced by $V_{H}$ and
$\alpha'<\alpha_{b}/2$, let $A_{H}$ is the set of endpoints of edges
in $\partial_{G}(V_{H})$ and $B_{H}=(A\cup A_{H})\cap V_{H}$ . If
$\phi(G_{b})\ge\alpha_{b}$, then any $\alpha'$-sparse cut $S\subset V_{H}$
where $|S|\le|V_{H}-S|$ must be $(B_{H},\sigma)$-overlapping in
$H$. That is, if $\phi(G_{b})\ge\alpha_{b}$, then $\opt(H,\alpha')=\opt(H,\alpha',B_{H},\sigma)$.
\label{thm:sparse must overlap}\end{lem}
\begin{proof}
First, consider any cut edge $(u,v)\in\partial_{G_{b}}(S)$ in the
initial graph $G_{b}$ where $u\in S$. We claim that either $(u,v)\in\partial_{H}(S)$
or $u\in A\cup A_{H}$. Indeed, if $u\notin A\cup A_{H}$, i.e. $u$
is not incident to any edge in $D$ nor $\partial_{G}(V_{H})$, all
edges incident to $u$ is inside $H$, and hence $(u,v)\in\partial_{H}(S)$.
Since there are at most $\Delta|S\cap(A\cup A_{H})|$ edges incidents
to all nodes $u\in S\cap(A\cup A_{H})$. It follows that $\delta_{G_{b}}(S)\le\delta_{H}(S)+\Delta|S\cap(A\cup A_{H})|$.

Suppose that there is an $\alpha'$-sparse cut $S\subset V_{H}$ which
is not $(B_{H},\sigma)$-overlapping, i.e. $|S\cap(A\cup A_{H})|<\sigma|S|$.
Then we have that
\[
\delta_{H}(S)\ge\delta_{G_{b}}(S)-\Delta|S\cap(A\cup A_{H})|>\alpha_{b}|S|-\sigma\Delta|S|=\alpha_{b}|S|-\alpha_{b}|S|/2=\alpha_{b}|S|/2.
\]
That is, $\phi_{H}(S)\ge\alpha_{b}/2>\alpha'$, which is a contradiction.
\end{proof}
Now, we can prove \textbf{the main invariant}:
\begin{lem}
\label{prop:invariant local decomp}Suppose that $\phi(G_{b})\ge\alpha_{b}$.
If $\decomp(H,I,\ell)$ is called, then the invariant $\opt(I,\alpha_{\ell})<\bar{s}_{\ell}$
is satisfied. \end{lem}
\begin{proof}
When $\ell=1$, $\opt(I,\alpha_{\ell})<\bar{s}_{1}$ by \ref{lem:invaraint 1}.
In particular, the invariant is satisfied when $\decomp(G,G,1)$ or
$\decomp(H[S],H[S],1)$ is called. The invariant $\decomp(H[V_{H}\setminus S],I,\ell)$
is the same as the one for $\decomp(H,I,\ell)$, and hence is satisfied
by induction. 

Finally, we claim that the invariant is satisfied when $\decomp(H,H,\ell+1)$
is called, i.e., $\opt(H,\alpha_{\ell+1})<\bar{s}_{\ell+1}$. By Step
5.a, $|S|<\bar{s}_{\ell+1}/c_{size}$. By Step 5, $\opt(H,\alpha_{\ell+1},B_{H},\sigma)/c_{size}\le|S|$
as $\alpha_{\ell+1}=\alpha_{\ell}/c_{exp}$. Since $H$ is induced
by $V_{H}$ and $\alpha_{\ell+1}\le\alpha_{1}<\alpha_{b}/2$ satisfying
the conditions in \ref{thm:sparse must overlap}, we have $\opt(H,\alpha_{\ell+1})=\opt(H,\alpha_{\ell+1},B_{H},\sigma)$.
Therefore, $\opt(H,\alpha_{\ell+1})\le c_{size}|S|<\bar{s}_{\ell+1}$
as desired. 
\end{proof}
Finally, we bound the expansion of the components of $G^{d}$.
\begin{lem}
\label{lem:expansion local decomp}Suppose that $\phi(G_{b})\ge\alpha_{b}$.
Each connected component $C$ in $G^{d}$ either is a singleton or,
otherwise, $\phi(C)\ge\alpha$.\end{lem}
\begin{proof}
$C$ can be returned as a connected component in $G^{d}$ in either
Step 2, 3 or 4 in \ref{alg:local decomp}. 

If $C$ is returned in Step 2, then $C$ is a connected component
in a graph induced by $E_{H}^{d}$ returned by $\cA_{decomp}(H,\alpha,p')$.
By \ref{thm:high-exp-decomp}, if $C$ is not a singleton, then $\phi(C)\ge\alpha$.

If $C$ is returned in Step 3, then we know $\decomp(H,H,B,L)$ was
called, and $C=H$. By the invariant, we have $\opt(H,\alpha_{L})<\bar{s}_{L}\le1$,
i.e. there is no $\alpha_{L}$-sparse cut in $H$. As $\alpha_{L}=\alpha$,
$\phi(H)\ge\alpha$.

If $C$ is returned in Step 4, then we know $C=H$. If $H$ is not
a singleton, then $\cA_{cut}(H,\alpha_{\ell},B,\sigma)$ reports that
there is no $(\alpha_{\ell}/c_{exp})$-sparse $(B_{H},\sigma)$-overlapping
cut in $H$ where $\ell<L$. By \ref{thm:sparse must overlap}, there
is actually no $(\alpha_{\ell}/c_{exp})$-sparse cut in $H$, i.e.
$\phi(H)\ge\alpha_{\ell}/c_{exp}\ge\alpha_{L}=\alpha$.
\end{proof}
We conclude the correctness of the algorithm from \ref{lem:expansion local decomp}
and \ref{thm:local decomp sparse } the following: 
\begin{cor}
Assume that $\cA_{decomp}$ is deterministic. Suppose that $\phi(G_{b})\ge\alpha_{b}$.
$\decomp(G,G,1)$ outputs $E^{s}$ such that $|E^{s}|\le4\Delta|A|/\alpha_{b}$,
and each connected component $C$ of $G^{d}$ is either a singleton
or has high expansion: $\phi(C)\ge\alpha$.\label{cor:local decomp correct}
\end{cor}
Now, it is left to analyze the running time.

\subsubsection{Running time}

In this section, we prove that if $\phi(G_{b})\ge\alpha_{b}$, then
$\decomp(G,G,1)$ takes at most $\bar{t}$ time. In other words, if
$\decomp(G,G,1)$ takes more that $\bar{t}$ time, then $\phi(G_{b})<\alpha_{b}$.

To analyze the running time, we need some more notation. We define
the recursion tree $\mathcal{T}$ of $\decomp(G,G,1)$ as follows:
\begin{itemize}
\item Each node of $\mathcal{T}$ represents the parameters of the procedure
$\decomp(H,I,\ell)$. 
\item $(G,G,1)$ is the root node. 
\item For each $(H,I,\ell)$, if $\decomp(H,I,\ell)$ returns a component
either in Step 2,3, or 4 \ref{alg:local decomp}, then $(H,I,\ell)$
is a leaf. 
\item If $\decomp(H,I,\ell)$ recurses on $\decomp(H[S],H[S],1)$ and $\decomp(H[V_{H}\setminus S],I,\ell)$,
then $(H[S],H[S],1)$ and $(H[V_{H}\setminus S],I,\ell)$ is a left
and right children of $(H,I,\ell)$ respectively, and the edges $((H,I,\ell),(H[S],H[S],1))$
and $((H,I,\ell),(H[V_{H}\setminus S],I,\ell))$ are \emph{left edge}
and \emph{right edge} respectively. 
\item If $\decomp(H,I,\ell)$ recurses on $\decomp(H,H,\ell+1)$, then $(H,H,\ell+1)$
is the only child of $(H,I,\ell)$, and the edge $((H,I,\ell),(H,H,\ell+1))$
is a \emph{down edge}. \end{itemize}
\begin{lem}
\label{lem:depth implies time local}Suppose that $\phi(G_{b})\ge\alpha_{b}$.
If the depth of the recursion tree $\mathcal{T}$ is $D$, then the
total running time of $\decomp(G,G,1)$ is $\tilde{O}(\frac{\Delta^{4.5}|A|^{1.5}}{\alpha_{b}^{3}}\gamma D\log\frac{1}{p})$
where $\gamma=n^{O(\sqrt{\log\log n/\log n}})$.\end{lem}
\begin{proof}
The running time on $\decomp(H,I,\ell)$, excluding the time spent
in the next recursion level, is contributed by either 1) the running
time of $\cA_{decomp}(H,\alpha,p')$ when $|V_{H}-B_{H}|<\frac{3}{\sigma}|B_{H}|$,
2) the running time of $\cA_{cut}(H,\alpha_{\ell},B_{H},\sigma)$,
and 3) the time for outputting $H$ as the component of $G^{d}$ when
$H$ is not the largest component. 

The total running time of the third case is bounded by $O(\Delta|\bigcup_{i=1}^{t'}S'_{i}|)=O(\Delta|A|/\alpha_{b})$
by \ref{claim:maximal small} where $\{S'_{i}\}_{i}$ are the \emph{maximal
}set defined in \ref{thm:local decomp sparse }. It is left to bound
the running time of the first and second cases.

In the first case, we have that $|E_{H}|\le2\Delta|V_{H}|=O(\frac{\Delta}{\sigma}|B_{H}|)$.
So, by \ref{thm:high-exp-decomp}, the running time of $\cA_{decomp}(H,\alpha,p')$
is $O(|E_{H}|\gamma\log\frac{1}{p'})=\tilde{O}(\frac{\Delta}{\sigma}|B_{H}|\gamma\log\frac{1}{p})$
because $p'=p/n^{2}$ and $\gamma=n^{O(\sqrt{\log\log n/\log n}})=n^{o(1)}$.
In the second case, the running time of $\cA_{cut}(H,\alpha_{\ell},B_{H},\sigma)$
is $\tilde{O}((\Delta|B_{H}|/\sigma)^{1.5})$ by \ref{thm:local MSO cut}.
In either case, the running time is at most $\tilde{O}((\Delta|B_{H}|/\sigma)^{1.5}\gamma\log\frac{1}{p})$. 

Let $\mathcal{T}_{d}$ be the set of all the nodes in $\mathcal{T}$
of depth \emph{exactly} $d$. Let $A^{s}$ be the endpoints of edges
in $E^{s}$. We can see that, for any two nodes $(H,I,\ell),(H',I',\ell')\in\mathcal{T}_{d}$,
$H$ and $H'$ are disjoint subgraphs of $G$, and so $B_{H}$ and
$B_{H'}$ are disjoint subset of $A\cup A^{s}$. So, the total running
time we spent on $\mathcal{T}_{d}$ is $\tilde{O}((\Delta|A\cup A^{s}|/\sigma)^{1.5}\gamma\log\frac{1}{p})$.
Hence, the total running time of $\mathcal{T}$ is $\tilde{O}((\Delta|A\cup A^{s}|/\sigma)^{1.5}\gamma D\log\frac{1}{p})$.
As $\sigma=\alpha_{b}/2\Delta$ and $|A^{s}|=O(\Delta|A|/\alpha_{b})$
by \ref{thm:local decomp sparse }, the running time is 
\[
\tilde{O}((\frac{\Delta(\Delta|A|/\alpha_{b})}{\alpha_{b}/2\Delta})^{1.5}\gamma D\log\frac{1}{p})=\tilde{O}(\frac{\Delta^{4.5}|A|^{1.5}}{\alpha_{b}^{3}}\gamma D\log\frac{1}{p}).
\]

\end{proof}
Now, we bound the depth $D$ of $\cT$. Recall $L\le1/\epsilon$,
$c_{size}=3/\sigma=6\Delta/\alpha_{b}$ and $\bar{s}_{1}=4|A|/\alpha_{b}+1$.
\begin{lem}
\label{lem:bound depth local}Suppose that $\phi(G_{b})\ge\alpha_{b}$.
Let $P$ be any path from leaf to root of $\mathcal{T}$. $P$ contains
at most $\log n$ left edges, $L\log n$ down edges, and $L\log n\times k$
right edges where $k=c_{size}\bar{s}_{1}^{\epsilon}$. That is, the
depth of $\cT$ is $D=L\log n\times c_{size}\bar{s}_{1}^{\epsilon}=\tilde{O}(\frac{\Delta|A|^{\epsilon}}{\alpha_{b}^{1+\epsilon}\epsilon})$.\end{lem}
\begin{proof}
Consider the left edge $((H,I,\ell),(H[S],H[S],1))$. Since $|S|\le|V_{H}|/2$,
$P$ contains $\log n$ left edges. Between any two left edges in
$P$ there are at most $L$ down edges because Step 3 in \ref{alg:local decomp}
always terminates the recursion when $\ell=L$. So $P$ contains at
most $L\log n$ down edges. To prove that there are at most $L\log n\times k$
right edges in $P$, it suffices to prove that there cannot be $\gamma$
right edges between any left edges or down edges in $P$.

Suppose that $(H_{1},I,\ell),\dots,(H_{\gamma},I,\ell)$ are nodes
in $P$ where, for each $i$, $((H_{i},I,\ell),(H_{i+1},I,\ell))$
is a right edge, and $(H_{1},I,\ell)$ is a deeper endpoint of left
edges or a down edges (hence $I=H_{1}$).

For each $i$, let $S_{i}$ be the cut such that $H_{i+1}=H_{i}[V_{H_{i}}\setminus S_{i}]$
and $\phi_{H_{i}}(S_{i})<\alpha_{\ell}$. Since $\{S_{i}\}_{i\le k}$
are mutually disjoint. We conclude $\phi_{H_{1}}(\bigcup_{i=1}^{k}S_{i})<\alpha_{\ell}$
using the same argument as in \ref{thm:local decomp sparse }. However,
we also have that $|S_{i}|\ge\bar{s}_{\ell+1}/\bar{c}_{size}$, for
all $i$, and hence $|\bigcup_{i=1}^{k}S_{i}|
\ge k\bar{s}_{\ell+1}/\bar{c}_{size}\ge\bar{s}_{1}^{\epsilon}\bar{s}_{\ell+1}=\bar{s}_{\ell}$. So $\bigcup_{i=1}^{k}S_{i}$ contradicts the invariant for $\decomp(H_{1},I,\ell)$,
where $I=H_{1}$, which says $\opt(H_{1},\alpha_{\ell})<\bar{s}_{\ell}$.
Note that the invariant must hold by \ref{prop:invariant local decomp}.\end{proof}
\begin{cor}
Suppose that $\phi(G_{b})\ge\alpha_{b}$. $\decomp(G,G,1)$ runs in
time $\bar{t}=\tilde{O}(\frac{\Delta^{5.5}|A|^{1.5+\epsilon}}{\alpha_{b}^{4+\epsilon}\epsilon}\gamma\log\frac{1}{p})$.\label{cor:local decomp time}\end{cor}
\begin{proof}
We have $D=\tilde{O}(\frac{\Delta|A|^{\epsilon}}{\alpha_{b}^{1+\epsilon}\epsilon})$
by \ref{lem:bound depth local}. By \ref{lem:depth implies time local},
we have that $\decomp(G,G,1)$ runs in time 
\[
\tilde{O}(\frac{\Delta^{4.5}|A|^{1.5}}{\alpha_{b}^{3}}\gamma D\log\frac{1}{p})=\tilde{O}(\frac{\Delta^{5.5}|A|^{1.5+\epsilon}}{\alpha_{b}^{4+\epsilon}\epsilon}\gamma\log\frac{1}{p})
\]
. 
\end{proof}
Now, we can conclude the main theorem.
\begin{proof}
[Proof of \ref{thm:local decomp}] From \ref{sec:Validity}, we have
that in time $\bar{t}=\tilde{O}(\frac{\Delta^{5.5}|A|^{1.5}n^{\epsilon}}{\alpha_{b}^{4}\epsilon}\gamma)$
$\decomp(G,G,1)$ outputs $E^{s}\subset E$ and all components, except
the largest one, of $G^{d}$ where $G^{d}=(V,E^{d})=(V,E-E^{s})$.
Otherwise, by \ref{cor:local decomp time}, we report $\phi(G_{b})<\alpha_{b}$.

When we assume that $\cA_{decomp}$ is deterministic. If $\phi(G_{b})\ge\alpha_{b}$,
by \ref{cor:local decomp correct}, we have $|E^{s}|\le4\Delta|A|/\alpha_{b}$
and each connected component $C$ of $G^{d}$ is either a singleton
or has high expansion $\phi(C)\ge\alpha$. To remove the assumption
that $\cA_{decomp}$ is deterministic, and show that with probability
$1-p$, the outputted decomposition is correct. Observe that, as the
depth of the recursion tree $\mathcal{T}$ of $\decomp(G,G,1)$ is
$D=L\log n\cdot c_{size}\bar{s}_{1}^{\epsilon}$, then $\cA_{decomp}$
is called at most $nD$ times. So the probability that $\cA_{decomp}$
always works correctly is at least $1-p'nD\ge1-p$.\end{proof}

%% file: monte_carlo.tex
\section{Monte Carlo Dynamic $\protect\st$ with $O(n^{0.4+o(1)})$ Update
Time \label{sec:monte_carlo}}

In this section, the goal is to prove:
\begin{thm}
There is a randomized dynamic $\st$ algorithm for any graphs with
$n$ nodes and $m$ initial edges that works correctly against adaptive
adversaries\emph{ }with high probability and has preprocessing time
$O(m^{1+o(1)})$ and worst-case update time $O(n^{0.4+o(1)})$ where
the term $o(1)=O(\sqrt{\log\log n/\log n})$.
\end{thm}
By \ref{thm:classic reduc mc}, we only need to prove the lemma below.
We note that the preprocessing time $t_{p}(n,p)$ below does \emph{not}
depend on the length of update sequence $T$.
\begin{lem}
\label{thm:dynST-mc}There is a randomized dynamic $\st$ algorithm
for $3$-bounded degree graphs with $n$ nodes and $\Theta(n)$ edges
that works correctly against adaptive adversaries\emph{ }with probability
$1-p$ for the first $T$ updates and has preprocessing time $t_{p}(n,p)=O(n^{1+o(1)}\log\frac{1}{p})$
and worst-case update time $t_{u}(n,p,T)$ such that $t_{u}(n,p,n^{0.6})=O(n^{0.4+o(1)}\log\frac{1}{p})$.
The term $o(1)=O(\sqrt{\log\log n/\log n})$ in both preprocessing
and update time.
\end{lem}

Throughout this section, we let $\beta=\beta(n)=2^{O(\log\log n)^{3}}=n^{o(1)}$
and $\gamma=\gamma(n)=n^{O(\sqrt{\log\log n/\log n})}=n^{o(1)}$ be
the function depending on $n$ from \ref{thm:cut recovery tree} and
\ref{thm:high-exp-decomp} respectively.

\subsection{Dynamic $\protect\st$ on Expanders\label{sec:dynST on expander}}

The following lemma shows a decremental $\st$ algorithm that runs
on a graph which is initially an expander. It exploits a cut recovery
tree from \ref{sec:linear sketches}.
\begin{lem}
There is a decremental $\st$ algorithm $\cA$ for $3$-bounded degree
graphs with $n$ nodes and has initial expansion at least $\alpha$,
that works correctly against adaptive adversaries\emph{ }with probability
$1-p$ for the first $T$ updates and $\cA$ has preprocessing time
$\tilde{O}(n\beta)$ and worst-case update time $\tilde{O}((1/\alpha+\sqrt{T/\alpha})\beta\log\frac{1}{p})$.\label{thm:mc on expander}
\end{lem}
Let $G=(V,E)$ be the initial graph. $\cA$ preprocesses $G$ as follows:
1) find any spanning tree $F$ of $G$. Then initialize both an augmented
ET tree $\cE$ and a cut recovery tree $\cD$ with parameter $k=\sqrt{T/\alpha}$
on $(G,F)$. By \ref{thm:aug ET tree} and \ref{thm:cut recovery tree},
this takes time $O(|E|+|E|\beta\log n)=\tilde{O}(n\beta)$. 

Next, $\cA$ handles edge deletion as follows. If the deleted edge
$e$ is not in tree edge, then just remove $e$ from $G$. Otherwise,
if $e=(u,v)$ is a tree edge of $F$ then we try to find a replacement
edge $e'$ for reconnecting the two separated components of $F$ as
described in \ref{alg:replace-mc} and, if $e'$ is returned, add
$e'$ into $F$. We set $\sigma=O((\frac{1}{\alpha}+\sqrt{\frac{T}{\alpha}})\log\frac{T}{p})$.
\begin{algorithm}[h]
\begin{enumerate}
\item Using $\cE$, set $T_{u}=\textsf{find\_tree}(u)$ and $T_{v}=\textsf{find\_tree}(v)$
where $\textsf{tree\_size}(T_{u})<\textsf{tree\_size}(T_{v})$.
\item Sample independently $\sigma$ non-tree edges in $G\setminus F$ incident
to $T_{u}$ using $\cE$. If there is a sampled edge $e'$ from $\partial_{G}(V(T_{u}))$,
then RETURN $e'$.
\item Query $\cutset k(T_{u})$ to $\cD$. If the returned set $\tilde{\partial}\neq\emptyset$,
then RETURN a random edge $e'\in\tilde{\partial}$.
\item RETURN ``no replacement edge''
\end{enumerate}
\caption{The algorithm for finding a replacement edge $e'$ if exists, after
deleting an edge $e=(u,v)\in F$. \label{alg:replace-mc}}
\end{algorithm}

\begin{claim}
\ref{alg:replace-mc} runs in time $\tilde{O}((1/\alpha+\sqrt{T/\alpha})\beta\log\frac{1}{p})$.\label{thm:run time prep mc expander}\end{claim}
\begin{proof}
Step 1 takes time at most $O(\log n)$. Step 2 takes time $O(\sigma\log n)=O((\frac{1}{\alpha}+\sqrt{\frac{T}{\alpha}})\log\frac{1}{p}\log n))$
by \ref{thm:aug ET tree}. Step 3 takes time $O(k\beta\log^{2}n)=O(\sqrt{\frac{T}{\alpha}}\beta\log^{2}n)$
by \ref{thm:cut recovery tree}. So the running time in \ref{alg:replace-mc}
is at most $\tilde{O}((1/\alpha+\sqrt{T/\alpha})\beta\log\frac{1}{p})$.\end{proof}
\begin{claim}
\ref{alg:replace-mc} returns a replacement edge, if exists, with
probability $1-p$.\label{thm:correct mc expander}
\begin{proof}
Let $S=V(T_{u})$ where $T_{u}$ is defined in \ref{alg:replace-mc}.
If a replacement edge exists, then $\delta_{G}(S)\ge1$. If $1\le\delta_{G}(S)\le\sqrt{T/\alpha}$,
then by the cut recovery tree $\cD$, all cut edges $\partial_{G}(S)$
is returned, and hence so is a replacement edge.

Otherwise, if $\delta_{G}(S)>\sqrt{T/\alpha}$, then we claim that
$\phi_{G}(S)\ge\min\{\frac{\alpha}{2},\frac{1}{2\sqrt{T/\alpha}}\}$.
There are two cases. For the first case, if $|S|\ge2T/\alpha$, then
we claim that $\phi_{G}(S)\ge\alpha/2$. Indeed, let $G_{0}$ be an
initial graph. We have that $\delta_{G}(S)\ge\delta_{G_{0}}(S)-T\ge\alpha|S|-T\ge\alpha|S|-\frac{\alpha}{2}|S|=\frac{\alpha}{2}|S|$
where the third inequality is because $\phi(G_{0})\ge\alpha$. For
the second case, we have $|S|<2T/\alpha$ while $\delta_{G}(S)>\sqrt{T/\alpha}$
and therefore $\phi_{G}(S)\ge\frac{\sqrt{T/\alpha}}{2T/\alpha}=\frac{1}{2\sqrt{T/\alpha}}$.
By the property of the augmented ET-tree $\cE$, the probability that
a sampled edge is from $\partial_{G}(S)$ is $\frac{\delta_{G}(S)}{nt\_vol_{(G,F)}(S)}\ge\frac{\delta_{G}(S)}{3|S|}\ge\frac{\phi_{G}(S)}{3}$.
Now, the probability that all sampled $\sigma$ edges are not from
$\partial_{G}(S)$ is $(1-\frac{\phi_{G}(S)}{3})^{\sigma}\le(1-\frac{\max\{\frac{\alpha}{2},\frac{1}{2\sqrt{T/\alpha}}\}}{3})^{\sigma}\le p$,
by choosing the constant in $\sigma=O((\frac{1}{\alpha}+\sqrt{\frac{T}{\alpha}})\log\frac{1}{p})$.
\end{proof}
\end{claim}
Now, we conclude the proof of \ref{thm:mc on expander}.
\begin{proof}
[Proof of \ref{thm:mc on expander}]The preprocessing time of $\cA$
is $\tilde{O}(n\beta)$. It is clear that for each update, the bottleneck
is the running time of \ref{alg:replace-mc} for finding replacement
edge which takes $\tilde{O}((1/\alpha+\sqrt{T/\alpha})\beta\log\frac{1}{p})$
by \ref{thm:run time prep mc expander}, and each update $\cA$ correctly
find a replacement edge with probability $1-p$ by \ref{thm:correct mc expander}.
\end{proof}

\subsection{Reduction to Expanders via Global Expansion Decomposition}

The following reduction shows that, given a dynamic $\st$ algorithm
that runs on expanders, we can obtain another algorithm that needs
not run on expanders. The main idea is to use the global expansion
decomposition algorithm from \ref{thm:high-exp-decomp} to partition
edges on the graph into a collection of expanders and some small set
of the remaining edges. We then separately maintain a spanning tree
in each expander using the algorithm from \ref{thm:mc on expander},
and ``combine'' them with the remaining edges using 2-dim ET tree
from \ref{thm:2dim ET tree}.

For any $n$, $\alpha$, $p$, and $T$, for technical reason, let
$t_{p}(n,\alpha,p,T)$ and $t_{u}(n,\alpha,p,T)$ be some functions
such that, for any $n_{1},n_{2}$ where $n_{1}+n_{2}\le n$, we have
$t_{p}(n_{1},\alpha,p,T)+t_{p}(n_{2},\alpha,p,T)\le t_{p}(n,\alpha,p,T)$
and $t_{u}(n_{1},\alpha,p,T)\le t_{u}(n,\alpha,p,T)$. 
\begin{lem}
For any $n$, $\alpha$, $p$, and $T$, suppose that there is a decremental
$\st$ algorithm $\cA$ for $3$-bounded degree graphs with $n$ nodes
and initial expansion at least $\alpha$, that works correctly against
adaptive adversaries\emph{ }with probability $1-p$ for the first
$T$ updates and $\cA$ has preprocessing time $t_{p}(n,\alpha,p,T)$
and worst-case update time $t_{u}(n,\alpha,p,T)$. Then there is a
dynamic $\st$ algorithm $\cB$ for $3$-bounded degree graphs with
$n$ nodes that works correctly against adaptive adversaries with
probability $1-O(nTp)$ for the first $T$ updates and $\cB$ has
preprocessing time $O(t_{p}(n,\alpha,p,T)+n\gamma\log\frac{1}{p})$
and worst-case update time $O(t_{u}(n,\alpha,p,T)+\sqrt{n\gamma\alpha+T})$.\label{thm:reduc global decomp}
\end{lem}
Before proving \ref{thm:reduc global decomp}, we show that it implies
\ref{thm:dynST-mc}.
\begin{proof}
[Proof of \ref{thm:dynST-mc}]We set $\alpha=1/n^{0.2}$. By \ref{thm:mc on expander}
and \ref{thm:reduc global decomp}, there is dynamic $\st$ algorithm
$\cB$ for $3$-bounded degree graphs with $n$ nodes that works correctly
against adaptive adversaries with probability $1-p$ for the first
$T$ updates and $\cB$ has preprocessing time $t_{p}(n,p)=\tilde{O}(n\beta+n\gamma\log\frac{1}{p})=O(n^{1+o(1)}\log\frac{1}{p})$
and worst-case update time $t_{u}(n,p,T)=\tilde{O}((1/\alpha+\sqrt{T/\alpha})\beta\log\frac{1}{p}+\sqrt{n\gamma\alpha+T})$.
We can verify that $t_{u}(n,p,n^{0.6})=O(n^{0.4+o(1)}\log\frac{1}{p})$
and the term $o(1)=O(\sqrt{\log\log n/\log n})$ in both preprocessing
and update time.
\end{proof}
Now we prove \ref{thm:reduc global decomp}. Given the algorithm $\cA$
that runs on a graph with initial expansion at least $\alpha$ by
the assumption, for convenience, let $\cA'$ be the algorithms that
runs on a graph whose all connected components has initial expansion
at least $\alpha$. $\cA'$ works by simply running $\cA$ on each
component. 

The target algorithm $\cB$ preprocesses a graph $G=(V,E)$ as described
in \ref{alg:preprocess-mc}. The idea is to partition edges $E$ into
$E^{d}$ and $E^{s}$ using the global expansion decomposition algorithm
from \ref{thm:high-exp-decomp} with success probability $1-p$. By
adding the failure probability of $\cB$ by $p$ at the end of analysis,
we now assume that the decomposition succeeds. Then let $F^{d}$ be
a spanning forest of $G^{d}=(V,E^{d})$ maintained by $\cA'$. The
key invariant to maintain a spanning forest $F$ of $G$ such that
$F^{d}\subseteq F$.

\begin{algorithm}[h]
\begin{enumerate}
\item Run the expansion decomposition on $G$ with expansion parameter $\alpha$
and failure probability $p$ and obtain $G^{d}=(V,E^{d})$ and $G^{s}=(V,E^{s})$. 
\item Preprocess $G^{d}$ using $\cA'$ and obtain a spanning forest $F^{d}$.
\item Construct a spanning forest $F$ of $G$ where $F\supseteq F^{d}$.
\item Preprocess $(F\cup E^{s},F)$ using a 2-dim ET tree $\cD$. 
\end{enumerate}
\caption{The algorithm for preprocessing the initial graph $G=(V,E)$\label{alg:preprocess-mc}}
\end{algorithm}

\begin{claim}
The preprocessing time of $\cB$ is $O(t_{p}(n,\alpha,p,T)+n\gamma\log\frac{1}{p})$.\label{lem:prep time mc decomp}\end{claim}
\begin{proof}
By \ref{thm:high-exp-decomp}, Step 1 takes $O(n\gamma\log\frac{1}{p})$.
Let $n_{1},\dots,n_{k}$ be the size of all connected components in
$G^{d}$. We have $\sum_{i=1}^{k}n_{i}\le n$, so Step 2 takes $\sum_{i=1}^{k}t_{p}(n_{i},\alpha,p,T)\le t_{p}(n,\alpha,p,T)$.
Step 3 can be done easily in $O(n)$. Step 4 takes in $O(n+|E^{s}|)=O(n)$
time by \ref{thm:2dim ET tree}.\end{proof}
\begin{rem}
\label{rem:cut sketch update-1-2}In the following description of
how to handle the updates, we maintain the following invariants: 1)
the underlying graph and forest of $\cD$  is $(F\cup E^{s},F)$ and
2) $\cA'$ maintains a dynamic forest $F^{d}$ on the underlying graph
$G^{d}$. Therefore, so whenever $G^{d}$, $E^{s}$, $F^{d}$ or $F$
are changed, we feed the update to $\cA'$ and $\cD$ accordingly.
\end{rem}
To insert $e=(u,v)$, we always insert $e$ into $E^{s}$, and if
$u$ and $v$ are not connected in $F$ (this can be checked using
the operation $\textsf{find\_tree}$ of the 2-dim ET tree $\cD$),
then we also add $e$ into $F$. 

To delete $e=(u,v)$, if $e$ was a tree-edge, then we find a replacement
edge $e'$ for reconnecting the two separated components of $F$ as
described in \ref{alg:replace-reduc-expander} and, if $e'$ is returned,
add $e'$ into $F$.

\begin{algorithm}[h]
\begin{enumerate}
\item If $e\in F^{d}$ then feed the deletion to $\cA'$. If $\cA'$ returns
a replacement edge $e'\in E^{d}$ of $F^{d}$, then RETURN $e'$.
\item Using $\cD$, set $T_{u}=\textsf{find\_tree}(u)$ and $e'=\textsf{find\_edge}(T_{u})$. 
\item If $e'\neq\emptyset$, then RETURN $e'$. Else RETURN ``no replacement
edge''
\end{enumerate}
\caption{The algorithm for finding a minimum weight replacement edge $e'$
if exists, after deleting an edge $e=(u,v)\in F$. \label{alg:replace-reduc-expander}}
\end{algorithm}

\begin{claim}
$F$ is a spanning forest of $G$ such that $F\supseteq F^{d}$ throughout
the whole update sequences with probability $1-pnT$.\label{lem:correct mc decomp}\end{claim}
\begin{proof}
We first assume that $\cA'$ is deterministic and we will remove it
later. We prove that $F$ is a spanning forest of $G$ where $F\supseteq F^{d}$
after each edge deletion.

It is clear that $F\supseteq F^{d}$ because, after preprocessing,
we have $F\supseteq F^{d}$ and then, in Step 1 of \ref{alg:replace-reduc-expander},
whenever $\cA'$ adds an edge $e'$ into $F^{d}$, we always add $e'$
into $F$. We only need to argue that $F$ is a spanning forest of
$G$. Indeed, $F$ is a forest in $G$ because an edge $(u,v)$ is
added into $F$ only if $u$ and $v$ are not connected in $F$. 

Now, we claim that $F$ spans $G$. It suffices to prove that after
each deletion if a replacement edge exists, then \ref{alg:replace-reduc-expander}
can find one. Suppose that $e=(u,v)\in F$ is deleted. Consider the
first case where $e\notin F^{d}$. Let $T_{u}\cup(u,v)\cup T_{v}$
be a connected component in $F$ before deletion. Any replacement
edge $e'$, if exists, must be in $E^{s}(T_{u},T_{v})$, i.e. a set
of edges $(x,y)\in E^{s}$ where $x\in V(T_{u})$ and $y\in V(T_{u})$).
Indeed, $e'\notin E^{d}$ because, after deletion, $F^{d}$ still
spans all components in $G^{d}=(V,E^{d})$. So $e'\in E^{s}$, and
moreover $e'\in E^{s}(T_{u},T_{v})$ otherwise $e'$ cannot reconnect
$T_{u}$ and $T_{v}$. Now, Step 2 of \ref{alg:replace-reduc-expander}
will return an edge from $E^{s}(T_{u},T_{v})$. Consider the second
case where $e\in F^{d}$. Let $T_{u}^{d}\cup(u,v)\cup T_{v}^{d}$
be a connected component in $F^{d}$. Any replacement edge $e'$,
if exists, must be in $E^{d}(T_{u}^{d},T_{v}^{d})\cup E^{s}(T_{u},T_{v})$
for the similar reason. An edge $E^{d}(T_{u}^{d},T_{v}^{d})$ and
$E^{s}(T_{u},T_{v})$ will be found in Step 1 and Step 2 of \ref{alg:replace-reduc-expander}
respectively.

Finally, we remove the assumption that $\cA$ is deterministic. As
$\cA$ works correctly at each step with probability $1-p$, and $\cA'$
runs instances of $\cA$ on at most $n$ components of $G^{d}$ for
$T$ time steps, so $\cA'$ works correctly throughout the whole update
sequences with probability at least $1-pnT$. \end{proof}
\begin{claim}
The update time of $\cB$ is $O(t_{u}(n,\alpha,p,T)+\sqrt{n\gamma\alpha+T})$.\label{lem:update time mc decomp}\end{claim}
\begin{proof}
Observe the the bottleneck of the update time is the update time of
$\cA'$ and 2-dim ET tree $\cD$. Let $n_{1},\dots,n_{k}$ be the
size of all components in $G^{d}$. The update time of $\cA'$ is
at most $\max_{i=1}^{k}t_{u}(n_{i},\alpha,p,T)\le t_{u}(n,\alpha,p,T)$.
Let $E_{0}^{s}$ be the set $E^{s}$ right after preprocessing. We
have that $E^{s}\subset E_{0}^{s}\cup I$ where $I$ is the set of
inserted edges and we know $|I|\le T$ as there are at most $T$ updates.
As $E^{s}$ is the set of non-tree edges in the underlying graph of
$\cD$, each operation of $\cD$ can be done in time $O(\sqrt{|E^{s}|})=O(\sqrt{|E_{0}^{s}\cup I|})=O(\sqrt{n\gamma\alpha+T})$
by \ref{thm:high-exp-decomp}. So the total update time is $O(t_{u}(n,\alpha,p,T)+\sqrt{n\gamma\alpha+T})$.
\end{proof}
We conclude the proof of \ref{thm:reduc global decomp}.
\begin{proof}
[Proof of \ref{thm:reduc global decomp}]Given $\cA$, we have that
$\cB$ has preprocessing time $O(t_{p}(n,\alpha,p,T)+n\gamma\log\frac{1}{p})$
by \ref{lem:prep time mc decomp}, and update time $O(t_{u}(n,\alpha,p,T)+\sqrt{n\gamma\alpha+T})$
by \ref{lem:update time mc decomp}. By \ref{lem:correct mc decomp},
$\cB$ works correctly for \emph{all} updates with probability $1-pnT$.
In particular, $\cB$ works correctly for each update with the same
probability.
\end{proof}

\subsubsection{Reduction for Las Vegas algorithms}

Since the reduction in \ref{thm:reduc global decomp} is only for
Monte Carlo algorithm, next, we state a similar reduction for Las
Vegas algorithm which we will use in \ref{sec:las_vegas}.
\begin{lem}
For any $n$, $\alpha$, $p$ and $T$, suppose that there is a decremental
$\st$ algorithm $\cA$ for $3$-bounded degree graphs with $n$ nodes
that works correctly against adaptive adversaries\emph{ }with certainty
for the first $T$ updates. Moreover, suppose that $\cA$ has the
following properties: Given an initial graph $G_{b}$, $\cA$ takes
$t_{p}(n,\alpha,p,T)$ time to preprocess or report failure. Given
each edge deletion, $\cA$ takes $t_{p}(n,\alpha,p,T)$ worst-case
time to update or report failure. If $\phi(G_{b})\ge\alpha$, then,
after preprocessing and for each update, $\cA$ does not report failure
with probability $1-p$.

Then there is a dynamic $\st$ algorithm $\cB$ for $3$-bounded degree
graphs with $n$ nodes that works correctly against adaptive adversaries
with certainty for the first $T$ updates and $\cB$ has preprocessing
time $O(t_{p}(n,\alpha,p,T)+n\gamma\log\frac{1}{p})$ and worst-case
update time $O(t_{u}(n,\alpha,p,T)+\sqrt{n\gamma\alpha+T})$ with
probability $1-O(p)$.\label{thm:reduc global decomp lv}
\end{lem}
\ref{thm:reduc global decomp lv} is different from \ref{thm:reduc global decomp}
as follows: 1) $\cA$ and $\cB$ work correctly with certainty, 2)
the initial graph $G_{b}$ of $\cA$ may not have expansion $\alpha$,
but only when $\phi(G_{b})\ge\alpha$, then $\cA$ will not fail with
high probability. Since the reduction works in very similar way as
in \ref{thm:reduc global decomp}, we will just sketch the proof below.
\begin{proof}
[Proof sketch]We use the same algorithm as in \ref{thm:reduc global decomp}.
Given any graph $G=(V,E)$, the algorithm $\cB$ first partitions
edges $E$ into $E^{d}$ and $E^{s}$ using the global expansion decomposition
algorithm from \ref{thm:high-exp-decomp} with success probability
$1-p$. We use the algorithm $\cA$ to maintain a spanning tree on
each component of $G^{d}=(V,E^{d})$ as in \ref{thm:reduc global decomp}.
Whenever $\cA$ reports failure, we just restart \ref{alg:preprocess-mc}
and obtain a new expansion decomposition. Since the decomposition
succeeds with probability $1-p$, the initial underlying graph of
$\cA$, which is a component $C$ of $G^{d}$, has expansion at least
$\alpha$ with probability $1-p$. If $\phi(C)\ge\alpha$, $\cA$
does not fails with probability $1-p$ for each update. Therefore,
with probability $1-O(p)$, $\cB$ has preprocessing time and update
time as claimed, using the same reduction as in \ref{thm:reduc global decomp}.\end{proof}

%% file: las_vegas.tex
\section{Las Vegas Dynamic $\protect\st$ with $O(n^{0.49305+o(1)})$ Update
Time \label{sec:las_vegas}}

Throughout this section, let $\epsilon=0.0069459$ and $\epsilon'=\sqrt{\frac{2\epsilon}{3\epsilon+0.5}}=0.16332$.
The main result of this section is the following:
\begin{thm}
There is a randomized dynamic $\st$ algorithm for any graphs with
$n$ nodes and $m$ initial edges that works correctly against adaptive
adversaries\emph{ }with certainty and has preprocessing time $O(m^{1+o(1)})$
and worst-case update time $O(n^{1/2-\epsilon+o(1)})$ with high probability.
The term $o(1)=O(\sqrt{\log\log n/\log n})$.\label{thm:dynamic lv final}
\end{thm}
To this end, it is enough the prove the following. Recall the factor
$\gamma=n^{O(\sqrt{\log\log n/\log n})}=n^{o(1)}$ from \ref{thm:high-exp-decomp}
\begin{lem}
For any $n,$ $p$ and $T$, there is a decremental $\st$ algorithm
$\cA$ for $3$-bounded degree graphs with $n$ nodes that works correctly
against adaptive adversaries\emph{ }with certainty for the first $T$
updates. $\cA$ preprocesses an initial graph $G_{b}$ in time $t_{p}(n,p)=O(n)$.
Then, given each edge deletion, $\cA$ takes $t_{u}(n,p,T)$ worst-case
time to update or report failure, where $t_{u}(n,p,n^{1/2+\epsilon})=O(n^{1/2-\epsilon}\gamma\log\frac{1}{p})$.
If $\phi(G_{b})\ge1/n^{2\epsilon}$, then, for each update, $\cA$
does not report failure with probability $1-O(p)$.\label{thm:lv on expander final}\end{lem}
\begin{proof}
[Proof of \ref{thm:dynamic lv final}]We set the expansion parameter
in \ref{thm:reduc global decomp lv} as $\alpha_{b}=1/n^{2\epsilon}$.
Plugging \ref{thm:lv on expander final} into \ref{thm:reduc global decomp lv},
there is a dynamic $\st$ algorithm $\cB$ for $3$-bounded degree
graphs with $n$ nodes that works correctly against adaptive adversaries
with certainty for the first $T$ updates and $\cB$ has preprocessing
time $t'_{p}(n,p)=O(n+n\gamma\log\frac{1}{p})=O(n\gamma\log\frac{1}{p})$
and worst-case update time $t'_{u}(n,p,T)=O(t_{u}(n,p,T)+\sqrt{n\gamma\alpha_{b}+T})$
where $t'_{u}(n,p,n^{1/2+\epsilon})=O(n^{1/2-\epsilon}\gamma\log\frac{1}{p})$
with probability $1-O(p)$. Then, by plugging $\cB$ into \ref{thm:classic reduc lv},
we obtain \ref{thm:dynamic lv final}.
\end{proof}
The rest of this section is to prove \ref{thm:lv on expander final}.
Throughout this section, for any $n$, we also denote $U_{b}=n^{1/2+\epsilon}$
and $\alpha_{b}=1/n^{2\epsilon}$. Let $U_{a}$ be such that $U_{a}/\alpha_{b}^{1/\epsilon'}=n^{1/2-\epsilon}$
so $U_{a}=n^{1/2-\epsilon}\alpha_{b}^{1/\epsilon'}=n^{1/2-\epsilon-2\epsilon/\epsilon'}$.

\subsection{Dynamic $\protect\st$ on Broken Expanders via Local Expansion Decomposition\label{sec:dynST broken expander}}

The goal is the main technical lemma below, and then we will show
later in \ref{sec:reduc to broken} that the lemma below implies \ref{thm:lv on expander final}
by using a standard technique of maintaining two concurrent data structures
as in \ref{rem:get stable}. Recall the definition of augmented ET
tree and 2-dim ET tree from \ref{def:aug ET tree} and \ref{def:2d ET tree}
respectively.
\begin{lem}
\label{thm:emergency decomp}There is an algorithm $\cA$ that can
do the following:
\begin{itemize}
\item $\cA$ is given $n,p,G_{b}$,$\alpha_{b}$,$D$,$F_{0}$, $\cE$ and
$\cE^{2d}$ as inputs: $G_{b}=(V,E_{b})$ is a $3$-bounded degree
graphs with $n$ nodes. $p$ is a failure probability parameter. $\alpha_{b}$
is an expansion parameter. $D\subset E_{b}$ is set of edges in $G_{b}$
of size at most $U_{b}$. Let $G=(V,E_{b}-D)$ that $\cA$ will compute
the decomposition on. $F_{0}$ is some spanning forest of $G$. $\cE$
is an augmented ET tree initialized on $(G,F_{0})$. $\cE^{2d}$ is
an 2-dim ET tree initialized on $(F_{0},F_{0})$.
\item \textbf{Preprocessing:} After given the input, $\cA$ takes $O((\frac{U_{b}{}^{1.5+\epsilon'}}{\alpha_{b}^{4+\epsilon'}}\gamma+\frac{U_{b}}{\alpha_{b}^{1+1/\epsilon'}})\log\frac{1}{p})$
time to either 

\begin{enumerate}
\item reports failure, or 
\item outputs the list of edges to be added or removed from $F_{0}$ to
obtain another spanning forest $F$ of $G$. 
\end{enumerate}
\item Then, there is an online sequence of edge deletions of length $U_{a}$
given to $\cA$. 
\item \textbf{Update:} After each edge deletion, $\cA$ takes $O(\frac{U_{a}}{\alpha_{b}^{1/\epsilon'}}\log\frac{1}{p}+\sqrt{\frac{U_{b}}{\alpha_{b}}})=O(n^{1/2-\epsilon}\log\frac{1}{p})$
worst-case time and either 

\begin{enumerate}
\item reports failure, or 
\item outputs the list of edges to be added or removed from $F$ to obtain
new one spanning forest of $G$.
\end{enumerate}
\item If $\phi(G_{b})\ge\alpha_{b}$, then, after preprocessing and for
each update, $\cA$ does not report failure with probability $1-p$.
\end{itemize}
\end{lem}
We call the time $\cA$ takes before the update sequences \emph{preprocessing
time}, and the time $\cA$ needs for each edge deletion \emph{update
time}.
\begin{rem}
\label{rem:when no D }In \ref{thm:emergency decomp}, if $D=\emptyset$,
then we can substitute $U_{b}$ by $0$. This implies that $\cA$
takes no time to preprocess (and without failure), and for each update
in the sequence of length $U_{a}$, $\cA$ takes $O(\frac{U_{a}}{\alpha_{b}^{1/\epsilon'}}\log\frac{1}{p})=O(n^{1/2-\epsilon}\log\frac{1}{p})$
time to either report failure or update the spanning forest. 
\end{rem}
The intuition of \ref{thm:emergency decomp} is clearer when we assume
that the ``before'' graph $G_{b}$ has expansion $\phi(G_{b})\ge\alpha_{b}$.
Then the initial graph $G_{0}$ is a ``broken'' expander, i.e. $G_{0}=G_{b}-D$
is an expander with some edge deleted. \ref{thm:emergency decomp}
says that, with high probability, the algorithm $\cA$ can preprocess
in \emph{sublinear time}, and obtain a dynamic $\st$ algorithm that
can handle $U_{a}$ more deletions with update time $O(n^{1/2-\epsilon})$.
The most important point of the algorithm $\cA$ is that the preprocessing
time is sublinear. This can be done by exploiting the local expansion
decomposition algorithm from \ref{sec:local decomp}. 

Now, we are ready to describe the algorithm for \ref{thm:emergency decomp}.

\subsubsection{Preprocessing}

Recall that the inputs are $n,p,G_{b}$,$\alpha_{b}$,$D$,$F_{0}$,
$\cE$ and $\cE^{2d}$. Let $G=(V,E)=(V,E_{b}-D)$. We denote $(G_{\cE},F_{\cE})$
and $(G_{\cE^{2d}},F_{\cE^{2d}})$ the underlying graphs and forests
of $\cE$ and $\cE^{2d}$ respectively. In the beginning, we are given
$(G_{\cE},F_{\cE})=(G,F_{0})$ and $(G_{\cE^{2d}},F_{\cE^{2d}})=(F_{0},F_{0})$. 

First, we run the local expansion decomposition algorithm $\cA_{local}$
from \ref{thm:local decomp} with input $(n,p,3,G_{b},D,\alpha_{b})$
and parameter $\epsilon'$. If $\cA_{local}$ reports that $G_{b}$
has expansion less than $\alpha_{b}$, then we also report failure
and halt. Otherwise, $\cA_{local}$ return the set of edges $E^{s}\subset E$
in $G$ and the all components of $G^{d}$, except the largest one,
where $G^{d}=(V,E^{d})=(V,E-E^{s})$. By \ref{thm:local decomp},
if $\phi(G_{b})\ge\alpha_{b}$, then with probability $1-p$, each
connected components in $G^{d}$ has expansion at least $\alpha=(\alpha_{b}/18)^{1/\epsilon'}$.
Let $\sigma=O(\frac{1}{\alpha}\log\frac{1}{p})$. Throughout the sequence
of edge deletions, we will maintain the following invariant: $(G_{\cE},F_{\cE})=(G^{d},F^{d})$
where $F^{d}$ is a spanning forest of $G^{d}$, and $(G_{\cE^{2d}},F_{\cE^{2d}})=(E^{s}\cup F^{d},F)$
where $F$ is a spanning forest of $G$ and $F\supseteq F^{d}$.

The preprocessing algorithm as described in \ref{alg:preprocess-lv}
is simply to validate these invariants before the first update comes. 

\begin{algorithm}[h]
\begin{enumerate}
\item If $\textsf{count\_tree}()>U_{b}$ on $F_{\cE}$, then REPORT FAIL.
\item Run the local expansion decomposition algorithm $\cA_{local}$ with
input $(n,p,3,G_{b},D,\alpha_{b})$ and parameter $\epsilon'$. If
$\cA_{local}$ reports that $\phi(G_{b})<\alpha_{b}$, then REPORT
FAIL. Else, $\cA_{local}$ returns $E^{s}\subset E$ in $G_{0}$ and
all components, except the largest one, of $G^{d}=(V,E-E^{s})$. 
\item For each $e\in E^{s}$: //To update $(G_{\cE},F_{\cE})=(G^{d},F_{0}-E^{s})$
and $(G_{\cE^{2d}},F_{\cE^{2d}})=(E^{s}\cup F_{0},F_{0}-E^{s})$

\begin{enumerate}
\item If $e\in F_{0}$, then $\deltree(e)$ on $F_{\cE}$ and $F_{\cE^{2d}}$.
\item $\delnon(e)$ on $G_{\cE}$ and $\insnon(e)$ on $G{}_{\cE^{2d}}$
\end{enumerate}
\item While there is a component $T\in F_{\cE}$ which not marked \emph{big}:
//To update $(G_{\cE},F_{\cE})=(G^{d},F^{d})$ and $(G_{\cE^{2d}},F_{\cE^{2d}})=(E^{s}\cup F^{d},F^{d})$
where $F^{d}$ is a spanning forest of $G^{d}$

\begin{enumerate}
\item Sample independently $\sigma$ edges incident to $T$ in $G_{\cE}$. 
\item If there is a sampled edge $e'$ from $\partial_{G_{\cE}}(V(T))$,
then $\instree(e')$ into $F_{\cE}$ and $F_{\cE^{2d}}$. Otherwise,
REPORT FAIL.
\item Let $T'\in F_{\cE}$ be the new component containing $e'$. Let $C'$
be the component of $G^{d}$ containing $T'$. 
\item If $|V(T')|>|V(C')|/2$, then mark $T'$ as \emph{big}.
\end{enumerate}
\item While there is a component $T\in F_{\cE^{2d}}$ where $\textsf{find\_edge}(T)\neq\emptyset$:
//To update $(G_{\cE^{2d}},F_{\cE^{2d}})=(E^{s}\cup F^{d},F)$ where
$F$ is a spanning forest of $G$

\begin{enumerate}
\item Set $e'=\textsf{find\_edge}(T)$.
\item $\instree(e')$ into $F_{\cE^{2d}}$. 
\end{enumerate}
\end{enumerate}
\caption{The algorithm for preprocessing the initial graph $G=(V,E)$\label{alg:preprocess-lv}}
\end{algorithm}

Now, we analyze \ref{alg:preprocess-lv}.
\begin{claim}
The running time of \ref{alg:preprocess-lv} is $\tilde{O}((\frac{U_{b}{}^{1.5+\epsilon'}}{\alpha_{b}^{4+\epsilon'}}\gamma+\frac{U_{b}}{\alpha_{b}^{1+1/\epsilon'}})\log\frac{1}{p})$. \end{claim}
\begin{proof}
We will prove that even if \ref{alg:preprocess-lv} does not fail,
the running time is at most $\tilde{O}((\frac{U_{b}{}^{1.5+\epsilon'}}{\alpha_{b}^{4+\epsilon'}}\gamma+\frac{U_{b}}{\alpha_{b}^{1+1/\epsilon'}})\log\frac{1}{p})$.

Step 2 takes $\tilde{O}(\frac{U_{b}{}^{1.5+\epsilon'}}{\alpha_{b}^{4+\epsilon'}}\gamma\log\frac{1}{p})$
by \ref{thm:local decomp}. If \ref{alg:preprocess-lv} proceeds to
Step 3, we have $|E^{s}|=O(U_{b}/\alpha_{b})$ by \ref{thm:local decomp}.
Step 3 takes $|E^{s}|\times O(\log n+\sqrt{|E^{s}|})=O((U_{b}/\alpha_{b})^{1.5})$
by \ref{thm:aug ET tree} and \ref{thm:2dim ET tree} because as there
are at most $|E^{s}|$ non-tree edges in $G_{\cE^{2d}}$, 

In Step 4, observe that in every iteration of the while loop the number
of components $T$ in $F_{\cE}$ which is not marked as big always
decreases by one. Right before Step 4, there are at most $|E^{s}|+U_{b}$
such components because $F_{\cE}$ has at most $U_{b}$ components
after Step 1 and then there are at most $|E^{s}|$ tree edge deletions
in Step 3. So there are at most $|E^{s}|+U_{b}$ iterations. In each
iteration of the while loop in Step 4, it takes at most $\tilde{O}(\sigma)=\tilde{O}(\frac{1}{\alpha}\log\frac{1}{p})=\tilde{O}(\frac{1}{\alpha_{b}^{1/\epsilon'}}\log\frac{1}{p})$
in Step 4.a. Step 4.b takes at most $O(\sqrt{|E^{s}|})=O(\sqrt{U_{b}/\alpha_{b}})$
by \ref{thm:2dim ET tree}. Note that Step 4.c and 4.d can be done
in constant time because of the following. Using $\cE$, we get $|V(T')|=\textsf{tree\_size}(T)$.
We can also compute $|V(C')|/2$ where $C'$ be the component of $G^{d}$
containing $T'$. Let $u\in V(T)$. By \ref{rem:know components},
we know a component $C'$ of $G^{d}$ where $u\in V(C)$, and also
the size $|V(C')|$. In total, Step 4 takes $(|E^{s}|+U_{b})\times\tilde{O}(\frac{1}{\alpha_{b}^{1/\epsilon'}}\log\frac{1}{p}+\sqrt{\frac{U_{b}}{\alpha_{b}}})=O(\frac{U_{b}}{\alpha_{b}^{1+1/\epsilon'}}\log\frac{1}{p}+\frac{U_{b}^{1.5}}{\alpha_{b}^{1.5}})$.

In Step 5, again, in every iteration, the number of components $T$
in $F_{\cE^{2d}}$ decreases by one. Note that after Step 1, $F_{\cE^{2d}}=F_{\cE}=F_{0}$.
By the same argument as we used for analyzing Step 4, there are at
most $|E^{s}|+U_{b}$ iterations. Each iteration takes at most $O(\sqrt{|E^{s}|})=O(\sqrt{U_{b}/\alpha_{b}})$.
The total running time is $(|E^{s}|+U_{b})\times O(\sqrt{\frac{U_{b}}{\alpha_{b}}})=O(\frac{U_{b}^{1.5}}{\alpha_{b}^{1.5}})$. 

We conclude that the total running time of \ref{alg:preprocess-lv}
is $\tilde{O}((\frac{U_{b}{}^{1.5+\epsilon'}}{\alpha_{b}^{4+\epsilon'}}\gamma+\frac{U_{b}}{\alpha_{b}^{1+1/\epsilon'}})\log\frac{1}{p})$.
\end{proof}
Next, we show that \ref{alg:preprocess-lv} either reports failure
or validate the invariants.
\begin{claim}
If \ref{alg:preprocess-lv} does not report failure, then $(G_{\cE},F_{\cE})=(G^{d},F^{d})$
where $F^{d}$ is a spanning forest of $G^{d}$, and $(G_{\cE^{2d}},F_{\cE^{2d}})=(E^{s}\cup F^{d},F)$
where $F$ is a spanning forest of $G$ and $F\supseteq F^{d}$.\end{claim}
\begin{proof}
Suppose there is no failure. It is clear that $(G_{\cE},F_{\cE})$
is updated to become $(G^{d},F_{0}-E^{s})$ and $(G^{d},F^{d})$ in
Step 3 and 4 respectively. $(G_{\cE^{2d}},F_{\cE^{2d}})$ is updated
to become $(E^{s}\cup F_{0},F_{0}-E^{s})$, $(E^{s}\cup F^{d},F^{d})$,
and $(E^{s}\cup F^{d},F)$ in Step 3, 4, and 5 respectively.\end{proof}
\begin{claim}
If $\phi(G_{b})\ge\alpha_{b}$, then \ref{alg:preprocess-lv} report
failure with probability at most $O(p)$.\label{thm:does not restart}\end{claim}
\begin{proof}
In Step 1, if $\textsf{count\_tree}()>U_{b}$ on $F_{\cE}$, then
in means that $G=(V,E_{b}-D)$ has more than $U_{b}$ connected components
and so $G_{b}$ is not connected, i.e. $\phi(G_{b})=0$ which contradicts
the assumption that $\phi(G_{b})\ge\alpha_{b}$. In Step 2, the algorithm
cannot fail because $\cA_{local}$ reports failure only when $\phi(G_{b})<\alpha_{b}$
by \ref{thm:local decomp}.

In Step 4.a, we sample $\sigma$ edges incident to $T\in F_{\cE}$
which is not marked as big, i.e. $|V(T)|\le|V(C)|/2$ where $C$ is
the component of $G^{d}$ containing $T$. As $\phi(G_{b})\ge\alpha_{b}$,
we have $\phi(C)\ge\alpha$ with probability $1-O(p)$ by \ref{thm:local decomp}.
Assuming that $\phi(C)\ge\alpha$, we have $\phi_{C}(V(T))\ge\alpha$
and the probability that none of $\sigma$ sampled edges is a cut
edge is at most $(1-\frac{\delta_{C}(V(T))}{vol_{C}(V(T))})^{\sigma}\le(1-\frac{\alpha}{3})^{\sigma}\le p$.
By union bound, the algorithm fails with probability at most $O(p)$. 
\end{proof}
From the claims above can be summarized as follows:
\begin{lem}
\label{lem:lv prep conclude }After given the inputs, $\cA$ takes
$O((\frac{U_{b}{}^{1.5+\epsilon'}}{\alpha_{b}^{4+\epsilon'}}\gamma+\frac{U_{b}}{\alpha_{b}^{1+1/\epsilon'}})\log\frac{1}{p})$
time to either 
\begin{enumerate}
\item reports failure, or 
\item outputs the list of edges to be added or removed from $F_{0}$ to
obtain another spanning forest $F$ of $G$. Moreover, the underlying
graphs and forested of the given augmented ET tree $\cE$ and 2-dim
ET tree $\cE^{2d}$ are updated such that $(G_{\cE},F_{\cE})=(G^{d},F^{d})$
where $F^{d}$ is a spanning forest of $G^{d}$, and $(G_{\cE^{2d}},F_{\cE^{2d}})=(E^{s}\cup F^{d},F)$
where $F$ is a spanning forest of $G$ and $F\supseteq F^{d}$. 
\end{enumerate}

Moreover, if $\phi(G_{b})\ge\alpha_{b}$, $\cA$ does not report failure
with probability $1-O(p)$.

\end{lem}
Lastly, let us state the lemma in the corner case when $D=\emptyset$.
In this case, we claim that the invariant already holds even if we
do nothing. By setting $E^{s}=\emptyset$ and $F^{d}=F_{0}$, we have 
\begin{lem}
\label{lem:lv prep conclude no delete}After given the inputs where
$D=\emptyset$, the underlying graphs and forested of the given augmented
ET tree $\cE$ and 2-dim ET tree $\cE^{2d}$ are already such that
$(G_{\cE},F_{\cE})=(G_{b},F_{0})=(G^{d},F^{d})$ where $F^{d}$ is
a spanning forest of $G^{d}$, and $(G_{\cE^{2d}},F_{\cE^{2d}})=(F_{0},F_{0})=(E^{s}\cup F^{d},F)$
where $F$ is a spanning forest of $G$ and $F\supseteq F^{d}$. 
\end{lem}

\subsubsection{Handling Updates}

Throughout the update sequence, we keep the invariant, which is valid
after preprocessing, that $(G_{\cE},F_{\cE})=(G^{d},F^{d})$ where
$F^{d}$ is a spanning forest of $G^{d}$, and $(G_{\cE^{2d}},F_{\cE^{2d}})=(E^{s}\cup F^{d},F)$
where $F$ is a spanning forest of $G$ and $F\supseteq F^{d}$. 

The algorithm handles each edge deletion as standard: if the deleted
edge $e\notin F$, then just remove $e$ from $G$. That is, if $e\in E^{d}$,
$\delnon(e)$ from $G_{\cE}$, otherwise $e\in E^{s}$, then $\delnon(e)$
from $G_{\cE^{2d}}$. Now, if $e\in F$, then we try to find a replacement
edge $e'$ for reconnecting the two separated components of $F$ as
described in \ref{alg:replace-lv}. Suppose that $e'$ is returned.
If $e'\in E^{s}$, then $\instree(e')$ to $F_{\cE^{2d}}$. If $e'\in E^{d}$,
then $\instree(e')$ to $F_{\cE}$, $\insnon(e')$ to $G_{\cE^{2d}}$,
and $\instree(e')$ to $F_{\cE^{2d}}$. Because whenever we add an
edge to $F_{\cE^{2d}}$ then we also add to $F_{\cE}$, we conclude:
\begin{fact}
The invariant $F\supseteq F^{d}$ is maintained.\label{thm:supset maintain}
\end{fact}
The important part is how \ref{alg:replace-lv} works. Recall $\sigma=O(\frac{1}{\alpha}\log\frac{1}{p})$. 

\begin{algorithm}[h]
\begin{enumerate}
\item If $e\in F^{d}$ then 

\begin{enumerate}
\item Using $\cE$, set $T_{u}^{d}=\textsf{find\_tree}(u)$ and $T_{v}^{d}=\textsf{find\_tree}(v)$. 
\item If $\textsf{tree\_size}(T_{u}^{d})<\textsf{tree\_size}(T_{v}^{d})$,
then set $T^{d}=T_{u}^{d}$, otherwise $T^{d}=T_{v}^{d}$.
\item If $\textsf{tree\_size}(T^{d})>2U_{a}/\alpha$, then sample independently
$\sigma$ edges in $G_{d}$ incident to $T^{d}$. If there is a sampled
edge $e'$ from $\partial_{G^{d}}(V(T^{d}))$, then RETURN $e'$.
Otherwise, REPORT FAIL.
\item Else $\textsf{tree\_size}(T^{d})\le2U_{a}/\alpha$, then list all
edges in $G_{d}$ incident to $T^{d}$. If there is a listed edge
$e'$ from $\partial_{G^{d}}(V(T^{d}))$, then RETURN $e'$.
\end{enumerate}
\item Using $\cE^{2d}$, set $T_{u}=\textsf{find\_tree}(u)$ and $e'=\textsf{find\_edge}(T_{u})$.
\item If $e'\neq\emptyset$, then RETURN $e'$. Else RETURN ``no replacement
edge''
\end{enumerate}
\caption{The algorithm for finding a minimum weight replacement edge $e'$
if exists, after deleting an edge $e=(u,v)\in F$. \label{alg:replace-lv}}
\end{algorithm}

\begin{claim}
\ref{alg:replace-lv} runs in time $O(\frac{U_{a}}{\alpha}\log\frac{1}{p}+\sqrt{\frac{U_{b}}{\alpha_{b}}})=O(n^{1/2-\epsilon}\log\frac{1}{p}+n^{1/4+1.5\epsilon})=O(n^{1/2-\epsilon}\log\frac{1}{p})$. \end{claim}
\begin{proof}
We prove that even there is no failure reported. \ref{alg:replace-lv}
runs in time $O(\frac{U_{a}}{\alpha}\log\frac{1}{p}+\sqrt{\frac{U_{b}}{\alpha_{b}}})$.
Step 1.a and 1.b take time at most $O(\log n)$. Step 1.c takes time
$O(\sigma\log n)=O(\frac{1}{\alpha}\log\frac{1}{p}\log n))$ using
the $\textsf{sample}$ operation of $\cE$ by \ref{thm:aug ET tree}
if the algorithm does not restart. Step 1.d takes time $O(\frac{U_{a}}{\alpha})$
using the $\textsf{list}$ operation of $\cE$ by \ref{thm:aug ET tree}.
Step 2 takes time $O(\sqrt{|E^{s}|})=O(\sqrt{\frac{U_{b}}{\alpha_{b}}})$.
So the running time in \ref{alg:replace-mc} is at most $O(\frac{\log n}{\alpha}\log\frac{1}{p}+\frac{U_{a}}{\alpha}+\sqrt{\frac{U_{b}}{\alpha_{b}}})=O(\frac{U_{a}}{\alpha}\log\frac{1}{p}+\sqrt{\frac{U_{b}}{\alpha_{b}}})$.\end{proof}
\begin{claim}
After each edge deletion, $\cA$ reports failure or updates $\cE$
and $\cE^{2d}$ such that the invariant is maintained.\label{lem:correct mc decomp-1}\end{claim}
\begin{proof}
Suppose that $\cA$ does not report failure. We will prove that the
invariant is maintained.

$F^{d}$ and $F$ are forests in $G^{d}$ and $G$, respectively,
because after preprocessing they are forest, and after that we only
add replacement edges into them. Moreover, $F\supseteq F^{d}$ by
\ref{thm:supset maintain}. It remains to prove that $F^{d}$ and
$F$ span $G^{d}$ and $G$ respectively. That is, we prove that after
each deletion if a replacement edge exists, then \ref{alg:replace-lv}
can find one. 

We first prove this for $F^{d}$. Suppose that $e=(u,v)\in F^{d}$
is deleted. Let $T_{u}^{d}\cup(u,v)\cup T_{v}^{d}$ be a connected
component in $F^{d}$. It easy to see that $\partial_{G^{d}}(V(T^{d}))$
is the set of replacement edges in $E^{d}$. Suppose \ref{alg:replace-lv}
cannot find a replacement edge in $E^{d}$. This happens only in Step
1.d as we assume $\cA$ does not report failure. In this case, all
edges incident to $T_{u}^{d}$ are listed. So a replacement edge really
do not exist.

Next, we prove the case of $F$. Suppose that $e=(u,v)\in F$ is deleted.
We also denote $T_{u}\cup(u,v)\cup T_{v}$ a connected component in
$F$. It is easy to see that $\partial_{G^{s}}(V(T))$ is the set
of replacement edges in $E^{s}$. In Step 2, $\textsf{find\_edge}(T_{u})$
returns an edge iff $\partial_{G^{s}}(V(T))\neq\emptyset$. So if
we cannot get any edge from both $\partial_{G^{d}}(V(T^{d}))$ and
$\partial_{G^{s}}(V(T))$, then a replacement edge really do not exist.\end{proof}
\begin{claim}
If $\phi(G_{b})\ge\alpha_{b}$, then \ref{alg:replace-lv} reports
failure in Step 1.c with probability at most $O(p)$.\end{claim}
\begin{proof}
We use the same argument as in \ref{thm:does not restart}. Let $C$
be the component of $G^{d}$ containing $T^{d}$. As $\phi(G_{b})\ge\alpha_{b}$,
by \ref{thm:local decomp} $\phi(C)\ge\alpha$ with probability $1-O(p)$.
Now, we assume $\phi(C)\ge\alpha$ and prove that \ref{alg:replace-lv}
fail with probability at most $p$.

By Step 1.b, $|V(T^{d})|\le|V(C)|/2$, so $\delta_{C}(V(T))\ge\alpha|V(T)|-U_{a}$
because there are at most $U_{a}$ deletions. As $|V(T)|\ge2U_{a}/\alpha$,
we have $\delta_{C}(V(T))\ge\frac{\alpha}{2}|V(T)|$. Therefore, all
$\sigma$ sampled edges are not cut edge with probability $(1-\frac{\delta_{C}(V(T))}{vol_{C}(V(T))})\le(1-\frac{\alpha/2}{3})^{\sigma}\le p$. 
\end{proof}
We conclude with the following lemma:
\begin{lem}
\label{thm:lv update conclude}After each edge deletion, $\cA$ takes
time $O(\frac{U_{a}}{\alpha}\log\frac{1}{p}+\sqrt{\frac{U_{b}}{\alpha_{b}}})=O(n^{1/2-\epsilon}\log\frac{1}{p})$
and either 
\begin{enumerate}
\item reports failure, or 
\item outputs the list of edges to be added or removed from $F$ to obtain
new one spanning forest of $G$. Moreover, the underlying graphs and
forested of the given augmented ET tree $\cE$ and 2-dim ET tree $\cE^{2d}$
are updated such that $(G_{\cE},F_{\cE})=(G^{d},F^{d})$ where $F^{d}$
is a spanning forest of $G^{d}$, and $(G_{\cE^{2d}},F_{\cE^{2d}})=(E^{s}\cup F^{d},F)$
where $F$ is a spanning forest of $G$ and $F\supseteq F^{d}$. 
\end{enumerate}

Moreover, if $\phi(G_{b})\ge\alpha_{b}$, $\cA$ does not report failure
with probability $1-O(p)$ for each update.

\end{lem}
Combining \ref{lem:lv prep conclude } and \ref{thm:lv update conclude},
we immediately obtain \ref{thm:emergency decomp}. Moreover, \ref{lem:lv prep conclude no delete}
justifies \ref{rem:when no D } when $D=\emptyset$.

\subsection{Reduction to Broken Expanders\label{sec:reduc to broken}}

In this section, we prove \ref{thm:lv on expander final} by using
the algorithm from \ref{thm:emergency decomp} as a main subroutine
in a standard technique. Let $\cA$ be the algorithm for \ref{thm:emergency decomp}.
We would like to devise an algorithm $\cB$ for \ref{thm:lv on expander final}.

Let $G_{b}$ be the initial graph which is given to $\cB$ as inputs.
$\cB$ preprocess $G_{b}$ as described in \ref{alg:preprocess-reduc-broken}.
We use two instances of augmented ET trees: $\cE_{1}$ and $\cE_{2}$,
two instances of 2-dim ET trees: $\cE_{1}^{2d}$ and $\cE_{2}^{2d}$,
and two instances of the algorithm $\cA$ for \ref{thm:emergency decomp}:
$\cA_{1}$ and $\cA_{2}$. We will maintain three spanning forests
$F_{1}$ and $F_{2}$ and $F_{final}$. 

\begin{algorithm}[h]
\begin{enumerate}
\item Find a spanning tree $F_{b}$ of $G_{b}$. 
\item Set $F_{1},F_{2},F_{final}=F_{b}$. 
\item Preprocess $(G_{b},F_{b})$ using two instances of augmented ET trees
$\cE_{1}$ and $\cE_{2}$. 
\item Preprocess $(F_{b},F_{b})$ using two instances of 2-dim ET trees
$\cE_{1}^{2d}$ and $\cE_{2}^{2d}$. 
\item Give $(n,p,G_{b},\alpha_{b},D,F_{1},\cE_{1},\cE_{1}^{2d})$ to $\cA_{1}$
as input where $D=\emptyset$. 
\end{enumerate}
\caption{The algorithm for preprocessing the initial graph $G_{b}=(V,E)$\label{alg:preprocess-reduc-broken}}
\end{algorithm}

\begin{claim}
$\cB$ takes $O(n)$ time to preprocess $G_{b}$ and does not report
failure. This running time is independent from the length of the update
sequence.\label{claim:prep lv reduc}\end{claim}
\begin{proof}
Consider \ref{alg:preprocess-reduc-broken}. Step 1 takes $O(n)$
time. Step 2 is trivial. Step 3 and 4 take $O(n)$ time by \ref{thm:aug ET tree}
and \ref{thm:2dim ET tree}. Step 6 takes no times by \ref{rem:when no D }.
\end{proof}
Throughout the update sequence, only the updates of the spanning forest
$F_{final}$ will be returned as the answer of $\cB$. $F_{1}$ and
$F_{2}$ are maintained by $\cA_{1}$ and $\cA_{2}$ respectively,
and they are only for helping us maintaining $F_{final}$. $\cB$
will report failure whenever $\cA_{1}$ or $\cA_{2}$ report failure. 

Suppose that the update sequence of $\cB$ is of size at most $U_{b}$.
We divides the sequence into phases of size $U_{a}/2$. We will first
describe an algorithm how to maintain $F_{1}$ and $F_{2}$ such that
the following holds: in odd (even) phase $i$, $\cA_{1}$ ($\cA_{2}$)
is maintaining a spanning forest $F_{1}$ ($F_{2})$ of the graph
of during the whole phase. 

For each odd phase, suppose that the first update of the phase is
the $(i_{0}+1)$-th update. We assume that $\cA_{1}$ is maintaining
$F_{1}$ of the current graph $G_{i_{0}}$ and can handle more $U_{a}/2$
updates for the whole phase. We will describe how to prepare $\cA_{2}$
so that, at the end of the phase, i.e. after the $(i_{0}+\frac{U_{a}}{2})$-th
update, $\cA_{2}$ is maintaining $F_{2}$ which is a spanning forest
of the current graph $G_{i_{0}+\frac{U_{a}}{2}}$ at that time, and
can handle $\frac{U_{a}}{2}$ more updates for the whole next phase.
Let $(G_{\cE_{2}},F_{\cE_{2}})$ and $(G_{\cE_{2}^{2d}},F_{\cE_{2}^{2d}})$
be the underlying graphs and forests of $\cE_{2}$ and $\cE_{2}^{2d}$
respectively. 

For time period $[i_{0}+1,i_{0}+\frac{U_{a}}{4}]$, we distribute
evenly the time needed for (i) setting $(G_{\cE_{2}},F_{\cE_{2}})=(G_{i_{0}},F_{final,i_{0}})$
and $(G_{\cE_{2}^{2d}},F_{\cE_{2}^{2d}})=(F_{final,i_{0}},F_{final,i_{0}})$
where $G_{i_{0}}$ and $F_{final,i_{0}}$ are an input graph $G$
and the maintained forest $F_{final}$ after the $i_{0}$-th update,
respectively, and (ii) $\cA_{2}$ to preprocess $(n,p,G_{b},\alpha_{b},D_{i_{0}},F_{final,i_{0}},\cE_{2},\cE_{2}^{2d})$
where $D_{i_{0}}$ is the set of deleted edge into time $i_{0}$ (Note
that this is a valid input for $\cA_{2}$). 
\begin{claim}
During $[i_{0}+1,i_{0}+\frac{U_{a}}{4}]$, the procedure as described
above takes $O(((\frac{U_{b}{}^{1.5+\epsilon'}}{\alpha_{b}^{4+\epsilon'}}\gamma+\frac{U_{b}}{\alpha_{b}^{1+1/\epsilon'}})\log\frac{1}{p})/U_{a})$
time per update.\end{claim}
\begin{proof}
First of all, by \ref{thm:emergency decomp}, $\cA_{2}$ preprocesses
in time $P=O((\frac{U_{b}{}^{1.5+\epsilon'}}{\alpha_{b}^{4+\epsilon'}}\gamma+\frac{U_{b}}{\alpha_{b}^{1+1/\epsilon'}})\log\frac{1}{p})$
given a valid input. We just need to show that we can set $(G_{\cE_{2}},F_{\cE_{2}})=(G_{i_{0}},F_{final,i_{0}})$
and $(G_{\cE_{2}^{2d}},F_{\cE_{2}^{2d}})=(F_{final,i_{0}},F_{final,i_{0}})$
in at most $O(P)$ as well. 

One way to argue this is that, at time $i_{0}-U_{a}$, we know $(G_{\cE_{2}},F_{\cE_{2}})=(G_{i_{0}-U_{a}},F_{final,i_{0}-U_{a}})$
and $(G_{\cE_{2}^{2d}},F_{\cE_{2}^{2d}})=(F_{final,i_{0}-U_{a}},F_{final,i_{0}-U_{a}})$
by induction. Also, from time $i_{0}-U_{a}$ to $i_{0}$, we have
spent time on $\cE_{2}$ and $\cE_{2}^{2d}$ at most $O(P)$ time.
By remembering the execution log, we can revert both $\cE_{2}$ and
$\cE_{2}^{2d}$ to their states at time $i_{0}-U_{a}$ in time $O(P)$.
Now, as $G_{i_{0}-U_{a}}$ and $F_{final,i_{0}-U_{a}}$ differ from
$G_{i_{0}}$ and $F_{final,i_{0}}$ by at most $O(U_{a})$ edges,
we can update both $\cE_{2}$ and $\cE_{2}^{2d}$ in time $O(U_{a})\times\tilde{O}(\sqrt{\frac{U_{b}}{\alpha_{b}}})=o(P)$\footnote{Another way to argue this is by making $\cE_{2}$ and $\cE_{2}^{2d}$
\emph{fully persistent}, we can just ``go back in time'' to start
updating $\cE_{2}$ and $\cE_{2}^{2d}$ from time $i_{0}-U_{a}$.
The overhead factor for making it fully persistent is $O(\log\log n)$
by Straka \cite{Straka13} (cf. \cite{DriscollSST89,kaplan1995persistent}).}. As there are $U_{a}/4$ steps to distribute the work time, the process
takes time per update as desired.
\end{proof}
For time period $[i_{0}+\frac{U_{a}}{4}+1,i_{0}+\frac{U_{a}}{2}]$,
at each time $i_{0}+\frac{U_{a}}{4}+k$ for $1\le k\le\frac{U_{a}}{4}$,
we feed the $i_{0}+(2k-1)$-th update and the $(i_{0}+2k)$-th update
to $\cA_{2}$. This takes $O(\frac{U_{a}}{\alpha}\log\frac{1}{p}+\sqrt{\frac{U_{b}}{\alpha_{b}}})=O(n^{1/2-\epsilon}\log\frac{1}{p})$
time per update. Therefore, at the end of the phase, i.e. after the
$(i_{0}+\frac{U_{a}}{2})$-th update, $\cA_{2}$ is maintaining $F_{2}$
which is a spanning tree of the current graph $G_{i_{0}+\frac{U_{a}}{2}}$,
and can handle $\frac{U_{a}}{2}$ more updates as desired. In even
phases, we do symmetrically for $\cE_{1}$, $\cE_{1}^{2d}$ and $\cA_{1}$. 

In the above description, we need to maintain $F_{final}$ altogether
with $F_{1}$ and $F_{2}$ because at the beginning of each phase,
we set $(G_{\cE_{2}},F_{\cE_{2}})=(G_{i_{0}},F_{final,i_{0}})$ and
$(G_{\cE_{2}^{2d}},F_{\cE_{2}^{2d}})=(F_{final,i_{0}},F_{final,i_{0}})$.
To maintain $F_{final}$ given that we know that $F_{1}$ and $F_{2}$
are correct spanning forests in odd and even phase respectively, we
can use the exactly same idea as in \ref{rem:get stable} using additional
$O(\log^{2}n)$ update time. 
\begin{claim}
For each update, $\cB$ takes time at most $O(n^{1/2-\epsilon}\gamma\log\frac{1}{p})$
when $U_{b}=n^{1/2+\epsilon}$, $\alpha_{b}=1/n^{2\epsilon}$, $U_{a}=n^{1/2-\epsilon}\alpha_{b}^{1/\epsilon'}=n^{1/2-\epsilon-2\epsilon/\epsilon'}$.\label{claim:update lv reduc}\end{claim}
\begin{proof}
From the above discussion, we have that the update time of $\cB$
is 
\begin{eqnarray*}
 &  & O(n^{1/2-\epsilon}\log\frac{1}{p}+\frac{(\frac{U_{b}{}^{1.5+\epsilon'}}{\alpha_{b}^{4+\epsilon'}}\gamma+\frac{U_{b}}{\alpha_{b}^{1+1/\epsilon'}})\log\frac{1}{p}}{U_{a}}+\log^{2}n)\\
 & = & O(\gamma\log\frac{1}{p})\times O(n^{1/2-\epsilon}+\frac{n^{(1.5+\epsilon')(1/2+\epsilon)+2\epsilon(4+\epsilon')}+n^{1/2+\epsilon+2\epsilon(1+1/\epsilon')}}{n^{1/2-\epsilon-2\epsilon/\epsilon'}})\\
 & = & O(\gamma\log\frac{1}{p})\times O(n^{1/2-\epsilon}+\frac{n^{(1.5+\epsilon')(1/2+\epsilon)+2\epsilon(4+\epsilon')}}{n^{1/2-\epsilon-2\epsilon/\epsilon'}})\\
 & = & O(\gamma\log\frac{1}{p})\times O(n^{1/2-\epsilon}+\frac{n^{3/4+9.5\epsilon+3\epsilon\epsilon'+\epsilon'/2}}{n^{1/2-\epsilon-2\epsilon/\epsilon'}})\\
 & = & O(\gamma\log\frac{1}{p})\times O(n^{1/2-\epsilon}+n^{1/4+10.5\epsilon+3\epsilon\epsilon'+2\epsilon/\epsilon'+\epsilon'/2})\\
 & = & O(n^{1/2-\epsilon}\gamma\log\frac{1}{p})
\end{eqnarray*}
when $\epsilon=0.0069459$ and $\epsilon'=\sqrt{\frac{2\epsilon}{3\epsilon+0.5}}=0.16332$. \end{proof}
\begin{claim}
If $\phi(G_{b})\ge\alpha_{b}$, then, for each update, $\cB$ does
not report failure with probability $1-O(p)$.\end{claim}
\begin{proof}
If $\phi(G_{b})\ge\alpha_{b}$, then, for each update, both $\cA_{1}$
and $\cA_{2}$ do not fail with probability $1-O(p)$ by \ref{thm:emergency decomp}.
As $\cB$ only reports failure when $\cA_{1}$ or $\cA_{2}$ does,
we are done.
\end{proof}
The resulting algorithm $\cB$ from the above reduction is the desired
algorithm for \ref{thm:lv on expander final}.
\begin{lem}
[Restatement of \ref{thm:lv on expander final}]For any $n,$ $p$
and $T$, there is a decremental $\st$ algorithm $\cB$ for $3$-bounded
degree graphs with $n$ nodes that works correctly against adaptive
adversaries\emph{ }with certainty for the first $T$ updates. $\cB$
preprocesses an initial graph $G_{b}$ in time $t_{p}(n,p)=O(n)$.
Given each edge deletion, $\cB$ takes $t_{u}(n,p,T)$ worst-case
time to update or report failure, where $t_{u}(n,p,n^{1/2+\epsilon})=O(n^{1/2-\epsilon}\gamma\log\frac{1}{p})$.
If $\phi(G_{b})\ge1/n^{2\epsilon}$, then, for each update, $\cB$
does not report failure with probability $1-O(p)$.\end{lem}
\begin{proof}
[Proof of \ref{thm:lv on expander final}]By \ref{claim:prep lv reduc},
$\cB$ has preprocessing time $t_{p}(n,p)=O(n)$ independent of $T$.
Given each edge deletion, $\cB$ takes $t_{u}(n,p,T)$ worst-case
time to update or report failure where $t_{u}(n,p,n^{1/2+\epsilon})=O(n^{1/2-\epsilon}\gamma\log\frac{1}{p})$
by \ref{claim:update lv reduc}. Finally, by \ref{thm:emergency decomp},
if $\phi(G_{b})\ge\alpha_{b}=1/n^{2\epsilon}$, then, for each update,
$\cB$ does not report failure with probability $1-O(p)$.\end{proof}

%% file: formalize-correctness-det_sketch.tex
\section{Formalization of Dynamic Algorithms vs. Adversaries\label{sec:formalization}}

In this section, we formalize the notions of dynamic algorithms and
adversaries themselves. This can be of independent interest because,
to be best of our knowledge, these notions were not defined formally
in the literature in the context of dynamic algorithms yet.\footnote{The notion of adversaries has been defined formally on the context
of online algorithms whose focus is on competitiveness \cite{Ben-DavidBKTW94,Karp92},
but has not been formalized for dynamic algorithms whose focuses are
on update and query time.}

First, we define \emph{dynamic problems. }A dynamic problem $\cP$
is specified by a tuple $(\cG,U,Q,A,\cS,Corr)$ with the following
definitions.
\begin{itemize}
\item $\cG,U,Q$ and $A$ are the sets of valid \emph{underlying objects},
\emph{updates}, \emph{queries}, and \emph{answers}, respectively.
Each update $u\in U$ is a function $u:\cG\rightarrow\cG$. We call
$U\cup Q$ a set of valid \emph{operations}. 
\item Let $\cS=\bigcup_{i\geq0}\cS_{i}$ where, for any $i$, $\cS_{i}\subseteq\cG\times(U\cup Q)^{i}$
is a set of valid \emph{operation sequences }of length $i$. For any
$S_{i}=(G_{0},o_{1},\dots,o_{i})\in\cS_{i}$, $G_{0}$ is called an
\emph{initial }object. For any $1\le j\le i$, we say $G_{j}$ is
the updated underlying object after the $j$-th operation. More precisely,
$G_{j}$ is obtains by iteratively applying to $G_{0}$ each update
operation $o_{k}$ for $k=1,\dots,j$ where $o_{k}\in U$.
\item Given an operation sequence $S_{i}=(G_{0},o_{1},\dots,o_{i})\in\cS_{i}$
and an answer sequence $\vec{a}_{i}=(a_{0},a_{1},\dots,a_{i})\in A^{i+1}$,
we say $\pi_{i}=(G_{0},a_{0},o_{1},a_{1},\dots,o_{i},a_{i})$ is a
\emph{transcript} at round $i$. We sometimes write $\pi_{i}=(S_{i},\vec{a}_{i})$.
Let $\Pi_{i}$ be the set of all transcripts at round $i$, and $\Pi=\bigcup_{i\geq0}\Pi_{i}$. 
\item For any $(\pi_{i},o_{i+1})\in\Pi_{i}\times(U\cup Q)$, $Corr(\pi_{i},o_{i+1})$
is a set of \emph{correct answers} for $o_{i+1}$ given $\pi_{i}$.
If $o_{i+1}\in Q$ is a query, then $Corr(\pi_{i},o_{i+1})\subseteq A$.
If $o_{i+1}\in U$ is an update, then for convenience we define $Corr(\pi_{i},o_{i+1})=\{\emptyset\}$. 
\end{itemize}
To illustrate, we give some examples how the tuple captures the definitions
of many dynamic problems in the literature.
\begin{example}
[Dynamic Connectivity]\label{exa:dyn conn} $\cG$ is the set of
all undirected graphs. Each update $u\in U$ specifies which edge
in the graph to be inserted or deleted. Each query $q\in Q$ specifies
a pair of nodes. The set of answers $A=\{\mbox{YES},\mbox{NO}\}$.
Given a transcript $\pi_{i}\in\Pi_{i}$ and a query $(u,v)\in Q$,
$Corr(\pi_{i},(u,v))=\{\mbox{YES}\}$ if $u$ and $v$ is connected
in $G_{i}$ and $Corr(\pi_{i},(u,v))=\{\mbox{NO}\}$ otherwise. For
any valid sequence $S\in\cS$, $S$ always excludes trivially invalid
updates and queries (e.g. a deletion of an edge $e$ when $e$ is
not in the graph, or a query $(u,v)$ when a node $u$ is not in the
graph).
\end{example}

\begin{example}
[Dynamic Spanning Tree]Similarly, we have $\cG$ is the set of all
undirected graphs. Each update $u\in U$ specifies which edge in the
graph to be inserted or deleted. The set of query is trivial: $Q=\{\textsf{list}\}$.
For any sequence $S_{i}=(G_{0},o_{1},\dots,o_{i})\in\cS_{i}$, we
have $o_{j}\in U$ for odd $j$ and $o_{j}\in Q$ for even $j$, so
that this captures the fact that we immediately query after each update
in this problem%
. The set of answers $A$ is the list of edges to be added or removed
from some set. So given any sequence $a_{0},\dots,a_{i}\in A^{i+1}$,
this defines a set $F_{i}$ of edges (by first applying $a_{0}$ to
an empty set, then applying $a_{1}$, and so on). Let $\pi_{i}=(G_{0},a_{0},\dots,o_{i},a_{i})\in\Pi_{i}$,
we have $a_{i+1}\in Corr(\pi_{i},\textsf{list})$ iff $F_{i+1}$ is
a spanning forest of $G_{i+1}$ where $F_{i+1}$ is obtained by applying
$a_{i+1}$ to $F_{i}$.
\end{example}

\begin{example}
For the dynamic problems on planar graphs, $\cG$ is the set of all
planar graphs, and $\cS$ is defined in such a way that after each
update the underlying graph remains planar. We can similar define
this for the problem on bounded arboricity graphs.
\end{example}
Next, we can define the notion of adversaries and algorithms for dynamic
problems. Fix some dynamic problem $\cP$. An \emph{adversary }$f$
is a function $f:\Pi\rightarrow(U\cup Q)$ (except $f(\cdot)=G_{0}\in\cG$)
such that for any $\pi_{i}=(S_{i},\vec{a}_{i})\in\Pi_{i}$ we have
$(S_{i},f(\pi_{i}))\in\cS_{i+1}$. For all $(S,\vec{a}),(S,\vec{a}')\in\Pi$
where $\vec{a}\neq\vec{a}'$, if we have $f(S,\vec{a})=f(S,\vec{a}')$,
then we say that $f$ is \emph{oblivious}, otherwise $f$ is \emph{adaptive}
in general. 

If an adversaries $f$ is oblivious, then $f$ is a function of this
form $f:\cS\rightarrow(U\cup Q)$. For any oblivious $f$, $f$ can
be written as a function of the form $f:\cS\rightarrow(U\cup Q)$
and the \emph{operation sequence generated by $f$} is $S_{f}=(G_{0},o_{1},o_{2},\dots)$
where $G_{0}=f(\cdot)$ and $o_{i}=f(G_{0},\dots,o_{i-1})$ for all
$i\ge1$. 

An \emph{algorithm }$\cA$ is a function $\cA:\{0,1\}^{\infty}\times(\cG\cup(\Pi\times(U\cup Q)))\rightarrow A\cup\{\emptyset\}$
such that for any infinite string $R\in\{0,1\}^{\infty}$ and transcript
$\pi\in\Pi$, we have $\cA(R,\pi,o)=\emptyset$ if $o\in U$, and
$\cA(R,\pi,o)\in A$ if $o\in Q$.

Given any adversary $f$, algorithm $\cA$ and an infinite string
$R$, this determines a sequence $(G_{0},a_{0},o_{1},a_{1},o_{2},a_{2},\dots)$
where $G_{0}=f(\cdot)$, $a_{0}=\cA(R,G_{0})$, $o_{i}=f(G_{0},a_{0},\dots,o_{i-1},a_{i-1})$,
and $a_{i}=\cA(R,G_{0},a_{0},\dots,o_{i-1},a_{i-1},o_{i})$ for $i\ge1$.\footnote{We note the resemblance of this definitions and the transcript in
the context of communication complexity is intentional.} For $i\ge0$, let $\pi_{i}=(G_{0},a_{0},\dots,o_{i},a_{i})\in\Pi_{i}$
denote a transcript at round $i$ given $f,\cA$ and $R$. We can
write $o_{i}=f(\pi_{i-1})$ and $a_{i}=\cA(R,\pi_{i-1},o_{i})$. Let
$t_{i}$ be the time period $\cA$ uses for outputting $a_{i}$ starting
counting from the time $a_{i-1}$ is outputted. We emphasize that
all variables in this paragraph are determined once $f,\cA$ and $R$
are fixed. Let $succ_{i}(f,\cA,R)$ be the successful event that occurs
when $a_{i}\in Corr(\pi_{i-1},o_{i})$ (or $a_{0}\in Corr(G_{0})$
when $i=0$).

Now, we are ready to define correctness of an algorithm.
\begin{defn}
[Correctness]Let $\cF$ be some set of adversaries. An algorithm
$\cA$ works correctly against $\cF$ with probability $p$ for the
first $t$ steps, if, for every $f\in\cF$, $\Pr_{R}[succ_{i}(f,\cA,R)]\ge p$
for all $0\le i\le t$. If $t=\infty$, then the phrase ``for the
first $t$ steps'' is omitted.
\end{defn}
Next, we define the precise meaning of preprocessing, update and query
time of an algorithm.
\begin{defn}
[Running time]Let $\cF$ be some set of adversaries. An algorithm
$\cA$ has worst-case preprocessing time $T_{p}$, worst-case update
time $T_{u}$ and worst-case query time $T_{q}$ against $\cF$ with
probability $p$ if, for every for every $f\in\cF$, we have $\Pr_{R}[t_{0}\le T_{p}]\ge p$
and, for each $i\ge1$, $\Pr_{R}[t_{i}\le T_{u}]\ge p$ when $o_{i}\in U$
and $\Pr_{R}[t_{i}\le T_{q}]\ge p$ when $o_{i}\in Q$.
\end{defn}
If $\cF$ is a set of all adversaries, then we say that $\cA$ works
\emph{against adaptive adversaries}. If $\cF$ is a set of all oblivious
adversaries, then we say that $\cA$ works \emph{against oblivious
adversaries}.

From above definition, adversaries are deterministic. We can also
define a randomized version, although as will be shown below, this
does not give more power. A \emph{randomized adversary} $f$ is a
function $f:\{0,1\}^{\infty}\times\Pi\rightarrow(U\cup Q)$ where
for any string $R'$, $f_{R'}=f(R',\cdot)$ is an adversary. We say
that $f$ is an \emph{oblivious randomized adversary} if, for all
$R'$, $f_{R'}$ is an oblivious adversary.
\begin{defn}
[Correctness (randomized adversaries)]Let $\cF$ be some set of randomized
adversaries. An algorithm $\cA$ works correctly against $\cF$ with
probability $p$ for the first $t$ steps, if, for every $f\in\cF$,
$\Pr_{R,R'}[succ_{i}(f_{R'},\cA,R)]\ge p$ for all $0\le i\le t$.
If $t=\infty$, then the phrase ``for the first $t$ steps'' is
omitted.
\end{defn}
The following proposition shows that, to design an algorithm that
works against adaptive adversaries or oblivious adversaries, it is
safe to assume that the adversaries are deterministic.
\begin{prop}
\label{thm:det adv}Let $\cF$ and $\cF_{r}$ and be a set of all
(oblivious) deterministic adversaries and (oblivious) randomized adversaries
respectively. If $\cA$ works correctly against $\cF$ with probability
$p$ for the first $t$ steps, then so does against $\cF_{r}$. \end{prop}
\begin{proof}
Let $f\in\cF_{r}$. We have that for all $0\le i\le t$, $\Pr_{R,R'}[succ_{i}(f_{R'},\cA,R)]\ge\min_{R'}\Pr_{R}[succ_{i}(f_{R'},\cA,R)]\ge p$
where the last inequality is because, for any $R'$, we have $f_{R'}\in\cF$.
\end{proof}
To say instead that $\cA$ has \emph{amortized} or \emph{expected
worst-case} or \emph{expected amortized} update time $T_{u}$, the
expression $t_{i}\le T_{u}$ is replaced by $\sum_{j=1}^{i}t_{i}\le i\times T_{u}$
or $E_{R}[t_{i}]\le T_{u}$ or $E_{R}[\sum_{j=1}^{i}t_{i}]\le i\times T_{u}$,
respectively. Similarly, for query time and preprocessing time. A
Monte Carlo algorithm works correctly with some probability and has
guarantee on running time without expectation. A Las Vegas algorithm
work correctly with certainty but its running time is guaranteed only
in expectation. A deterministic algorithm both works correctly with
certainty and has guarantee on running time without expectation.

In many dynamic problems including dynamic $\st$ problem, each update
is immediately followed by a query. There are many other examples
including dynamic $\mst$, and most, if not all, dynamic problems
which maintains an approximate/exact value (e.g. the size of maximum
matching, the size of a global min cut, etc.) For these dynamic problems,
if the update and query time from the definition above are $T_{u}$
and $T_{q}$, then the update time from the original definition is
$T_{u}+T_{q}$.

%% file: omit.tex
\section{Omitted proof\label{sec:Omitted-proof}}

\subsection{\label{sec:proof aug ET tree}Proof of \ref{thm:aug ET tree}}
\begin{proof}
By using \emph{ET tree }data structure from \cite[Section 2.1.4 - 2.1.5]{HenzingerK99},
we can preprocess $G$ and $F$ in time $O(m)$ and handle all operations
(except $\textsf{path}$) in time $O(\log n)$, and handle $\textsf{tree\_size}$
and $\textsf{count\_tree}$ in $O(1)$ time.

In order to also handle $\textsf{path}$, it suffices to show that
there is a data structure that preprocess $G$ and $F$ in time $O(m)$
and handle $\instree$, $\deltree$ and $\textsf{path}$ in time $O(\log n)$.
This is because we can just implement this data structure in parallel
with ET tree. A data structure called \emph{link-cut tree }from \cite{SleatorT83}
or\emph{ top tree }from \cite{AlstrupHLT05} can indeed do the job.
From \cite[Section 2]{AlstrupHLT05}, top tree can handle $\instree$,
$\deltree$. In order to handle $\textsf{path}(u,v,i)$, we can the
operation $\textsf{expose}(u,v)$ (defined in \cite{AlstrupHLT05})
which allows us to do binary search on the unique path $P_{uv}$ from
$u$ to $v$ in $F$ and find the $i$-th vertex in $P_{uv}$.
\end{proof}

\subsection{From \ref{sec:2d ET tree} and \ref{sec:linear sketches}}

\subsubsection{\label{sec:proof list intersection}Proof of \ref{thm:list intersection}}

To describe the algorithm for \ref{thm:list intersection}, we need
to define several notions: \emph{base clusters}, \emph{non-base clusters},
\emph{clusters}, and \emph{hierarchies of clusters}. 

For any list $L=(e_{1},e_{2},\dots,e_{s})\in\cL$, we call a collection
of sublists $B_{1},\dots,B_{g}$ of $L$ as \emph{base clusters }of
$L$ if the following hold: 1) $g=O(\sqrt{k})$, 2) $L=B_{1}\dots B_{g}$,
and 3) $\min\{\sqrt{k}/2,|L|\}\le|B_{i}|\le2\sqrt{k}$. That is, the
base clusters partition $L$ into $O(|L|/\sqrt{k})$ parts where each
part of length around $\sqrt{k}$ (or as long as $|L|$). Let $\cB(L)$
denote the set of base clusters of $L$. Let $\cB(\cL)=\bigcup_{L\in\cL}\cB(L)$
denote the set of base clusters of $\cL$. 

Let $\cB(L)$ be the base clusters of $L$. Let $\cT_{L}$ be a balanced
binary search tree whose leaves correspond to the base clusters $\cB(L)$.
Each node $u\in\cT_{L}$ corresponds to a sublist $C_{u}$ of $L$.
If $u$ is the $i$-th leaf, then $C_{u}=B_{i}$ is the $i$-th base
cluster of $L$. If $u$ is a non-leaf and have left child $v$ and
right child $w$, then $C_{u}=C_{v}C_{w}$ is the concatenation of
the two sublists corresponding the two nodes $v$ and $w$. For each
non-leaf node $u\in\cT_{L}$, we call $C_{u}$ a \emph{non-base cluster
}of $L$. Note that, for the root $r\in\cT_{L}$, $C_{r}=L$. We call
$C_{r}$ the \emph{root cluster }of $L$. We denote $\cN(L)$ the
set of all non-base clusters of $L$, and $\cN(\cL)=\bigcup_{L\in\cL}\cN(L)$
denote the set of non-base clusters of $\cL$. Similarly, we call
$\cT_{L}$ the \emph{hierarchy of clusters }in $L$, and $\cT_{\cL}=\bigcup_{L\in\cL}\cT_{L}$
be the set of such hierarchies. Let $\cC(L)=\cB(L)\cup\cN(L)$ denote
the set of \emph{clusters} of $L$, and similarly $\cC(\cL)=\bigcup_{L\in\cL}\cC(L)$
denote the set of clusters of $\cL$. 

Next, we describe the main data structure that for answering the $\textsf{intersection}$
query. We say that $A$ is an \emph{intersection table} \emph{of clusters
in }$\cL$ if for any two clusters $C_{1},C_{2}\in\cC(\cL)$, $A(C_{1},C_{2})=e$
if a number $e\in C_{1}\cap C_{2}$ exists, otherwise, $A(C_{1},C_{2})=\emptyset$. 
\begin{prop}
Let $\cL$ be a collection of lists. Given an intersection table of
clusters in $\cL$, for any $L_{1},L_{2}\in\cL$, a query $\textsf{intersection}(L_{1},L_{2})$
can be answer in $O(1)$ time.\label{thm:answer intersection}\end{prop}
\begin{proof}
Let $C_{1}$ and $C_{2}$ be the root clusters of $L_{1}$ and $L_{2}$.
Recall that $C_{1}=L_{1}$ and $C_{2}=L_{2}$. We can just return
$A(C_{1},C_{2})$ to the query $\textsf{intersection}(L_{1},L_{2})$.
\end{proof}
Now, the goal is to prove the following. We will exploit the assumptions
from \ref{thm:list intersection} so that we can preprocess and update
quickly.
\begin{prop}
Let $\cL=\{L_{i}\}_{i}$ be a collection of lists such that $\sum_{i}|L_{i}|\le k$
and the number of lists in $\cL$ is bounded by some constant. There
are at most $O(\sqrt{k})$ clusters of $\cL$, i.e. $|\cC(\cL)|=O(\sqrt{k})$.\label{thm:bound clusters}\end{prop}
\begin{proof}
As there are only a constant number of lists in $\cL$, we only any
to bound $|\cC(L)|$ for each $L\in\cL$. By definition of base clusters,
it is clear that $|\cB(L)|=O(\sqrt{k})$ because $|L|\le k$, and
so $|\cN(L)|=O(\sqrt{k})$. Hence $|\cC(L)|\le|\cB(L)|+|\cN(L)|=O(\sqrt{k})$.
\end{proof}
We note that $\cN(\cL)$ and $\cC(\cL)$ and is defined based on the
base clusters $\cB(\cL)$ and the hierarchies of clusters $\cT_{\cL}$.
Therefore, in our data structure, we will only explicitly maintain
$\cB(\cL)$ and $\cT_{\cL}$ and the intersection table $A$ of clusters
in $\cL$. We show how to build and maintain quickly the first two
in \ref{thm:maintain base and hierarchy}, and then do the same for
the intersection table in \ref{thm:maintain intersection}.
\begin{lem}
Let $\cL$ be a collection of lists such that the total length of
lists in $\cL$ is bounded by $\sum_{L\in\cL}|L|\le k$. Then the
base clusters $\cB(\cL)$ and the hierarchies of cluster $\cT_{\cL}$
of $\cL$ can be built in $O(k)$ time and can be maintained under
$\textsf{create}$, $\textsf{merge}$, and $\textsf{split}$ in $O(\log k)$
time.\label{thm:maintain base and hierarchy}\end{lem}
\begin{proof}
To build the base clusters $\cB(\cL)$, we just read through $\cL$
and split it into sublists of size $\sqrt{k}$. After having the base
clusters $\cB(\cL)$, for each $L\in\cL$, we just build $\cT_{L}$
as some balanced binary search tree whose leaves correspond $\cB(L)$.
The total time is obviously $O(k)$. Next, we show how to handle updates.

It is trivial to handle $\textsf{create}$. Given either $\textsf{merge}$
or $\textsf{split}$ to $\cL$, to maintain $\cB(\cL)$, observe that
we only need to merge and split the base clusters by constant number
of time. We will store each base cluster of $\cL$ in a balanced binary
search tree (binary search tree with logarithmic depth). As we can
merge and split balanced binary search tree with at most $k$ elements
in time $O(\log k)$, we are done.

To maintain $\cT_{\cL}$, for each $L$, $\cT_{L}$ is simply some
balanced binary search tree whose leaves are base clusters of $L$.
If some base cluster of $L$ is splitted or merged, we maintain $\cT_{L}$
as a balanced binary search tree in time $O(\log k)$. Because either
$\textsf{merge}$, or $\textsf{split}$ to $\cL$ only generate a
constant number of merging and splitting to base clusters, we are
done.\end{proof}
\begin{lem}
Let $\cL=\{L_{i}\}_{i}$ be a collection of lists. Suppose that 1)
the total length of lists in $\cL$ is bounded by $\sum_{i}|L_{i}|\le k$,
2) the number of lists in $\cL$ is bounded by some constant, and
3) each number occurs in $\cL$ at most a constant number of times.
The intersection table of clusters in $\cL$ can be built in time
$O(k)$ and can be maintained under $\textsf{create}$, $\textsf{merge}$,
and $\textsf{split}$ in $O(\sqrt{k}\log k)$ time.\label{thm:maintain intersection}\end{lem}
\begin{proof}
By \ref{thm:bound clusters}, we know that there are only $O(\sqrt{k})$
clusters. By the third assumption, we can maintain the following data
structure: given any number $e$, all the base clusters $B\in\cB(\cL)$,
such that $e$ occurs in $B$, are returned in $O(1)$ time. We first
describe how to handle the update $A$. We will describe how to build
$A$ later.

It is trivial to handle $\textsf{create}$. We only show how to handle
$\textsf{merge}$, and $\textsf{split}$. Given the operations $\textsf{merge}$
or $\textsf{split}$ to $\cL$, we only only need to merge and split
the base clusters $\cB(\cL)$ by constant number of times. It suffices
to show that intersection table of clusters in $\cL$ can be maintained
in $O(\sqrt{k}\log k)$ time under under each merge or split of base
clusters in $\cB(\cL)$.

Consider any list $L\in\cL$. Let $\cB(L)$ be the base clusters of
$L$ and $\cT_{L}$ be the hierarchy of clusters of $L$. Suppose
that two base clusters $B_{1},B_{2}\in\cB(L)$ are merged into $B'=B_{1}B_{2}$.
Let $\cB'(L)$ and $\cT'_{L}$ be the updated base clusters and hierarchy
of clusters of $L$ after merging. By \ref{thm:maintain base and hierarchy},
we can obtain $\cB'(L)$ and $\cT'_{L}$ in time $O(\log k)$. Let
$P$ the path in $\cT'_{L}$ from root of to the leaf corresponding
to the new base cluster $B'$. Observe that, any cluster of $C_{u}$
corresponding to a node $u\in\cT'_{L}\setminus P$ outside the path
$P$ is not affected by the merging. Therefore, to obtain an updated
intersection table $A$ of clusters, it is enough to compute $A(C_{u},C)$
for all $u\in P$ and $v\in\cT'_{L}$. We will do this bottom-up.

First, we compute $A(B',C_{v})$ for all $v\in\cT'_{L}$ as follows.
For each number $e\in B'$, for each base clusters $B''$ that $B''\ni e$,
and for each clusters $C$ that $C$ is an ancestor of $B''$ in $\cT'_{L}$,
we set $A(B',C)=e$. Other entries are set to $\bot$ by default.
This can be done in $O(\sqrt{k}\log k)$ times. Because there are
$O(\sqrt{k})$ numbers in $B''$, for each number $e$ we can list
all base clusters $B''$ containing $e$ in $O(1)$, and each such
base cluster $B''$ has at most $O(\log k)$ ancestor because $\cT'_{L}$
has $O(\log k)$ depth.

Next, we compute $A(C_{u},C_{v})$ for all $v\in\cT'_{L}$ where $u\in P$
is a non-leaf in $\cT'_{L}$ and has $u_{1}$ and $u_{2}$ has children.
For each $v\in\cT'_{L}$, we set $A(C_{u},C_{v})=A(C_{u_{1}},C_{v})$
if $A(C_{u_{1}},C_{v})\neq\bot$, otherwise $A(C_{u},C_{v})=A(C_{u_{2}},C_{v})$.
This can be done in $O(\sqrt{k})$ because there are only $O(\sqrt{k})$
clusters. In total, the time need to compute $A(C_{u},C_{v})$ for
all $v\in\cT'_{L}$ and for all non-leaves $u\in P$ is $O(\sqrt{k}\log k)$
because $\cT'_{L}$ has $O(\log k)$ depth.

The correctness of the intersection table is clear. To build $A$,
use the same idea as the way we update, it is obvious how to do it
in $O(k\log k)$ but we can save the $O(\log k)$ factor by building
it ``bottom-up''.
\end{proof}
By \ref{thm:answer intersection} and \ref{thm:maintain intersection},
this concludes the proof of \ref{thm:list intersection}.

\subsubsection{\label{sec:proof sparse recovery}Proof of \ref{thm:sparse recovery}}

To prove the theorem, we need to define unbalanced expanders. 
\begin{defn}
\label{def:unbalanced expander}Fix any $k,d,\epsilon$. A \emph{$(k,\epsilon)$-unbalanced
expander} is a bipartite simple graph $G=((L,R),E)$ with left degree
$d$ (i.e. for every $v\in L$, $\deg(v)=d$) such that for any $X\subset L$
with $|X|\le k$, the set of neighbors $N(X)$ of $X$ has size $|N(X)|\ge(1-\epsilon)d|X|$.
\end{defn}
There is an explicit construction of a $(k,\epsilon)$-unbalanced
expander by \cite{CapalboRVW02}.
\begin{thm}
[Theorem 7.3 of \cite{CapalboRVW02} and Proposition 7 of \cite{BerindeGIKS08}]Fix
any $n\ge k$ and $\epsilon>0$. There is a $(k,\epsilon)$-unbalanced
expander $G=((L,R),E)$ with left degree $\beta_{1}=2^{O(\log(\log n/\epsilon))^{3}}$
, $|L|=n$ and $|R|=k\beta_{1}/\epsilon^{O(1)}$. Given $u\in L$
and $i\in[\beta]$, we can compute the $i$-th edge incident to $u$
in $\beta_{2}=O(\poly(\log(n),1/\epsilon))$ time.\label{thm:unbalanced expander}
\end{thm}
Next, we need some more definitions about matrices. The \emph{Hamming
code matrix} (or \emph{bit test matrix}) $H^{(n)}\in\{0,1\}^{\left\lceil \log_{2}n+1\right\rceil \times n}$
is a matrix whose the $i$-th column has $1$ followed by the binary
representation of $i$. Let $P$ be a $p\times n$ matrix and $Q$
be a $q\times n$ matrix with rows $\{P_{i}\}_{0\le i<p}$ and $\{Q_{j}\}_{0\le j<q}$
respectively. The \emph{row tensor product} $P\rtensor Q$ of $P$
and $Q$ is a $pq\times n$ matrix whose rows are $\{P_{i}Q_{j}\}_{0\le i<p,0\le j<q}$
where $P_{i}Q_{j}$ denotes the component-wise product of the two
row vectors. The construction of $k$-sparse recovery linear sketch
in \cite{BerindeGIKS08} is as follows.
\begin{thm}
[Theorem 18 of \cite{BerindeGIKS08}]Fix and $n\ge k$, $r$ and $\epsilon=1/8$.
Let $\Psi\in\{0,1\}^{r\times n}$ be a bi-adjacency matrix of a $(k,\epsilon)$-unbalanced
expander. Let $\Phi=\Psi\otimes_{R}H^{(n)}$. Then there is an algorithm
$\cA_{recover}$ such that, for any $k$-sparse vector $v$ $v\in[-M,M]^{n}$
where $M=poly(n)$, given $\Phi v$, all non-zero entries of $v$
can be identified in time $O(r\log^{2}n)$. \label{thm:sparse recovery matrix}
\end{thm}
Now, we are ready to prove \ref{thm:sparse recovery}.
\begin{proof}
[Proof of \ref{thm:sparse recovery}]By \ref{thm:unbalanced expander}
when setting $\epsilon=1/8$, there is a bi-adjacency matrix $\Psi\in\{0,1\}^{O(k\beta_{1})\times n}$
of a $(k,\epsilon)$-unbalanced expander $G=((L,R),E)$ with left
degree $\beta_{1}=2^{O(\log\log n)^{3}}$. Then, by \ref{thm:sparse recovery matrix},
we have that $\Phi=\Psi\otimes_{R}H^{(n)}\in\{0,1\}^{d\times n}$
where $d=O(k\beta_{1}\log n)$ and there is algorithm $\cA_{recover}$
for recovering $k$-sparse vector $v$, given $\Phi v$, in time $O(k\beta_{1}\log^{2}n)=O(d\log n)$.
Since each column of $\Psi$ is $\beta$-sparse, each column of $\Phi$
is $s$-sparse where $s=O(\beta_{1}\log n)$. Finally, we show the
algorithm $\cA_{construct}$ for listing a non-zero entries in a column
of $\Phi$. Since, for each node $u\in L$, edges incident to $u$
correspond to the non-zero entries of the $u$-th column of $\Psi$,
and each edge can be computed in $\beta_{2}=O(\poly\log(n)))$ time,
we have that all non-zero entries of each column of $\Phi$ can also
be identified in $s\beta_{2}$ time.
\end{proof}

\subsubsection{\label{sec:proof d word ET}Proof of \ref{thm:d-word ET tree}}
\begin{proof}
First, we need this definition. An \emph{Euler tour} $ET(T)$ of $T$
is an (ordered) list of size $2|V(T)|-1$ with the following properties.
Each element in the list is mapped to some node of $T$ such that
1) each sublist of $ET(T)$ corresponds to some connected subtree
of $T$, and 2) the whole list corresponds to the whole $T$. (See
\cite{HenzingerK99} for the precise definition of Euler tour.) For
each node $v\in T$, there is some element from $ET(T)$ mapped to
and we call the first element in $ET(T)$ mapped to $v$ the \emph{representative
}of $v$ in $ET(T)$. Given a tree $T$ in $F$, an ET tree $\cT_{T}$
is a balanced binary search tree whose leaves are exactly the elements
of $ET(T)$. So each internal node $u$ of $\cT_{T}$, called \emph{ET
node}, corresponds to all leaves in the subtree rooted at $u$ in
$\cT_{T}$, which form a sublist of $ET(T)$, which in turn corresponds
to a subtree of $T$. For each $v\in T$, the leaf of $\cT_{T}$ corresponding
to the representative of $v$ in $ET(T)$ is called a representative
of $v$ in $\cT_{T}$.

We augment a standard ET-tree as follows. Each ET node $u$ of $\cT_{T}$
has an additional corresponding vector $e_{u}$ of dimension $d$.
Entries of $e_{u}$ are either marked ``active'' or ``inactive''.
We have the following invariant about $e_{u}$ for each ET-node $u\in\cT_{T}$:
(1) if $u$ is a leaf and $u$ is not the representative of any vertex
of $T$, then no entry of $e_{u}$ is marked as active, (2) if $u$
is a leaf and $u$ is the representative of a vertex $v$ of $F$,
then, for $i\in[d]$, $e_{u}[i]$ is active if $x_{v}[i]\neq0$ initially
or $x_{v}[i]$ has been updated after preprocessing, and we have $e_{u}[i]=x_{v}[i]$
for each active entry, and (3) if $u$ is an internal node with $v$
and $w$ as children in $\cT_{T}$, then, for $i\in[d]$, $e_{u}[i]$
is active if either $e_{v}[i]$ or $e_{w}[i]$ is active, and the
value of active entry $e_{v}[i]$ is the sum of $e_{v}[i]$ and $e_{w}[i]$
that are active. 

It is easy to see that, given these invariants, the root $r$ of $\cT_{T}$
is such that if $e_{r}[i]$ is active, then $\sum_{u\in T}x_{u}[i]=e_{r}[i]$
and if $e_{r}[i]$ is inactive then $\sum_{u\in T}x_{u}[i]=0$. Hence,
to return return $\textsf{sum\_vec}$, $\sum_{u\in T}x_{u}$ can be
read off from $e_{r}$ in time $O(d)$. Since ET-trees are balanced,
the time for, $\textsf{update\_vec}$, updating each some entry $x_{u}[i]$
is $O(\log n)$. 

We analyze the time needed for preprocessing and other update operations.
The preprocessing algorithm is to, for each $T$ in $F$, mark as
active and set the value of the entries of each ET node in $\cT_{T}$
correctly according to the invariants. There are $nnz(x)$ entry to
update, hence the preprocessing time takes $O(nnz(x)\log n)$. To
update tree edges, by using the same algorithms as in typical ET trees,
each tree-edge update operation requires us to recompute the vector
$e_{u}$ of $O(\log n)$ many ET nodes $u$. So the time needed for
$\instree$ and $\deltree$ is $O(d\log n)$ time. 
\end{proof}

\subsection{From \ref{sec:reductions}}

\subsubsection{Proofs of \ref{lem:reduc stable}, \ref{lem:reduc poly length},
\ref{lem:reduc bound degree} and \ref{lem:reduc lower bound edge}.}

For completeness, we give the proofs of \ref{lem:reduc stable}, \ref{lem:reduc poly length},
\ref{lem:reduc bound degree} and \ref{lem:reduc lower bound edge}.
In \cite[Section 3.1]{HenzingerK99}, Henzinger and King show a reduction
from dynamic $k$-weight $\mst$ to dynamic $\st$ problem. We modify
the idea in their reduction and use it in all the proofs of this section.
The basic idea is illustrated in the proof of \ref{lem:reduc stable}
and \ref{rem:get stable}.
\begin{lem}
[Reminder of \ref{lem:reduc stable}]Suppose there is an $(n,p_{c},p_{t},T)$-algorithm
$\cA$ for any graph with preprocessing time $t_{p}$ and update time
$t_{u}$. Then there is a stable $(n,O(Tp_{c}),p_{t},T)$-algorithm
$\cB$ for the same graph with preprocessing time $O(t_{p}+m_{0})$
where $m_{0}$ is the number of initial edges, and update time $O(t_{u}\log n+\log^{2}n)$.\end{lem}
\begin{proof}
[Proof of \ref{lem:reduc stable}]The algorithm $\cB$ preprocesses
as follows. Given an initial graph $G_{0}$, preprocess $G_{0}$ using
$\cA$. Then $\cA$ outputs a spanning forest $F'$. Then we find
another spanning forest $F$ of $G_{0}$ in time $O(m_{0})$ where
$m_{0}=|E(G_{0})|$. Then we initialize two augmented ET trees $\cE$
and $\cE'$. Both the input graph and forest of $\cE'$ are $F'$.
Similarly, both the input graph and forest of $\cE$ are $F$. The
total preprocessing time is $O(t_{p}+m_{0})$.

To an update, we first describe how $F'$ handle updates: We feed
each update to $\cA$ so that it updates its spanning forest $F'$
of $G$ in time $t_{u}$. As there are at most $t_{u}$ changes in
the tree edges of $F'$, we update $\cE'$ accordingly in time $O(t_{u}\log n)$. 

To maintain $F$, if a tree-edge of $F$ is not deleted, then $F$
does not change. Suppose a tree-edge $(u,v)$ is deleted and some
connected component $T$ of $F$ is separated into $T_{u}\ni u$ and
$T_{v}\ni v$. Then we query $\cE'$ on $F'$ using $\textsf{findtree}$
operation, whether $u$ and $v$ are still connected in $F'$. If
not, there there is no replacement edge to reconnect $T_{u}$ and
$T_{v}$ and we are done. Otherwise, there is a path $P_{uv}=(u,\dots,v)$
in $F'$ and there must be some edge $(u',v')$ where $u'\in T_{u}$
and $v'\in T_{v}$. Using $\textsf{path}$ operation on $\cE'$ and
$\textsf{findtree}$ operation on $\cE$, we can binary search in
$P_{uv}$ to find $(u',v')$ and use it as a replacement edge in $O(\log n\times\log n)=O(\log^{2}n)$
time. (To be more precise, we query $\cE'$ for the $k$-th edge $(u'',v'')$
in $P_{uv}$ for some $k$, and check if $u''\in T_{u}$ and $v''\in T_{v}$
using $\cE$). Therefore, the update time for $F$ itself is $O(\log^{2}n)$.
Hence, the total update time of $\cB$ is $O(t_{u}\log n+\log^{2}n)$.
By union bound, $\cB$ is $(n,p'_{c},p_{t}',U)$-algorithm where $p'_{c}=O(Tp_{c})$
and $p'_{t}=p_{t}$. 

Finally, observe that $\cB$ is stable because $F$ is not updated
except a tree-edge $(u,v)$ of $F$ is deleted. In that case, $\cB$
either update a spanning forest into $F\setminus(u,v)$ or $F\cup(u',v')\setminus(u,v)$. \end{proof}
\begin{rem}
\label{rem:get stable}In the proof of \ref{lem:reduc stable}, observe
that get a stable algorithm $\cB$ maintaining the spanning forest
$F$, we just need to be able query connectivity between two nodes
in some other spanning forest $F'$ of $G$. Suppose that at any time
$i$, we can query some spanning forest $F'_{i}$ where $F'_{i}$
can be totally different from $F'_{i-1}$. Then, by paying additional
$O(m_{0})$ in preprocessing time and $O(\log^{2}n)$ in update time
as in the proof of \ref{lem:reduc stable}, we can also maintain $F$
stably.
\end{rem}
Next, we prove \ref{lem:reduc poly length}.
\begin{lem}
[Reminder of \ref{lem:reduc poly length}]Suppose there is a stable
$(n,p_{c},p_{t},T)$-algorithm $\cA$ for a 3-bounded-degree graph
where the number of edges is always between $n/4$ and $2n$ with
preprocessing time $t_{p}(n,p_{c},p_{t},T)$ and update time $t_{u}(n,p_{c},p_{t},T)$.
Then, for any $U\ge T$, there is a stable $(n,O(Up_{c}),O(p_{t}),U)$-algorithm
$\cB$ for a 3-bounded-degree graph where the number of edges is always
between $n/4$ and $2n$ with preprocessing time $O(t_{p}(n,p_{c},p_{t},T))$
and update time $O(t_{u}(n,p_{c},p_{t},T)+t_{p}(n,p_{c},p_{t},T)/T+\log^{2}n)$.\end{lem}
\begin{proof}
[Proof of \ref{lem:reduc poly length}]Let $G_{i}$ denote the graph
after the $i$-th update. We divide the update sequence of length
$U$ for the algorithm $\cB$ into phases, each of size $T_{0}=T/2$.
In high level, given the algorithm $\cA$, we will first show an algorithm
$\cA'$ which uses two instances $\cA_{odd}$ and $\cA_{even}$ of
$\cA$ so that during each odd (even) phase, some spanning tree is
maintained by $\cA_{odd}$ ($\cA_{even}$). However, the two spanning
trees at the end of each odd phase and at the beginning of the next
even phase can be totally different. After having $\cA'$, we can
devise the  dynamic $\st$ algorithm $\cB$ as desired stated in \ref{rem:get stable}.
Let us denote $t_{p}=t_{p}(n,p_{c},p_{t},T)$ and $t_{u}=t_{u}(n,p_{c},p_{t},T)$
for convenience. 

Now, we describe the algorithm $\cA'$. In the first phase, we preprocess
$G$ using $\cA_{odd}$ in time $t_{p}(n,p_{c},p_{t},T)$ and handle
$T_{0}$ updates using $\cA_{odd}$. For each non-first phase, we
do the following. Consider each odd phase with time period $[i_{0}+1,i_{0}+T_{0}]$.
During time period $[i_{0}+1,i_{0}+T_{0}/2]$, we preprocess $G_{i_{0}}$
using $\cA_{even}$ by evenly distributing the work so that each step
takes at most $\frac{t_{p}}{T_{0}/2}=O(t_{p}/T)$ time. During time
period $[i_{0}+T_{0}/2,i_{0}+T_{0}]$, at time $i_{0}+T_{0}/2+k$,
we feed the two updates at time $i_{0}+2k-1$ and $i+2k$ to $\cA_{even}$.
Therefore, we have that after the $(i_{0}+T_{0})$-th update which
is the beginning of even phase, $\cA_{even}$ is maintaining a spanning
tree of the current graph $G_{i_{0}+T_{0}}$. As $\cA_{even}$ has
only been only fed $T_{0}$ updates in time period $[i_{0}+T_{0}/2,i_{0}+T_{0}]$,
$\cA_{even}$ can handle $T-T_{0}=T_{0}$ more incoming updates in
the next even phase, taking time $t_{u}$ on each update. We do symmetrically
at each even phase. Therefore, we have that during each odd (even)
phase, some spanning tree is maintained by $\cA_{odd}$ ($\cA_{even}$)
as desired. The preprocessing time and update time of $\cA'$ is $t_{p}$
and $O(t_{u}+t_{p}/T)$ respectively.

After the $i$-th update, let $F'_{i}$ be the spanning forest maintained
by either $\cA_{odd}$ or $\cA_{even}$. We will assume that that
an augmented ET tree from \ref{thm:aug ET tree} is implemented on
it. This slightly increase preprocessing time and update time of $\cA'$
to $t_{p}+O(n)=O(t_{p})$ and $O(t_{u}+t_{p}/T+n/T+\log n)=O(t_{u}+t_{p}/T+\log n)$.
By \ref{rem:get stable}, we can stably maintain a spanning forest
$F$ with the preprocessing time and update time $O(t_{p}+n)=O(t_{p})$
and $O(t_{u}+t_{p}/T+\log^{2}n)$ respectively.

By union bound, $\cB$ is $(n,p'_{c},p_{t}',U)$-algorithm where $p'_{c}=O(Up_{c})$
and $p'_{t}=O(p_{t})$. To be more clear, we have $\cA$ works correctly
on all first $U$ steps with probability $1-Up_{c}$. Hence, $\cB$
also works correctly on \emph{all }first $U$ steps with with probability
$1-Up_{c}$. In particular, $\cB$ works correctly on each of the
first $U$ steps. For the update time, as we run two instances of
$\cA$ and other parts are deterministic, we have that $\cB$ has
preprocessing time and updated time as stated with probability $1-2p_{t}$.\end{proof}
\begin{lem}
[Reminder of \ref{lem:reduc bound degree}]Suppose there is a stable
$(n,p_{c},p_{t},T)$-algorithm $\cA$ for a 3-bounded-degree graph
where the number of edges is always between $n/4$ and $2n$, and
$\cA$ has preprocessing time $t_{p}(n,p_{c},p_{t},T)$ and update
time $t_{u}(n,p_{c},p_{t},T)$. Then there is a stable $(n,O(Tp_{c}),O(p_{t}),T)$-algorithm
$\cB$ for any graph with where the number of edges is always between
$n/4$ and $2n$ with preprocessing time $O(t_{p}(O(n),p_{c},p_{t},T))$
and update time $O(t_{u}(O(n),p_{c},p_{t},T)+\log^{2}n)$.\end{lem}
\begin{proof}
[Proof of \ref{lem:reduc bound degree}]Let us define an $(n,p_{c},p_{t},T)$-algorithm
$\cA'$ for 2-weight $\mst$ as an $(n,p_{c},p_{t},T)$-algorithm
that maintains a 2-weight $\mst$ instead of just a spanning forest.
The reduction from dynamic $\st$ to dynamic $2$-weight $\mst$ in
\cite[Section 3.1]{HenzingerK99} can be stated as follows: given
that there is $\cA$, there is a stable $(n,O(Tp_{c}),O(p_{t}),T)$-algorithm
$\cA'$ for 2-weight $\mst$ for a 3-bounded-degree graph where the
number of edges is always between $n/4$ and $2n$, and $\cA'$ has
preprocessing time $O(t_{p}(n,p_{c},p_{t},T))$ and update time $O(t_{u}(n,p_{c},p_{t},T)+\log^{2}n)$.
Now, we show how to get $\cB$ from $\cA'$.

Let $G=(V,E)$ be an input graph for $\cB$ where $n/2\le|E|\le2n$
but max degree of $G$ maybe more than $3$. Let $G'$ be obtained
from $G$ by ``splitting'' each node in $G$ into a path in $G'$
of length equals to its degree. To be more precise, for each node
$u\in V$ with incident edges $(u,v_{1}),\dots,(u,v_{s})$, we create
a path $P_{u}=(u_{1},\dots,u_{s})$ in $G'$ and there are edges $(u_{i},v_{i})$
for every $i$. The number of nodes in $G'$ is still $O(n)$ because
there are $O(n)$ edges in $G$. For each edge update of $G$, we
can update $G'$ accordingly using 3 edge updates. 

Now, we assign weight 0 to edges in those paths, and 1 to original
edges of $G'$. Observe that tree-edges with weight 1 of a spanning
forest in $G'$ form a spanning forest of $G$. Therefore, the algorithm
$\cB$ just runs a $(n,O(Tp_{c}),O(p_{t}),T)$-algorithm $\cA'$ for
2-weight $\mst$ on $G'$ with preprocessing time $O(t_{p}(O(n),p_{c},p_{t},T))$
and update time $O(t_{u}(O(n),p_{c},p_{t},T)+\log^{2}n)$ and maintain
2-weight minimum spanning forest $F'$ in $G'$. By reporting only
the updates of the tree-edges with weight 1 in the 2-weight $\mst$
$F'$ of $G'$, we obtain updates for a spanning forest of $G$.\end{proof}
\begin{lem}
[Reminder of \ref{lem:reduc lower bound edge}]Suppose there is a
stable $(n,p_{c},p_{t},T)$-algorithm $\cA$ for any graph where the
number of edges is always between $n/4$ and $2n$, and $\cA$ has
preprocessing time $t_{p}(n,p_{c},p_{t},T)$ and update time $t_{u}(n,p_{c},p_{t},T)$.
Then there is a stable $(n,O(Tp_{c}),O(p_{t}),T)$-algorithm $\cB$
for any graph with $n$ nodes where the number of edges is always
at most $2n$, and $\cB$ has preprocessing time $O(t_{p}(O(m_{0}),p_{c},p_{t}))$
where $m_{0}$ is a number of initial edges, and update time $O(t_{u}(O(m),p_{c},p_{t})+\log^{2}n)$
where $m$ is a number of edges when update.\end{lem}
\begin{proof}
First, we define the notion of \emph{reduced graphs}. Let $G_{i}=(V,E_{i})$
be a graph after the $i$-th update for any $i$. Let $V_{i}$ be
the set of all endpoints of edges in $E_{i}$. The reduced graph $G'$
of $G_{i}$ is defined as follows: $G'=(V_{i}\cup X,E_{i})$ where
$X$ is a set of $|E_{i}|$ \emph{auxiliary} \emph{nodes} which are
isolated. So $|E_{i}|\le|V(G')|\le2|E_{i}|$. After time $i$, suppose
an edge $(u,v)$ is inserted in $G$. If $u\notin V_{i}$, i.e. it
is not an endpoint of edges in $E_{i}$, we let $u'$ be some isolated
node in $X$, add pointers between $u$ and $u'$, and say that $u'\in X$
is used. Otherwise $u\in V_{i}$, we let $u'=u$. We do similarly
for $v$. Then we insert $(u',v')$ into $G'$. When $(u,v)$ is deleted
in $G$, we delete $(u',v')$ in $G'$. Suppose not all auxiliary
nodes in $X$ are used, and suppose that we can maintain a spanning
forest in $G'$, it is obvious how to maintain a spanning forest in
$G$ because there are pointers that maps between used node in $X$
and nodes in $G$. Note also that we can easily construct the reduced
graph $G'$ of $G_{i}$ in time $O(|E_{i}|)$ (as we assume the standard
assumption that the list of edges of $E_{i}$ is given).

We will use exactly the same technique as in the proof of \ref{lem:reduc poly length}.
The update sequence is again divided into phases. We will use two
instances $\cA_{odd}$ and $\cA_{even}$ of $\cA$ so that during
each odd (even) phase, some spanning tree is maintained by $\cA_{odd}$
($\cA_{even}$). After having this, we can devise the stable  dynamic
$\st$ algorithm $\cB$ using augmented ET trees as in the proof of
\ref{lem:reduc poly length}. However, in this reduction, the length
of each phase is not fixed as in the proof of \ref{lem:reduc poly length}.
Here, the number of edges in the current graph will define the length
of the next phase. To be more precise,  suppose that the current phase
is a period $[i_{0}+1,i_{0}+P_{0}]$. Then the length of the next
phase is $|E(G_{i_{0}+P_{0}/2})|/4$. The length of the first phase
is set to be $|E(G_{0})|/2$. By this, we have the following invariant. 
\begin{claim}
For any phase of a period $[i_{1}+1,i_{1}+P_{1}]$, we have $|E(G_{i_{1}})|/8\le P_{1}\le|E(G_{i_{1}})|/2$.\label{claim:length of phase}\end{claim}
\begin{proof}
We prove by induction. Let $[i_{0}+1,i_{0}+P_{0}]$ be the period
of the previous phase, i.e. $i_{1}=i_{0}+P_{0}$. We claim that $|E(G_{i_{0}})|\le2|E(G_{i_{1}})|$.
Otherwise $|E(G_{i_{0}})|-P_{0}\le|E(G_{i_{0}+P_{0}})|=|E(G_{i_{1}})|<|E(G_{i_{0}})|/2$,
which means $P_{0}>|E(G_{i_{0}})|/2$, contradicting the induction
hypothesis. 

By definition, we have $P_{1}=\frac{|E(G_{i_{0}+P_{0}/2})|}{4}=\frac{|E(G_{i_{1}-P_{0}/2})|}{4}$.
So we have 
\[
P_{1}=\frac{|E(G_{i_{1}-P_{0}/2})|}{4}\le\frac{|E(G_{i_{1}})|+P_{0}/2}{4}\le\frac{|E(G_{i_{1}})|+|E(G_{i_{0}})|/4}{4}\le\frac{|E(G_{i_{1}})|+|E(G_{i_{1}})|/2}{4}<|E(G_{i_{i}})|/2
\]
,and we also have
\[
P_{1}=\frac{|E(G_{i_{1}-P_{0}/2})|}{4}\ge\frac{|E(G_{i_{1}})|-P_{0}/2}{4}\ge\frac{|E(G_{i_{1}})|-|E(G_{i_{0}})|/4}{4}\ge\frac{|E(G_{i_{1}})|-|E(G_{i_{1}})|/2}{4}=|E(G_{i_{1}})|/8
\]

\end{proof}
Now, we are ready to describe the algorithm. To preprocess the initial
graph $G_{0}$, we construct a reduced graph $G'_{0}$ of $G_{0}$
and preprocess the reduced graph $G_{0}'$ using $\cA_{odd}$ using
$O(t_{p}(O(m_{0}),p_{c},p_{t})+m_{0})=O(t_{p}(O(m_{0}),p_{c},p_{t}))$
time in total where $m_{0}$ is a number of initial edges. Then we
maintain a spanning forest in the reduced graph $G_{0}'$ using $\cA_{odd}$
in the first phase. 

For each odd phase with period $[i_{0}+1,i_{0}+P_{0}]$, we maintain
a spanning forest $F_{odd}$ in the reduced graph using $\cA_{odd}$.
Note that, by \ref{claim:length of phase}, the number $m'$ of edges
in the reduced graph maintained by $\cA_{odd}$ is always between
$n'/4\le m'\le n'$ where $n'$ is the number of nodes in the reduced
graph. This satisfies the condition of the underlying graph of $\cA_{odd}$.
Additionally, during time $[i_{0}+\frac{P_{0}}{2}+1,i_{0}+\frac{3}{4}P_{0}]$,
we evenly distribute the work for constructing a reduced graph $G'_{i_{0}+P_{0}/2}$
of $G{}_{i_{0}+P_{0}/2}$, and preprocessing $G'_{i_{0}+P_{0}/2}$
using $\cA_{even}$. During time $[i_{0}+\frac{3}{4}P_{0}+1,i_{0}+P_{0}]$,
we evenly feed the updates of period $[i_{0}+\frac{P_{0}}{2}+1,i_{0}+P_{0}]$
to $\cA_{even}$, so that at the beginning of even phase at time $i_{0}+P_{0}+1$,
$\cA_{even}$ is maintaining a spanning tree of the current graph.
As the number of edges during this phase is $\Theta(P_{0})$ by \ref{claim:length of phase},
the update time in this odd phase is $O(t_{p}(O(P_{0}),p_{c},p_{t})/P_{0}+t_{u}(O(P_{0}),p_{c},p_{t}))=O(t_{p}(O(m),p_{c},p_{t})/m+t_{u}(O(m),p_{c},p_{t}))=O(t_{u}(O(m),p_{c},p_{t}))$
where $m$ is the number of edges at the time of update. Note that
$t_{p}(O(m),p_{c},p_{t})=O(m\cdot t_{u}(O(m),p_{c},p_{t}))$ because
we can always preprocess the graph of $m$ edges by just updating
each edge one-by-one. We work symmetrically in even phase. Therefore,
we during each odd (even) phase, some spanning tree is maintained
by $\cA_{odd}$ ($\cA_{even}$) as desired. By \ref{rem:get stable},
we are done. The probability of failure is again bounded by union
bound.
\end{proof}

\subsubsection{Proofs of \ref{lem:reduc sparse} and \ref{lem:reduc infty length}}

In this section, we give the proofs of \ref{lem:reduc sparse} and
\ref{lem:reduc infty length}. 
\begin{lem}
[Reminder of \ref{lem:reduc sparse}]Suppose there is a stable $(n,p_{c},p_{t},T)$-algorithm
$\cA$ for a graph with $n$ nodes where the number of edges is always
at most $2n$, and $\cA$ has preprocessing time $t_{p}(m_{0},p_{c},p_{t},T)$
where $m_{0}$ is a number of initial edges and update time $t_{u}(m,p_{c},p_{t},T)$
where $m$ is a number of edges when update. Suppose that for any
$n_{1},n_{2}$, $t_{p}(m_{1},p_{c},p_{t},T)+t_{p}(m_{2},p_{c},p_{t},T)\le t_{p}(m_{1}+m_{2},p_{c},p_{t},T)$,
then there is a stable $(n,O(np_{c}),O(np_{t}),T)$-algorithm $\cB$
for any graph with $n$ nodes with preprocessing time $O(t_{p}(m_{0},p_{c},p_{t},T)\log n)$
where $m_{0}$ is a number of initial edges and update time $O(t_{u}(O(n),p_{c},p_{t},T))$\end{lem}
\begin{proof}
We assume first that $\cA$ is deterministic and we will remove the
assumption later. We partition a complete graph with $n$ nodes into
$n$ arbitrary trees $T_{1},\dots,T_{n}$. For example, $T_{i}$ can
be a star having the vertex $i$ as a core and vertices $1,\dots,i-1$
as leaves. Then we define a complete binary tree called \emph{sparsification
tree} $\cT$ with $n$ leaves corresponding to $T_{1},\dots,T_{n}$.
For each node $x\in\cT$, let $\cT_{x}$ be the subtree of $\cT$
rooted at $x$, and $x$ corresponds to a graph $H_{x}$ formed by
all edges of $T_{i}$ where $T_{i}$ are the leaves in $\cT_{x}$.
The root of $\cT$ corresponds to the complete graph. 

Let $G$ be the dynamic graph which is an input of the algorithm $\cB$.
For each node $x\in\cT$, we define $G_{x}$ to be a graph such that
$E(G_{x})=E(G)\cap E(H_{x})$. For every node $x$, let $G'_{x}$
be a certain subgraph (to be define next) of $G_{x}$ which is an
underlying graph of an instance of the algorithm $\cA$. Let $F_{x}$
be the spanning forest of $G'_{x}$ maintained by $\cA$. The following
invariant about $G'_{x}$ must hold: if $x$ is a leaf of $\cT$,
then we set $G'_{x}=G_{x}$. If $v$ is an internal node of $\cT$
with $y$ and $z$ as children, then $G'_{x}=F_{y}\cup F_{z}$. Observe
that for any $x\in\cT$, $|E(G'_{x})|\le2n$. Observe further that
$F_{x}$ is not only a spanning forest of $G'_{x}$ but, by induction,
it is also a spanning forest of $G_{x}$. Hence, at the root $r$
of $\cT$, $F_{r}$ is a spanning forest of $G$ as desired.

Now, we are ready to describe how $\cB$ works. $\cB$ preprocesses
the initial graph $G$ as follows. First, for each $x\in\cT$, since
$G'_{x}$ is an underlying graph of an instance of the algorithm $\cA$,
$\cA$ requires that $G'_{x}$ is represented as a list of its edge.
We first do this by scanning edges in $G$, and obtain the list of
edges on $G'_{x}$ for all leaves $x\in\cT$. Then we preprocess those
$G'_{x}$ using the algorithm $\cA$ and obtain $F_{x}$. For each
non-leaf node $x\in\cT$, we proceed in bottom-up manner: let $y,z\in\cT$
be the children of $x$, we scan the list of edges in $F_{y}\cup F_{z}=G'_{x}$
and preprocess $G'_{x}$ and obtain $F_{x}$. To bound the running
time, let $\cT^{d}$ be a set of nodes in $\cT$ with depth $d$.
For different nodes $x,y\in\cT^{d}$ , we know that $G'_{x}$ and
$G'_{y}$ are edge disjoint. So $\sum_{x\in\cT^{d}}|E(G'_{x})|\le m$,
and hence the preprocessing time spent on graphs $G'_{x}$, for all
$x\in\cT^{d}$, is $\sum_{x\in\cT^{d}}(t_{p}(n,|E(G'_{x})|,p_{c},p_{t},T))+O(|E(G'_{x})|)=O(t_{p}(n,m_{0},p_{c},p_{t},T))$
by the assumption about $t_{p}(\cdot)$. Since $\cT$ has depth $O(\log n)$,
the total preprocessing time is $O(t_{p}(n,m_{0},p_{c},p_{t},T)\log n)$.

Next, we describe how to update. The goal is to maintain the invariant
that $G'_{x}=G_{x}$ for each leaf $x\in\cT$ and $G'_{x}=F_{y}\cup F_{z}$
for each non-leaf $x\in\cT$ with children $y$ and $z$. When an
edge $e$ is updated in $G$, let $x$ be the unique leaf of $\cT$
that $H_{x}$ contains $e$. Let $(x=x_{0},x_{1},\dots,x_{s}=r)$
be the path in $\cT$ from $x$ to the root $r$ of $\cT$. There
are two cases. 

First, suppose the update is to insert $e=(u,v)$. Let $k$ be the
smallest index such that $u$ and $v$ were connected in $G'_{x_{k}}$
before the insertion. For $j=1$ to $k-1$, $u$ and $v$ were \emph{not}
connected in $G'_{x_{j}}$, so $e$ will be added into $F_{x_{j}}$
and this is the only change in $F_{x_{j}}$ as $F_{x_{j}}$ is stable.
So only $e$ is inserted into $G'_{x_{j+1}}$ in the next level. For
$j=k$, $u$ and $v$ were connected in $G'_{x_{j}}$ and $e$ will
be inserted and become a non-tree edge in $G'_{x_{j}}$ as $F_{x_{j}}$
is stable. So there is no further update to the next level. Therefore,
we only do at most one update on $O(\log n)$ graphs, which takes
$O(t_{u}(n,2n,p_{c},p_{t},T)\log n)$ time in total.

Second, suppose the update is to delete $e=(u,v)$. Let $k$ be the
largest index such that $u$ and $v$ were a tree-edge of $F_{x_{k}}$
before the deletion. So $e$ is also a non-tree edge in $G'_{x_{k+1}}$
and $e$ is not in any $G'_{x_{j}}$ for $j>k+1$. First, we delete
the edge $e$ from all $G'_{x_{0}},\dots,G'_{x_{k+1}}$ using $\cA$
and find the replacement edge in each graph if exists. Let $\ell$
be the smallest index such that $u$ and $v$ are still connected
in $G'_{x_{\ell}}$after the deletion. For $j<\ell$, there is no
replacement edge. So removing $e$ from $F_{x_{j}}$ is the only change
by stability of $F_{x_{j}}$. To maintain the invariant, there is
no further update needed to be propagate to $G'_{x_{j+1}}$, since
$e$ is already removed from $G'_{x_{j+1}}$. For $j$ such that $\ell\le j\le k$,
there is also a replacement edge $e_{j}$ newly added into $F_{x_{j}}$.
By stability of $F_{x_{j}}$, only $e$ is removed and $e_{j}$ is
added into $F_{x_{j}}$. Since $e$ is already removed from $G'_{x_{j+1}}$,
we only need to insert $e_{j}$ into $G'_{x_{j+1}}$. But we know
that $e_{j}\notin F_{x_{j+1}}$ because $u$ and $v$ are already
reconnected in $F{}_{x_{j+1}}$ by $e_{j+1}$. Hence, the insertion
of $e_{j}$ into $G'_{x_{j+1}}$ will not be propagated into the next
level $j+2$. For $j=k+1$, we know that $e$ is not a tree-edge.
So $F_{x_{k+1}}$ does not change and also and so there is no further
update in the upper level. Therefore, we only do at most two update
on $O(\log n)$ graphs, which takes $O(t_{u}(n,2n,p_{c},p_{t},T)\log n)$
time in total.

Finally, to remove the assumption that $\cA$ is deterministic, we
know that if all instances of $\cA$ works correctly on $G'_{x}$
for all $x\in\cT$, then $F_{r}$ is a spanning forest of $G$ where
$r$ is the root of $\cT$. There are $O(n)$ nodes in $\cT$ and
each instance works correctly at any time with probability $1-p_{c}$.
By union bound $F_{r}$ is a spanning forest of $G$ with probability
$1-O(np_{c})$. By the same argument, as there are $n$ instances
of $\cA$, the preprocessing time and update time of $\cB$ are bounded
as state with probability $1-O(np_{t})$.

We note that here, we only show that the algorithm $\cB$ has update
time $O(t_{u}(n,2n,p_{c},p_{t},T)\log n)$. However, a logarithmic
factor can be shaved off using a more clever sparsification tree $\cT$
as described in \cite[Section 3.1]{EppsteinGIN97} where each node
has 3 or 4 leaves and the number of vertices at level $2i$ is at
most $O(n/2^{i})$.\end{proof}
\begin{lem}
[Reminder of \ref{lem:reduc infty length}]Suppose there is an $(n,p_{c},p_{t},T)$-algorithm
$\cA$ for any graph with $n$ nodes with preprocessing time $t_{p}(m_{0},p,p_{t},T)$
where $m_{0}$ is a number of initial edges and update time $t_{u}(n,p_{c},p_{t},T)$.
If $T\ge n^{2}$, then there is an $(n,O(Tp_{c}),O(p_{t}),\infty)$-algorithm
$\cB$ for any graph with $n$ nodes with preprocessing time $O(t_{p}(m_{0},p_{c},p_{t},T))$
where $m_{0}$ is a number of initial edges and update time $O(t_{u}(n,p_{c},p_{t},T))$.\end{lem}
\begin{proof}
The reduction from dynamic $\st$ to dynamic $2$-weight $\mst$ in
\cite[Section 3.1]{HenzingerK99} can be stated as follows: given
that there is $\cA$, there is a stable $(n,p'_{c},p'_{t},T)$-algorithm
$\cA'$ for 2-weight $\mst$ for any graph where $p'_{c}=O(Tp_{c})$
and $p'_{t}=O(p{}_{t})$, and $\cA'$ has preprocessing time $O(t_{p}(n,p_{c},p_{t},T))$
and update time $O(t_{u}(n,p_{c},p_{t},T)+\log^{2}n)$. We will use
$\cA'$ in the reduction below. Let $\cA^{odd}$ and $\cA^{even}$
be two instances of the algorithm $\cA$. We denote $F^{odd},F^{even}$
and $F'$ as the spanning forests maintained by $\cA^{odd}$, $\cA^{even}$
and $\cA'$ respectively. For any time $i$, $F_{i}^{odd},F_{i}^{even}$
and $F_{i}'$ are $F^{odd},F^{even}$ and $F'$ after the $i$-th
update respectively. We divide the sequence of updates in phases.
Each phase has size $P=T/2$. We will first show the reduction assuming
that $\cA^{odd}$, $\cA^{even}$ and $\cA'$ are deterministic, and
then we will remove the assumption later. For convenience, we denote
$t_{p}(m)=t_{p}(m_{0},p_{c},p_{t},T)$ for any $m$.

Suppose that $i_{0}$ is the ending time of some odd phase. For each
update at time $i_{0}+1$ to $i_{1}=i_{0}+2t_{p}(n^{2})$, we feed
each update to $\cA^{odd}$ taking time $t_{u}(n)$ per update and
then report the answer from $\cA^{odd}$ (the list of tree-edges to
be added or removed from the spanning forest $F^{odd}$ maintained
by $\cA^{odd}$). But, in addition, during this period of time we
will also spend time on $\cA^{even}$ as follows. We distribute evenly,
from time $i_{0}+1$ to $i_{0}+t_{p}(n^{2})$, the work of $\cA^{even}$
for preprocessing the graph $G_{i_{0}}$ at time $i_{0}$ which has
at most $n^{2}$ edges. This takes $O(1)$ time per update. So, at
time $i_{0}+t_{p}(n^{2})$, $\cA^{even}$ finishes preprocessing $G_{i_{0}}$
which is a graph from $t_{p}(n^{2})$ steps before. For each time
$i_{0}+t_{p}(n^{2})+k$ where $k\in[1,t_{p}(n^{2})]$, we feed the
two updates from time $i_{0}+2k-1$ and $i_{0}+2k$. This takes $O(t_{u}(n))$
time per update. Hence, at time $i_{1}=i_{0}+2t_{p}(n^{2})$, $\cA^{even}$
is maintaining a spanning forest $F_{i_{1}}^{even}$ of the current
graph $G_{i_{1}}$.

Next, for each update at time $i_{1}+1$ to $i_{2}=i_{1}+t_{p}(0)+n$,
we feed each update to both $\cA^{odd}$ and $\cA^{even}$ taking
time $O(t_{u}(n))$ per update, but we only report the answer from
$\cA^{odd}$. During this period of time we will, in addition, spend
time on $\cA'$. From time $i_{1}+1$ to $i_{1}+t_{p}(0)$, we distribute
evenly the work of $\cA'$ for preprocessing an empty graph.\footnote{Note that $t_{p}(0)$ may not be trivial. For instance, $t_{p}(0)$
can be linear in $n$.} At time $i_{1}+t_{p}(0)$, $\cA'$ is ready to handle an update.
Let $G'$ be the underlying graph of $\cA'$. Now, the goal is that
at time $i_{2}$ we have edges of $G'$ is exactly $F^{odd}$ (i.e.
$G'_{i_{2}}=F_{i_{2}}^{odd})$, and each edge has weight 1. To do
this, for each time $i_{1}+t_{p}(0)+1$ to $i_{2}=i_{1}+t_{p}(0)+n$,
we have the invariant that $E(G')\subseteq F^{odd}$. We insert a
constant number, say three, of edges of weight 1 from $F^{odd}$ to
$G'$ and delete an edge $e$ from $G'$ whenever $e$ is removed
from $F^{odd}$. Since $F^{odd}$ is stable, at each step there is
at most one edge removed from $F^{odd}$. But because we insert strictly
more than one edge at each step, and $F^{odd}$ has at most $n-1$
edges, and we spend at most $n$ steps doing this, taking time $O(t_{u}(n))$
per step. We conclude $G'_{i_{2}}=F_{i_{2}}^{odd}$. Note that we
have $F'_{i_{2}}=F_{i_{2}}^{odd}$ as well.

Next, from each update at time $i_{2}+1$ to $i_{3}=i_{2}+n$, we
feed each update to all $\cA^{odd}$, $\cA^{even}$ and $\cA'$. But
we use report the answer from $\cA'$. Note that we can switch from
reporting the answer from $\cA^{odd}$ to reporting from $\cA'$ because
$F'_{i_{2}}=F_{i_{2}}^{odd}$. Now, the goal is that at time $i_{3}$
we have $G'_{i_{3}}=F_{i_{3}}^{odd}\cup F_{i_{3}}^{even}$ and each
edge $e\in F^{even}$ has weight 0 in $G'$. Therefore, we have $F'_{i_{3}}=F_{i_{3}}^{even}$.
Even when having the invariant that $F^{odd}\subseteq E(G')\subseteq F^{odd}\cup F^{even}$,
this can still be done similarly as above by doing a constant number
of updates to $G'$ at each step, taking time $O(t_{u}(n))$ per update.
At time $i_{3}$, we now have $F'_{i_{3}}=F_{i_{3}}^{even}$ and we
terminate $\cA^{odd}$ and $\cA'$ (they will be initialized again
around the ending time of the next phase).

Next, from time $i_{3}+1$ to $i_{0}+P$ which is the ending time
of next phase, we just feed each update to $\cA^{even}$ and also
report its answer. Note again that we can switch from reporting the
answer from $\cA'$ to reporting from $\cA^{even}$ because $F'_{i_{3}}=F_{i_{2}}^{even}$. 

We need that $i_{3}+1\le i_{0}+P$. This is true because $i_{3}-i_{0}+1\le2t_{p}(n^{2})+t_{p}(0)+n+n+1\le6t_{p}(n^{2})\le P$.
We also need that $\cA^{odd}$ and $\cA'$ need to handle at most
$T$ updates before they are terminated. This is true because $\cA^{odd}$
is surely initialized after time $i_{0}-P$ and terminated before
time $i_{0}+P$, and $T=2P$. Similar argument holds for $\cA^{even}$
and $\cA'$. 

To conclude, we have obtained an algorithm $\cB$ that alternatively
invoke $\cA^{odd}$ and $\cA^{even}$ on each phase, and periodically
invoke $\cA'$ during the transition between each phase. The preprocessing
time is $t_{p}(m)$, which is exactly preprocessing time of $\cA^{odd}$.
The update time is $O(t_{u}(n))$ time. 

Now, we remove the assumption that $\cA^{odd}$, $\cA^{even}$ and
$\cA'$ are deterministic, and we want to prove that the probability
at the maintained forest by $\cB$ spans the underlying graph at each
step with probability at least $1-O(Tp_{c})$. Observe that, at any
time, the forest $F^{\cB}$ maintained by $B$ is the forest maintained
by either $\cA^{odd}$, $\cA^{even}$ or $\cA'$, i.e. $F^{odd}$,
$F^{even}$, or $F'$ respectively. From the algorithm, when we report
either $F^{odd}$ or $F^{even}$, the underlying graph is $G$ itself.
So $F^{odd}$ and $F^{even}$ are both spanning forests of $G$ with
probability at least $1-p_{c}$. However, when we report $F'$, the
underlying graph is $G'$ where $F^{odd}\subseteq E(G')\subseteq F^{odd}\cup F^{even}$.
Since $F^{odd}$ and $F^{even}$ are spanning forests of $G$ with
probability $1-p_{c}$ and $F'$ is a spanning forest of $G'$ with
probability $1-p_{c}'$, we have $F'$ is a spanning forest of $G$
with probability at least $1-p_{c}'-2p_{c}=1-O(Tp_{c})$ by union
bound. Using union bound, we have that the failure probability about
running time is $O(p_{t})$.
\end{proof}

\subsection{\label{sec:proof most balanced}Proof of \ref{thm:most balanced sparse cut}}

We prove \ref{thm:most balanced sparse cut} by extending the cut-matching
game of \cite{KhandekarRV09}. In the following, the notion of graph
embedding is required. We say that a weighted graph $H=(V,E_{H},w_{H})$
can be \emph{embedded} in another weighted graph $G=(V,E_{G},w_{G})$
with congestion $C$ iff a flow of $w_{H}(f)$ units can be routed
in $G$ between the end-points of $f$ \emph{simultaneously} for all
$f\in E_{H}$ without violating the edge-capacities $w_{G}(e)$, for
any edge $e\in E_{G}$, by a factor more than $C$. 
we state an easy observation:
\begin{prop}
\label{prop:embed graph} Let $G$ and $H$ be two unweighted undirected
graphs with the same set of vertices. If $H$ can be embedded in $G$
with congestion $C$, then $\phi(G)\ge\phi(H)/C$. Moreover, if every
cut of $H$ of size at least $s$ has expansion at least $\alpha$,
then every cut of $G$ of size at least $s$ has expansion at least
$\alpha/C$. 
\end{prop}
Next, we prove some lemmas which are extending the cut-matching game
of \cite{KhandekarRV09}.
\begin{lem}
\label{lem:sparse from approx max flow}Given an undirected graph
$G=(V,E)$ with $n$ vertices and $m$ edges, an expansion parameter
$\alpha>0$, and an approximation ratio $\gamma\ge1$, there is a
randomized algorithm that with high probability outputs either 
\begin{itemize}
\item an $\alpha$-sparse cut $S$ where $|S|\le n/2$ or 
\item a graph $H$ embeddable in $G$ with congestion at most $O(\gamma\log^{3}n/\alpha)$
where $\phi(H)\ge1/2$. 
\end{itemize}

Furthermore the algorithm runs in $O(\gamma\log^{4}n\cdot T(m))$
time where $T(m)$ is the time for computing $\gamma$-approximate
max-flow in an undirected graph with $m$ edges.

\end{lem}
\begin{proof}
We first summarize the cut-matching game \cite{KhandekarRV09} that
our proof is based on. Given an expansion parameter $\alpha$ on a
graph $G=(V,E)$, there are two players, a cut player and a matching
player, interacting with each other for $r=\Theta(\log^{2}n)$ rounds.
In the $i$-th round, the cut player will choose a cut $(S_{i},V\setminus S_{i})$
of size $|S_{i}|=n/2$, and then the matching player either 1) outputs
a perfect matching $M_{i}$ between $S_{i}$ and $V\setminus S_{i}$
and the game is continued to the next round, or 2) outputs an $\alpha$-sparse
sparse cut and terminate the game.

It is shown in \cite{KhandekarRV09} that, after $r$ rounds, if the
matching player did not output an $\alpha$-sparse cut, then a union
of the matchings $H=\bigcup_{i=1}^{r}M_{i}$ is such that $\phi(H)\ge1/2$.
Moreover, the time needed for the cut player in each round is only
$O(n)$. For the matching player, by computing \emph{exact} max flow
in each round, the matching $M_{i}$ found by the matching player
can be embedded in $G$ with congestion at most $c/\alpha$ where
$c=1$. Hence, $H$ can be embedded in $G$ with congestion $r\times c/\alpha=O(\log^{2}n/\alpha)$.
However, computing \emph{exact} max flow is too costly for us.

We will show that if we relax $c=\Theta(\gamma\log n)$, then the
matching player only needs to compute $c$ many $\gamma$-approximate
max flow computations to finish each round. On the high level, in
the $i$-th round of matching player, given a cut $(S_{i},V\setminus S_{i})$
from the cut player, there will be $c$ sub-rounds. In each sub-round,
the matching player will compute a $\gamma$-approximate max flow.
He will either output an $\alpha$-sparse cut, or, after $c$ sub-rounds,
output a perfect matching to complete the $i$-th round for the cut-matching
game. We describe the procedure precisely below.

In the $j$-th sub-round, let $T_{1}^{(j)},T_{2}^{(j)}\subset V$
be two disjoint sets of vertices each of size $k^{(j)}=|T_{1}^{(j)}|=|T_{2}^{(j)}|$.
In the first sub-round, $T_{1}^{(1)}=S_{i}$ and $T_{2}^{(1)}=V\setminus S_{i}$.
Let $G^{(j)}$ be a \emph{weighted} undirected graph obtained from
$G$ by adding two new vertices $t_{1}$ and $t_{2}$. There are edges
in $G^{(j)}$, with unit capacity, from $t_{1}$ to every vertex in
$T_{1}^{(j)}$, and similarly from $t_{2}$ to $T_{2}^{(j)}$. The
original edges of $G$ have capacity $1/\alpha$ in $G^{(j)}$. The
matching player runs a $\gamma$-approximate max flow algorithm in
$G^{(j)}$, which also produces an $\gamma$-approximate min cut $C^{(j)}$. 
\begin{claim}
\label{claim:small cut then sparse}If $\delta_{G^{(j)}}(C^{(j)})<k^{(j)}$
then $\phi_{G}(C^{(j)}\setminus\{t_{1}\})<\alpha$. \end{claim}
\begin{proof}
This is similar to \cite[Lemma 3.7]{KhandekarRV09}. Let $n_{t_{1}}$
and $n_{t_{2}}$ be a number of cut edges $\partial_{G^{(j)}}(C^{(j)})$
incident to $t_{1}$ and $t_{2}$ respectively. Observe that $\delta_{G^{(j)}}(C^{(j)})=\frac{1}{\alpha}\delta_{G}(C^{(j)}\setminus\{t_{1}\})+n_{t_{1}}+n_{t_{2}}$.
So $\delta_{G}(C^{(j)}\setminus\{t_{1}\})=\alpha(\delta_{G^{(j)}}(C^{(j)})-n_{t_{1}}-n_{t_{2}})$.
But $|C^{(j)}\setminus\{t_{1}\}|\ge k^{(j)}-n_{t_{1}}>\delta_{G^{(j)}}(C^{(j)})-n_{t_{1}}$
and similarly we can show that $|V\setminus\{C^{(j)}\cup\{t_{2}\}\}|>\delta_{G^{(j)}}(C^{(j)})-n_{t_{2}}$,
therefore, $\phi_{G}(C^{(j)}\setminus\{t_{1}\})=\frac{\delta_{G}(C^{(j)}\setminus\{t_{1}\})}{\min\{|C^{(j)}\setminus\{t_{1}\}|,|V\setminus\{C^{(j)}\cup\{t_{2}\}\}|\}}<\alpha$. 
\end{proof}
Hence, if the total weight of the approximate min cut $C^{(j)}$ is
$\delta_{G^{(j)}}(C^{(j)})<k^{(j)}$, then the matching player can
return the smaller side of $(C^{(j)}\setminus\{t_{1}\},V\setminus\{C^{(j)}\cup\{t_{2}\}\})$
which is the desired $\alpha$-sparse cut, and terminate the cut-matching
game. Otherwise, we know that the size of the max flow is at least
$k^{(j)}/\gamma$, and by \cite{KangP15} the flow rounding algorithm
in $O(m\log(n^{2}/m))$ time, we can assume that the flow is integral.
We iteratively decompose the flow into paths using dynamic tree \cite{SleatorT83}
in total time $O(m\log n)$. Then we construct a partial matching
$M_{i}^{(j)}$ between $T_{1}^{(j)}$ and $T_{2}^{(j)}$ of size at
least $k^{(j)}/\gamma$ as follows: For each decomposed path $P$,
if $P=(t_{1},u_{1},\dots,u_{2},t_{2})$ where $u_{1}\in T_{1}^{(j)}$
and $u_{2}\in T_{2}^{(j)}$ then we add $(u_{1},u_{2})\in M_{i}^{(j)}$.
There are at least $k^{(j)}/\gamma$ paths because of the size of
the max flow, hence the size of $M_{i}^{(j)}$ is at least $k^{(j)}/\gamma$.
Let $T_{1}^{(j+1)}=T_{1}^{(j)}\setminus V(M_{i}^{(j)})$ and $T_{2}^{(j+1)}=T_{2}^{(j)}\setminus V(M_{i}^{(j)})$.
As $|T_{i}^{(j+1)}|\le(1-\frac{1}{\gamma})|T_{i}^{(j)}|$, there are
at most $c=\Theta(\gamma\log n)$ sub-rounds before $T_{1}^{(j+1)},T_{2}^{(j+1)}=\emptyset$.
If the game is not terminated until then, the matching player returns
a perfect matching $M_{i}=\bigcup_{j}M_{i}^{(j)}$ between $S_{i}$
and $V\setminus S_{i}$. Observe that $M_{i}$ can be embedded in
$G$ with congestion $c/\alpha$ as desired, because each $M_{i}^{(j)}$
can be embedded in $G$ with congestion $1/\alpha$ as certified by
the paths from the flow decomposition.

To conclude, when we relax $c=O(\gamma\log n)$, we can simulate one
round of matching player in the cut matching game of \cite{KhandekarRV09}
using $c$ sub-rounds where each sub-round needs just one $\gamma$-approximate
max flow computation. As a result, we either get an $\alpha$-sparse
cut, or a graph $H=\bigcup_{i=1}^{r}M_{i}$ embeddable in $G$ with
congestion $O(r\times c/\alpha)=O(\gamma\log^{3}n/\alpha)$ where
$\phi(H)\ge1/2$. The total cost for the cut player is $r\times O(n)$.
The total cost for the matching player is $rc\times O(T(m)+m\log n)$.
So the total running time is $O(\gamma\log^{4}n\cdot T(m))$. 

\end{proof}
Next, we extend the previous proof so that we can compute a sparse
cut of any specified size $s$. 
\begin{lem}
\label{lem:sparse s-size from approx max flow}Given an undirected
graph $G=(V,E)$ with $n$ vertices and $m$ edges, an expansion parameter
$\alpha>0$, an approximation ratio $\gamma\ge1$, and a size parameter
$s=O(\frac{n}{\log^{2}n})$, there is a randomized algorithm that
with high probability outputs either 
\begin{itemize}
\item an $\alpha$-sparse cut $S$ where $s\le|S|\le n/2$, or 
\item a graph $H$ embeddable in $G$ with congestion at most $O(\gamma\log^{3}n/\alpha)$
where every cut in $H$ of size $\Omega(s\log^{2}n)$ has expansion
at least $1/4$. 
\end{itemize}

Furthermore the algorithm runs in $O(\gamma\log^{4}n\cdot T(m))$
time where $T(m)$ is the time for computing $\gamma$-approximate
max-flow in a graph with $m$ edges.

\end{lem}
\begin{proof}
We extend the proof of \ref{lem:sparse from approx max flow}. Suppose
that we run the cut matching game in the $i$-th round and in the
$j$-subround of matching player. We are going to find a partial matching
between two disjoint set $T_{1}^{(j)},T_{2}^{(j)}\subset V$ each
of size $k^{(j)}=|T_{1}^{(j)}|=|T_{2}^{(j)}|$. Let $C^{(j)}$ be
the $\gamma$-approximate min cut computed by a $\gamma$-approximate
max flow algorithm in a graph $G^{(j)}$, which defined as in the
proof of \ref{lem:sparse from approx max flow}. 
\begin{claim}
\label{claim:small cut then large} If $\delta_{G^{(j)}}(C^{(j)})\le k^{(j)}-s$,
then two sides of the cut are of size at least $s$ i.e. $|C^{(j)}\setminus\{t_{1}\}|,|V\setminus(C^{(j)}\cup\{t_{2}\})|\ge s$.
\begin{proof}
We only give a proof for $|C^{(j)}\setminus\{t_{1}\}|$ because they
are symmetric. Let $n_{t_{1}}$ and $n_{t_{2}}$ be a number of cut
edges $\partial_{G^{(j)}}(C^{(j)})$ incident to $t_{1}$ and $t_{2}$
respectively. As in \ref{claim:small cut then sparse}, $\delta_{G^{(j)}}(C^{(j)})=\frac{1}{\alpha}\delta_{G}(C^{(j)}\setminus\{t_{1}\})+n_{t_{1}}+n_{t_{2}}\ge n_{t_{1}}$
and hence $|C^{(j)}\setminus\{t_{1}\}|\ge k^{(j)}-n_{t_{1}}\ge k^{(j)}-\delta_{G^{(j)}}(C^{(j)})\ge s$. 
\end{proof}
\end{claim}
So, if $\delta_{G^{(j)}}(C^{(j)})\le k^{(j)}-s$, then we return the
smaller side of $(C^{(j)}\setminus\{t_{1}\},V\setminus\{C^{(j)}\cup\{t_{2}\}\})$
which has expansion less than $\alpha$ by \ref{claim:small cut then sparse}
and has size at least $s$ by \ref{claim:small cut then large}. Otherwise,
the size of the max flow is at least $(k^{(j)}-s)/\gamma$. Again,
we can decompose the flow into paths which define a partial matching
$M_{i}^{(j)}$ between $T_{1}^{(j)}$ and $T_{2}^{(j)}$ of size at
least $(k^{(j)}-s)/\gamma$. Let $T_{1}^{(j+1)}=T_{1}^{(j)}\setminus M_{i}^{(j)}$
and $T_{2}^{(j+1)}=T_{2}^{(j)}\setminus M_{i}^{(j)}$. We know that
$(|T_{i}^{(j+1)}|-s)\le(1-\frac{1}{\gamma})(|T_{i}^{(j)}|-s)$. So
there are at most $c=\Theta(\gamma\log n)$ sub-rounds before $|T_{1}^{(j+1)}|,|T_{2}^{(j+1)}|\le s$.
If the game is not terminated until then, let $M''_{i}$ be an \emph{arbitrary
}matching between $T_{1}^{(j+1)}$ and $T_{2}^{(j+1)}$, and return
a perfect matching $M'_{i}=\bigcup_{j}M_{i}^{(j)}\cup M''_{i}$ between
$S_{i}$ and $V\setminus S_{i}$ as an output of matching player in
the $i$-th round.

Suppose that, after $r$ rounds, the matching game did not return
a sparse cut. Let $H'=\bigcup_{i=1}^{r}M'_{i}$ and $H''=\bigcup_{i=1}^{r}M''_{i}$.
We claim that $H=H'\setminus H''$ is our desired output. Observe
that $H=\bigcup_{i=1}^{r}(\bigcup_{j}M_{i}^{(j)})$, so $H$ can be
embedded in $G$ with congestion $O(rc/\alpha)=O(\gamma\log^{3}n/\alpha)$.
It is left to show that every cut in $H$ of size $4rs=\Omega(s\log^{2}n)$
has expansion at least $1/4$:
\begin{claim}
Let $S\subset V$ be a cut of size $4rs\le|S|\le n/2$, then $\phi_{H}(S)\ge1/4$. \end{claim}
\begin{proof}
From the property of cut matching game shown in \cite{KhandekarRV09},
$H'$ is such that $\phi(H')\ge1/2$. So $\delta_{H'}(S)\ge|S|/2$.
Since $\delta_{H}(S)=\delta_{H'}(S)-\delta_{H''}(S)$ and the number
of edges in $H''=\bigcup_{i=1}^{r}M''_{i}$ is at most $rs$, we conclude
that $\delta_{H}(S)=\delta_{H'}(S)-\delta_{H''}(S)\ge|S|/2-rs\ge|S|/2-|S|/4=|S|/4$. 
\end{proof}
The running time analysis is as in \ref{lem:sparse from approx max flow}. 
\end{proof}
By plugging in the $\tilde{O}(m)$-time $(1+\epsilon)$-approximate
max flow algorithm by \cite{Peng14} to \ref{lem:sparse from approx max flow}
and \ref{lem:sparse s-size from approx max flow}, and by assigning
the parameter $c_{size}$ and $c_{exp}$, we get the following corollaries:
\begin{cor}
\label{cor:cut fixed expansion}Given an undirected graph $G=(V,E)$
with $n$ vertices and $m$ edges, an expansion parameter $\alpha>0$,
there is a $\tilde{O}(m)$-time randomized algorithm that with high
probability outputs either 
\begin{itemize}
\item an $\alpha$-sparse cut $S$ where $|S|\le n/2$ or 
\item a graph $H$ embeddable in $G$ with congestion at most $c_{exp}'/\alpha$
where $\phi(H)\ge1/2$, 
\end{itemize}
where $c_{exp}'=\Theta(\log^{3}n)$.
\end{cor}

\begin{cor}
\label{cor:cut fixed expansion and size}Given an undirected graph
$G=(V,E)$ with $n$ vertices and $m$ edges, an expansion parameter
$\alpha>0$, and a size parameter $s\le n/c_{size}'$, there is a
$\tilde{O}(m)$-time randomized algorithm that with high probability
outputs either 
\begin{itemize}
\item an $\alpha$-sparse cut $S$ where $s\le|S|\le n/2$, or 
\item a graph $H$ embeddable in $G$ with congestion at most $c_{exp}'/\alpha$
where every cut in $H$ of size at least $sc_{size}$ has expansion
at least $1/4$. 
\end{itemize}
where $c_{size}'=\Theta(\log^{2}n)$ and $c_{exp}'=\Theta(\log^{3}n)$. 
\end{cor}
Now, we are ready to prove \ref{thm:most balanced sparse cut}. 
\begin{proof}
[Proof of \ref{thm:most balanced sparse cut}]Given a graph $G=(V,E)$
with parameter $\alpha$, Let $c_{size}'$ and $c'_{exp}$ be the
parameters from \ref{cor:cut fixed expansion} and \ref{cor:cut fixed expansion and size}.
Set $c_{size}=2c_{size}'$ and $c_{exp}=4c'_{exp}$ to be the parameter
of \ref{thm:most balanced sparse cut}. Let $A$ and $B$ be the algorithms
from \ref{cor:cut fixed expansion} and \ref{cor:cut fixed expansion and size},
respectively.

We give $G$ and $\alpha$ to $A$. If $A$ returns a graph $H$ embeddable
in $G$ with congestion $c'_{exp}/\alpha$ where $\phi(H)\ge1/2$.
By \ref{prop:embed graph}, we can just report $\phi(G)\ge\phi(H)2\alpha/c'_{exp}\ge\alpha/c_{exp}$
and we are done.

Suppose that $A$ returns an $\alpha$-sparse cut $S_{A}$ where $|S_{A}|\le n/2$.
We binary search for the largest $s$ such that when $B$ is given
$G,\alpha$ and $s$ as inputs, then $B$ returns an $\alpha$-sparse
cut. Suppose that, when $s=s^{*}-1$, $B$ returns an $\alpha$-sparse
cut $S_{B}$ where $s^{*}-1\le|S_{B}|\le n/2$, and when $s=s^{*}$,
$B$ returns a graph $H_{B}$ embeddable in $G$ with congestion $c_{exp}'/\alpha$
where every cut in $H_{B}$ of size at least $c_{size}'s^{*}$ has
expansion at least $1/4$. By \ref{prop:embed graph}, this implies
that every cut $S$ of size at least $c_{size}'s^{*}$ in $G$ has
expansion at least $\alpha/4c'_{exp}=\alpha/c_{exp}$. Therefore,
$\opt(G,\alpha/c_{exp})<c_{size}'s^{*}=c{}_{0}s^{*}/2$. We can return
$S_{B}$ because $|S_{B}|=s^{*}-1\ge s^{*}/2>\opt(G,\alpha/c_{exp})/c_{size}.$
Otherwise, suppose that $B$ never returns a cut.

Then, when $s=1$, $B$ returns a graph $H_{B}$ embeddable in $G$
with congestion $c'_{exp}/\alpha$ where every cut in $H_{B}$ of
size at least $c_{size}'$ has expansion at least $1/4$. By \ref{prop:embed graph},
this implies that every cut $S$ of size at least $c_{size}'$ in
$G$ has expansion at least $\alpha/4c'_{exp}=\alpha/c_{exp}$. Therefore,
$\opt(G,\alpha/c_{exp})<c_{size}'\le c_{size}$. So we can return
$S_{A}$ because $|S_{A}|\ge1>\opt(G,\alpha/c_{exp})/c_{size}$. 
\end{proof}

\subsection{\label{sec:proof local MOS cut}Proof of \ref{thm:local MSO cut}}

First, we slightly change the definition of augmented graphs defined
in \cite{OrecchiaZ14}. 
\begin{defn}
[Augmented Graph]Given an undirected graph $G=(V,E)$, a set $A\subset V$
where $|A|\le|V-A|$, an expansion parameter $\alpha$, and an overlapping
parameter $\sigma\in[\frac{3|A|}{|V-A|},\frac{3}{4}]$, the \emph{augmented
graph }$G_{A}(\alpha,\sigma)$ is the capacitated directed graph with
the following properties:
\begin{itemize}
\item The vertex set $V(G_{A}(\alpha,\epsilon))=V(G)\cup\{s,t\}$. We call
$s$ and $t$ the source and the sink vertices, respectively.
\item The edge set $E(G_{A}(\alpha,\epsilon))$ contains all original undirected\footnote{In directed graph, an edge $e$ is undirected means that flows can
go through $e$ in both direction in contrast to directed edges.} edges $e\in E(G)$ with capacity $1/\alpha$. $E(G_{A}(\alpha,\epsilon))$
also contains directed edges $(s,u)$ for all $u\in A$ with unit
capacity and $(v,t)$ for all $v\in V-A$ with capacity $\epsilon_{\sigma}$
where $\epsilon_{\sigma}\triangleq\frac{1}{3(1/\sigma-1)}\in[\frac{|A|}{|V-A|},1]$.
\end{itemize}
\end{defn}
The following lemma about the properties of $s$-$t$ min cut in augmented
graphs illustrates the usefulness of augmented graphs. The proof below
is different from Lemma 3.2 in \cite{OrecchiaZ14} in two aspects.
First, we argue about the expansion of cuts instead of arguing about
conductance of cuts as in \cite{OrecchiaZ14}, which is a slight different
notion. Second, we show that the $s$-$t$ min cut is a most-balanced
sparse cut, not just a sparse cut as in \cite{OrecchiaZ14}.

Note that $s$-$t$ max-flow-min-cut value of the augmented graph
$G_{A}(\alpha,\sigma)$ is at most $|A|$ because the cut $\{s\}$
has capacity $|A|$. 
\begin{lem}
\label{thm:property local min cut}Given an undirected graph $G=(V,E)$,
a set $A\subset V$ where $4|A|\le|V-A|$, an expansion parameters
$\alpha$ and an overlapping parameter $\sigma\in[\frac{3|A|}{|V-A|},\frac{3}{4}]$,
let $S'$ be the min $s$-$t$ cut of the augmented graph $G_{A}(\alpha,\sigma)$,
and $F^{*}$ is the value of max $s$-$t$ flow in $G_{A}(\alpha,\sigma)$.
Suppose that $F^{*}=|A|-c$ for some value $c\ge0$. Then we have
that following:
\begin{itemize}
\item For any $(A,\sigma)$-overlapping cut $S$ in $G$ where $|S|\ge3c/\sigma$,
then $\phi(S)\ge\frac{\sigma}{3}\alpha$, and
\item if $c>0$, $S'$ is an $\alpha$-sparse cut in $G$ such that $|S'|,|V-S'|\ge c$. 
\end{itemize}
\end{lem}
\begin{proof}
Let $G'=G_{A}(\alpha,\sigma)$. For any cut $S$, let $\delta'(S)$
be the capacity of $S$ in $G'$. Since $F^{*}\le\delta'(S)$ by max-flow
min-cut theorem. We have that 
\begin{equation}
|A|-c=F^{*}\le\delta'(S)=\frac{\delta(S)}{\alpha}+|A-S|+\epsilon_{\sigma}|S-A|=\frac{\delta(S)}{\alpha}+|A|-|A\cap S|+\epsilon_{\sigma}|S-A|.\label{eq:local cut}
\end{equation}
This implies
\begin{eqnarray*}
\frac{\delta(S)}{\alpha} & \ge & |A\cap S|-\epsilon_{\sigma}|S-A|-c.
\end{eqnarray*}
Therefore, for any $(A,\sigma)$-overlapping cut $S$ where $|S|\ge3c/\sigma$,
we have that
\begin{eqnarray*}
\frac{\delta(S)}{|S|} & \ge & \alpha(\frac{|A\cap S|}{|S|}-\frac{\epsilon_{\sigma}(|S|-|A\cap S|)}{|S|}-\frac{c}{|S|})\\
 & \ge & \alpha((1+\epsilon_{\sigma})\sigma-\epsilon_{\sigma}-\sigma/3)\\
 & \ge & \alpha(2\sigma/3-\sigma/3)\\
 & \ge & \frac{\sigma}{3}\alpha
\end{eqnarray*}
where the second inequality is because $S$ is $(A,\sigma)$-overlapping
and $|S|\ge3c/\sigma$, and the third inequality is because $\epsilon_{\sigma}=\frac{1}{3(1/\sigma-1)}$.
This concludes the first part of the lemma.

Let $S'$ be the $s$-$t$ min cut in $G'$. We have that \ref{eq:local cut}
holds with equality for $S'$. That is,
\begin{equation}
|A|-c=\frac{\delta(S')}{\alpha}+|A-S'|+\epsilon_{\sigma}|S'-A|,\label{eq:local min cut}
\end{equation}
which implies 
\[
\frac{\delta(S')}{\alpha}=|A\cap S'|-\epsilon_{\sigma}|S'-A|-c.
\]
We want to prove that $\phi(S')=\frac{\delta(S')}{\min\{|S'|,|V-S'|\}}<\alpha$
when $c>0$. Note again that $\epsilon_{\sigma}=\frac{1}{3(1/\sigma-1)}\in[\frac{|A|}{|V-A|},1]$
as $\sigma\in[\frac{3|A|}{|V-A|},\frac{3}{4}]$. First, we have
\[
\frac{\delta(S')}{|S'|}=\alpha\frac{|A\cap S'|-\epsilon_{\sigma}|S'-A|-c}{|S'|}<\alpha.
\]
Second, we have 

\begin{eqnarray*}
\frac{\delta(S')}{|V-S'|} & = & \alpha\frac{|A\cap S'|-\epsilon_{\sigma}|S'-A|-c}{|V-S'|}\\
 & = & \alpha\frac{(|A|-|A-S'|)-\epsilon_{\sigma}(|V-A|-|(V-S')-A|)-c}{|V-S'|}\\
 & \le & \alpha\frac{\epsilon_{\sigma}|(V-S')-A|-|A-S'|-c}{|V-S'|}<\alpha,
\end{eqnarray*}
where the first inequality is because $\epsilon_{\sigma}\ge\frac{|A|}{|V-A|}$,
and the second inequality is because $\epsilon_{\sigma}\le1$ and
$c>0$. So we conclude that $\phi(S')<\alpha$. Next, we need to show
that $|S'|,|V-S'|\ge c$. From \ref{eq:local min cut}, we have
\[
|A|-c\ge|A-S'|=|A|-|S'\cap A|\implies|S'|\ge|S'\cap A|\ge c.
\]
Similarly, we have
\begin{eqnarray*}
|A|-c & \ge & \epsilon_{\sigma}|S-A|=\epsilon_{\sigma}|V-A|-\epsilon_{\sigma}|V-S-A|\ge|A|-\epsilon_{\sigma}|V-S-A|,
\end{eqnarray*}
which implies
\[
|V-S|\ge|V-S-A|\ge c/\epsilon_{\sigma}\ge c,
\]
because $\epsilon_{\sigma}\le1$. This concludes the second part of
the lemma.
\end{proof}
In \cite{OrecchiaZ14}, they show how to adjust Goldberg and Rao's
binary blocking flow algorithm to \emph{locally} compute $s$-$t$
min cut in augmented graphs, i.e. the running time does not depend
on the size of the whole graph.
\begin{thm}
[Lemma B.3 of \cite{OrecchiaZ14}]Given an undirected graph $G=(V,E)$,
a set $A\subset V$ where $|A|\le4|V-A|$, an expansion parameters
$\alpha$ and an overlapping parameter $\sigma\in[\frac{3|A|}{|V-A|},\frac{3}{4}]$,
there is a deterministic algorithm for computing the $s$-$t$ min
cut $S$ of the augmented graph $G_{A}(\alpha,\sigma)$ in time $O((\frac{vol(A)}{\sigma})^{1.5}\log^{2}(\frac{vol(A)}{\sigma}))$.\footnote{In \cite{OrecchiaZ14}, it is shown also that $vol(S)\le\frac{3}{\sigma}vol(A)$
although we do not need it.}\label{thm:local min cut}
\end{thm}
Given this theorem, we conclude that our most-balanced sparse cut
w.r.t. overlapping cuts algorithm:
\begin{proof}
[Proof of \ref{thm:local MSO cut}]Given a graph $G=(V,E)$, a set
$A\subset V$ where $|A|\le4|V-A|$, an expansion parameter $\alpha$,
and an overlapping parameter $\sigma\in[\frac{3|A|}{|V-A|},\frac{3}{4}]$,
we compute the $s$-$t$ min cut $S$ in the the augmented graph $G_{A}(\alpha,\sigma)$.
Let $c_{size},c_{exp}=3/\sigma$.

If the capacity $\delta'(S)$ of $S$ in $G_{A}(\alpha,\sigma)$ is
$|A|$, then by the first part of \ref{thm:property local min cut},
all $(A,\sigma)$-overlapping cuts $S$ in $G$ have expansion $\phi(S)\ge\frac{\sigma}{3}\alpha$.
Hence, we just report that there is no $(\alpha/c_{exp})$-sparse
$(A,\sigma)$-overlapping cuts in $G$ and we are done.

Otherwise, $\delta'(S)=|A|-c$ for some $c>0$. Then we claim that
the $s$-$t$ min cut $S$ is a $(c_{size},c_{exp})$-approximate
most-balanced $\alpha$-sparse cut w.r.t. overlapping cuts where $c_{size},c_{exp}=3/\sigma$.
Recall $\opt(G,\alpha/c_{exp},A,\sigma)$ that is the size of the
largest $(\alpha/c_{exp})$-sparse $(A,\sigma)$-overlapping cut $S^{*}$
where $|S^{*}|\le|V-S^{*}|$. By the first part of \ref{thm:property local min cut},
we have that $\opt(G,\alpha/c_{exp},A,\sigma)<3c/\sigma=c_{size}\cdot c$.
But the second part of \ref{thm:property local min cut} implies that
the min-cut $S$ is $\alpha$-sparse and $|S|\ge c>\opt(G,\alpha/c_{exp},A,\sigma)/c_{size}$.
Therefore, we return that $S$ is a $(c_{size},c_{exp},\alpha,A,\sigma)$-LBS
cut of $G$.\end{proof}

%% file: warm-up.tex
\section{Very Simple Dynamic $\protect\st$ with $O(m^{0.5+o(1)})$ Update
Time \label{sec:warm-up}}

In this section, let $\beta=2^{O(\log\log n)^{3}}=n^{o(1)}$ be the
factor from \ref{thm:cut recovery tree}. The goal of this section
is to show a simple dynamic $\st$ algorithm with update time $O(m^{0.5+o(1)})$
when the size of update sequence is polynomial (the same assumption
used in \cite{KapronKM13}) and the number of edges of the underlying
graph is \textbf{always bounded} by $m$:
\begin{thm}
\label{thm:dynST warm-up short update}For any fixed constant $c$,
there is a randomized dynamic $\st$ algorithm for any graphs with
$n$ node and at most $m$ edges that works correctly against adaptive
adversaries for the first $n^{c}$ update\emph{ }with high probability
and has preprocessing time $O(m\beta\log n)$ and worst-case update
time $O(\sqrt{m}\beta\log^{2}n)$\footnote{To be more precise, the update time is $O((c+c')\sqrt{m}\beta\log^{2}n)$
if we want the algorithm to works correctly against adaptive adversaries
for the first $n^{c}$ update\emph{ }with probability $1-1/n^{c'}$.}.
\end{thm}
Although there are several faster dynamic $\st$ algorithms, the algorithm
from \ref{thm:dynST warm-up short update} is conceptually very simple
especially when we use a cut recovery tree from \ref{sec:linear sketches}
as a black box.

\subsection{The Algorithm}

Let $G=(V,E)$ be the underlying graph. To preprocess $G$, we do
the following: 1) construct an arbitrary spanning forest $F$ of $G$,
and 2) preprocess $(G,F)$ using an instance $\cE$ of augmented ET
tree and an instance $\cD$ of cut recovery tree with parameter $\sqrt{m}$.
\begin{rem}
\label{rem:cut sketch update-1-1}In the following description of
the algorithm, we only describe how $E$ and $F$ are changed, but
we always give the update to $\cD$ and $\cE$ so that the underlying
graphs and forests of them are $(G,F)$.
\end{rem}
To insert $e=(u,v)$, we add $e$ into $E$ and if $u$ and $v$ are
not connected in $F$ (this can be checked using augmented ET tree),
then we also add $e$ into $F$. To delete $e=(u,v)$, we remove $e$
into $E$ and if $e\in F$ was a tree-edge, then we find a replacement
edge $e'$ for reconnecting the two separated components of $F$ as
described in \ref{alg:replace warm up} and, if $e'$ is returned,
add $e'$ into $F$.

\begin{algorithm}[h]
\begin{enumerate}
\item Using $\cE$, set $T_{u}=\textsf{find\_tree}(u)$ and $T_{v}=\textsf{find\_tree}(v)$.
If $\textsf{tree\_size}(T_{u})<\textsf{tree\_size}(T_{v})$, then
set $S=V(T_{u})$, otherwise $S=V(T_{v})$.
\item Sample $\sigma$ non-tree edges in $G\setminus F$ incident to $S$
using $\cE$. If there is a sampled edge $e'$ from $\partial_{G}(S)$,
then RETURN $e'$.
\item Query $\cutset k(S)$ to $\cD$. If the returned set $\tilde{\partial}\neq\emptyset$,
then RETURN any edge $e'\in\tilde{\partial}$.
\item RETURN ``no replacement edge''
\end{enumerate}
\caption{The algorithm for finding a replacement edge $e'$ if exists, after
deleting an edge $e=(u,v)\in F$. \label{alg:replace warm up}}
\end{algorithm}
This completes the description of our algorithm.

\subsection{Analysis}

Let $G_{i}$ and $F_{i}$ be the underlying graph and the maintained
forest after the $i$-th update. $G_{0}$ and $F_{0}$ are the ones
after preprocessing. Let $\cA$ be our dynamic $\st$ algorithm. Suppose
that we fix an adaptive adversary $f$ and a string $R$ which is
used as a random choices of $\cA$. Let $inv_{i}(f,\cA,R)$ be the
invariant that holds if $F_{i}$ is a spanning forest of $G_{i}$.
For convenience, we write $inv_{<i}(f,\cA,R)=inv_{0}(f,\cA,R)\wedge\dots\wedge inv_{i-1}(f,\cA,R)$.
Our goal is to prove that the invariant holds with high probability:
\begin{lem}
\label{thm:inv hold-1}Let $\cA$ be the dynamic $\st$ algorithm
from \ref{sec:warm-up} and $R$ be the random choices of $\cA$.
For any adaptive adversary $f$, we have $\Pr_{R}[inv_{i}(f,\cA,R)\mid inv_{<i}(f,\cA,R)]\ge1-1/n^{c}$
for $i\ge0$ and $c$ can be made arbitrarily large.\end{lem}
\begin{proof}
When $i=0$, we have $\Pr_{R}[inv_{i}(f,\cA,R)\mid inv_{<i}(f,\cA,R)]=\Pr_{R}[inv_{0}(f,\cA,R)]=1$
because $inv_{<0}(f,\cA,R)$ holds vacuously and $F_{0}$ is a spanning
forest of $G_{0}$ by construction. When $i\ge1$, there are two cases
for the update at time $i$: insertion and deletion. First, when we
insert $e=(u,v)$, we have $F_{i}=F_{i-1}\cup\{e\}$ iff $u$ and
$v$ are not connected in $F_{i-1}$. As $F_{i-1}$ is a spanning
forest of $G_{i-1}$, $F_{i}$ is a spanning forest of $G_{i}$. The
next case is when we delete $e=(u,v)$. On one hand, suppose $\delta_{G}(S)\le k$.
Hence, in Step 3 the queried set returned from $\cutset k(S)$ is
$\tilde{\partial}=\partial_{G}(S)$ and so an edge $e'\in\partial_{G}(S)$,
if exists, is returned. So $F_{i}=F_{i-1}\cup\{e'\}\setminus\{e\}$
and $F_{i}$ a spanning forest of $G_{i}$. On the other hand, suppose
$\delta_{G}(S)>k$. We claim that a random edge $e'\in\partial_{G}(S)$
is returned in Step 2 with high probability, and we are done for the
same reason.

To prove the claim, by the property of the augmented ET-tree $\cE$,
the probability that a sampled edge is from $\partial_{G}(S)$ is
$\frac{\delta_{G}(S)}{vol_{G}(S)}$ for each sampling. Now, the probability
that all sampled $\sigma$ edges are not from $\partial_{G}(S)$ is
at most $(1-\frac{\delta_{G}(S)}{vol_{G}(S)})^{\sigma}\le(1-\frac{\sqrt{m}}{2m})^{\sigma}\le1/n^{c}$
where $c$ can be made arbitrarily large by choosing the constant
in $\sigma=\Theta(\sqrt{m}\log n)$. Note that the first inequality
holds because $vol_{G}(S)\le2m$ as the number of edges of $G$ is
always bounded by $m$.
\end{proof}
Next, we bound the running time.
\begin{lem}
\label{thm:preprocessing time-2}Let $\cA$ be the dynamic $\st$
algorithm from \ref{sec:warm-up}. The preprocessing time of $\cA$
is $O(m\beta\log n)$, and the update time is $O(\sqrt{m}\beta\log^{2}n)$.\end{lem}
\begin{proof}
For the preprocessing time, 1) we can find an arbitrary spanning forest
$F$ on an initial graph $G$ in time $O(m)$, and 2) by \ref{thm:aug ET tree}
and \ref{thm:cut recovery tree}, we can preprocess $(G,F)$ using
augmented ET tree and cut recovery tree in time $O(m\beta\log n)$. 

For the update time, we first analyze the running time inside \ref{alg:replace warm up}.
Step 1 takes time at most $O(\log n)$. Step 2 takes time $O(\sigma\log n)=O(\sqrt{m}\log^{2}n)$
by \ref{thm:aug ET tree}. Step 3 takes time $O(\sqrt{m}\beta\log^{2}n)$
by \ref{thm:cut recovery tree}. So the running time in \ref{alg:replace warm up}
is at most $O(\sqrt{m}\beta\log^{2}n)$. The update time spent outside
\ref{alg:replace warm up} is dominated by the time for updating $\cD$.
This takes time at most $O(\sqrt{m}\beta\log^{2}n)$.
\end{proof}
We conclude that \ref{thm:preprocessing time-2} and \ref{thm:inv hold-1}
together imply \ref{thm:dynST warm-up short update}.